\newif\ifTR
\TRtrue

\documentclass[a4paper,UKenglish,cleveref,autoref,thm-restate]{lipics-v2021}

\ifTR
\pdfoutput=1 
\hideLIPIcs  
\fi

\nolinenumbers

\usepackage[ruled,linesnumbered]{algorithm2e}
\usepackage{mathtools}
\usepackage{paralist}
\usepackage[inline]{enumitem}
\usepackage{stmaryrd}
\usepackage{tikz}
\usepackage{xspace}

\newcommand{\vh}[1]{\textcolor{orange}{\ifmmode \text{[VH: #1]}\else [VH: #1] \fi}}
\newcommand{\ol}[1]{\textcolor{blue}{\ifmmode \text{[OL: #1]}\else [OL: #1] \fi}}
\newcommand{\mh}[1]{\textcolor{red}{\ifmmode \text{[MH: #1]}\else [MH: #1] \fi}}
\newcommand{\lh}[1]{\textcolor{pink}{\ifmmode \text{[LH: #1]}\else [LH: #1] \fi}}

\newcommand{\contains}[0]{\mathrm{Contains}}
\newcommand{\notcontains}{\ensuremath{\neg\contains}\xspace}
\newcommand{\indexof}[0]{\mathrm{indexOf}}

\newcommand{\langconstrof}[1]{\Phi_{#1}}
\newcommand{\langconstr}[0]{\langconstrof \lang}

\newcommand{\vars}{\ensuremath{\mathbb{X}}}
\newcommand{\varsOf}[1]{\ensuremath{\mathit{Vars}(#1)}}

\newcommand{\lang}[0]{\ensuremath{\mathcal{L}}}
\newcommand{\langOf}[1]{\ensuremath{\lang(#1)}}

\newcommand{\nat}[0]{\ensuremath{\mathbb{N}}}

\newcommand{\naturals}[0]{\ensuremath{\mathbb{N}}}

\newcommand{\suf}[0]{\mathrm{Suf}}
\newcommand{\pref}[0]{\mathrm{Pref}}

\newcommand{\prefixesOf}[1]{\ensuremath{\mathrm{Pref}(#1)}}
\newcommand{\suffixesOf}[1]{\ensuremath{\mathrm{Suf}(#1)}}
\newcommand{\factorsOf}[1]{\ensuremath{\mathrm{F}(#1)}}

\newcommand{\concat}[0]{\ensuremath{\circ}}

\newcommand{\defequiv}{\ensuremath{\stackrel{\mathit{def}}{\Leftrightarrow}}}
\newcommand{\defeq}[0]{\mathrel{\triangleq}}
\newcommand{\aut}[0]{\ensuremath{\mathcal{A}}}
\newcommand{\K}[0]{\ensuremath{\mathcal{K}}}

\newcommand{\lcm}[0]{\mathbf{lcm}}
\newcommand{\mygcd}[0]{\mathbf{gcd}}
\newcommand{\gcdof}[2]{\mygcd(#1, #2)}
\newcommand{\lcmof}[2]{\lcm(#1, #2)}

\newcommand{\variant}[0]{\mathbin{\triangleleft}}

\newcommand{\clExpspace}[0]{\textsc{ExpSpace}}
\newcommand{\clExpSpace}[0]{\textsc{ExpSpace}}
\newcommand{\clNExpTime}[0]{\textsc{NExpTime}}

\newcommand{\clPspace}[0]{\textsc{PSpace}}
\newcommand{\clNP}[0]{\textsc{NP}}

\newcommand{\needle}[0]{\mathcal{N}}
\newcommand{\haystack}[0]{\mathcal{H}}

\newcommand{\sep}[0]{\texttt{\#}}

\newcommand{\alphabet}[0]{\Sigma}
\newcommand{\emptyWord}[0]{\epsilon}

\newcommand{\gammaPref}[1]{\Gamma_{\prefixesOf{#1}}}
\newcommand{\gammaSuf}[1]{\Gamma_{\suffixesOf{#1}}}

\newcommand{\move}[1]{\ensuremath{ \stackrel{#1}{\leadsto} }}

\newcommand{\tran}[1]{\ensuremath{ \stackrel{#1}{\rightarrow} }}
\newcommand{\transition}[3]{#1\tran{#2}#3}
\newcommand{\con}[0]{\ensuremath{ \mathrm{con}}}

\newcommand{\primitiveWords}[0]{\ensuremath{ \mathrm{Prim}}}

\newcommand{\flatVars}[0]{\ensuremath{ \vars_{\mathit{Flat}} }}

\newcommand{\stateLabelingFn}[0]{\mathrm{st}}
\newcommand{\stateLabelOf}[1]{\stateLabelingFn(#1)}

\newcommand{\edgeLabelingFn}[0]{\mathcal{W}}
\newcommand{\edgeLabelOf}[1]{\mathcal{W}(#1)}

\newcommand{\letter}[1]{\texttt{#1}}

\newcommand{\pMaxLit}[0]{M_{\mathrm{Lit}}}
\newcommand{\pMaxAut}[0]{M_{Q}}
\newcommand{\pMaxPrim}[0]{M_{\alpha}}

\newcommand{\boundFlat}[0]{\lambda_\mathrm{Flat}}
\newcommand{\boundAut}[0]{\lambda_\mathrm{Q}}
\newcommand{\boundGamma}[0]{ \lambda_{\kappa} }

\newcommand{\glue}[0]{\mathsf{glue}}

\newcommand{\shortPathsSet}[0]{\mathcal{P}}
\newcommand{\boundReachingPaths}[0]{\mathcal{P}_{\ge \boundAut}}

\newcommand{\flatFinLang}[0]{\lang_{\mathrm{Short}}}
\newcommand{\shortPrefixes}[0]{\lang^{\mathrm{Short}}}

\newcommand{\lAlphaFlat}[0]{\lang_{\alpha}}
\newcommand{\lGammaFlat}[0]{\lang_{\Gamma}}
\newcommand{\langflatunderapp}[0]{\lang^{\mathrm{Flat}}}

\newcommand{\pathExit}[0]{\mathrm{Ex}}

\newcommand{\piSigma}[0]{\sigma_{w_{\pi}}}
\newcommand{\chiSigma}[0]{\sigma_{w_{\chi}}}

\newcommand{\ziiinoodler}{\textsc{Z3-Noodler}\xspace}
\newcommand{\cvc}{\textsc{cvc}\xspace}
\newcommand{\cvcv}{\cvc{}5\xspace}

\newcommand{\ziii}{\textsc{Z3}\xspace}

\newcommand{\ostrich}[0]{\textsc{Ostrich}\xspace}

\newcommand{\trau}[0]{\textsc{Trau}\xspace}

\newcommand{\reasonof}[1]{\textcolor{black!50}{\lbag\text{#1}\rbag}}

\usetikzlibrary{calc}
\pgfdeclarelayer{bg}    
\pgfsetlayers{bg,main}  

\newtheorem{fact}[theorem]{Fact}

\title{Negated String Containment is Decidable \ifTR(Technical Report)\fi}

\author{Vojtěch Havlena}
  {Brno University of Technology, Czech Republic}
  {ihavlena@fit.vutbr.cz}
  {https://orcid.org/0000-0003-4375-7954}
  {}

\author{Michal Hečko}
  {Brno University of Technology, Czech Republic}
  {ihecko@fit.vutbr.cz}
  {https://orcid.org/0009-0003-2428-8547}
  {}

\author{Lukáš Holík}
  {Aalborg University, Denmark \and Brno University of Technology, Czech Republic}
  {holik@fit.vutbr.cz}
  {https://orcid.org/0000-0001-6957-1651}
  {}

\author{Ondřej Lengál}
  {Brno University of Technology, Czech Republic}
  {lengal@fit.vutbr.cz}
  {https://orcid.org/0000-0002-3038-5875}
  {}

\authorrunning{V. Havlena, M. Hečko, L. Holík, and O. Lengál}

\ifTR
\else
\relatedversiondetails[cite={techrep}]{Technical Report}{http://arxiv.org/abs/2506.22061}
\fi

\Copyright{Vojtěch Havlena, Michal Hečko, Lukáš Holík, and Ondřej Lengál} 



\EventEditors{Pawe\l{} Gawrychowski, Filip Mazowiecki, and Micha\l{} Skrzypczak}
\EventNoEds{3}
\EventLongTitle{50th International Symposium on Mathematical Foundations of Computer Science (MFCS 2025)}
\EventShortTitle{MFCS 2025}
\EventAcronym{MFCS}
\EventYear{2025}
\EventDate{August 25--29, 2025}
\EventLocation{Warsaw, Poland}
\EventLogo{}
\SeriesVolume{345}
\ArticleNo{54}

\begin{document}

\maketitle

\begin{abstract}
We provide a positive answer to a long-standing open question of the decidability of the not-contains string predicate. 
Not-contains is practically relevant, for instance in symbolic execution of string manipulating programs.
Particularly, we show that the predicate $\neg\mathrm{Contains}(x_1 \ldots
  x_n, y_1 \ldots y_m)$, where $x_1 \ldots x_n$ and $y_1 \ldots y_m$ are sequences of string variables constrained by regular languages, is decidable. 
Decidability of a~not-contains predicate combined with chain-free word equations and regular membership constraints follows.
\end{abstract}

\vspace{-0.0mm}
\section{Introduction}
\vspace{-0.0mm}

String constraints have been recently intensely studied in relation to their applications in analysis of string manipulation in programs, e.g., in the analyses of security of web applications or cloud resource access policies~\cite{Rungta22}. 
Apart from a plethora of practical solvers, e.g., 
\cvcv~\cite{cvc4_string14,tinelli-fmsd16,tinelli-hotsos16,tinelli-frocos16,cvc417,cvc422,cvc420},
\ziii~\cite{z3, BTV09,LuSJDM024},
\ostrich~\cite{AnthonyTowards2016,AnthonyReplaceAll2018,AnthonyComplex2019,AnthonyRegex2022,AnthonyInteger2020},
\ziiinoodler~\cite{ChenCHHLS24,chen2023solving,BlahoudekCCHHLS23}, 
\trau~\cite{ChainFree, Trau, Flatten}
\textsc{Z3Str/2/3/4/3RE} \cite{Z3str4,Z3str3RE,BerzishDGKMMN23},
\textsc{Woorpje}~\cite{DayEKMNP19}, 
and
\textsc{nfa2sat}~\cite{nfa2sat23},
the theoretical landscape of string constraints has been intensely studied too. 
The seminal work of Makanin~\cite{makanin77}, establishing decidability of word equations,
was followed by the work of Plandowski~\cite{plandowski99} (and later Je\.z's work on recompression) that placed the problem in \clPspace.
A number of relatively recent works study extensions of string constraints with constraints over string lengths, transducer-defined relational constraints, string-integer conversions, extensions of regular properties, replace-all, etc.
As the extended string constraints are in general undecidable, these works
focus on finding practically relevant decidable fragments such as the
straight-line~\cite{ChenHLRW19,AnthonyTowards2016,AnthonyReplaceAll2018,AnthonyInteger2020,AnthonyRegex2022}
and chain-free~\cite{ChainFree,chen2023solving} fragments, quadratic
equations~\cite{nielsen1917}, and others (e.g., \cite{aiswarya22,day23}).

The most essential constraints, from the practical perspective, are considered to be word equations, regular membership constraints, length constraints, and also  $\notcontains$, as argued, e.g., in~\cite{Kudzu}, and as can also be seen in benchmarks, for instance, in~\cite{trauc,parosh-notsubstring}. 
While the three former types of constraints are intensely studied, $\notcontains$ was studied only little. 
Yet, it is important as well as theoretically interesting:
besides the occurrence in existing benchmarks, its importance follows also from its ability to capture other highly practical types of constraints. 
E.g., the $\indexof(x, y)$ function should return the position of the first
occurrence of~$y$ in~$x$.
It can be converted to the word equation $x=p.y.s$ after which the returned value equals $|p|$.
To ensure that~$y$ is indeed the first occurrence in~$x$, there should
be no occurrence of~$y$ in $p.y'$ where~$y'$ is the prefix of~$y$ without the last
symbol, i.e., $y'.z = y$ for $z \in \Sigma$. This can be expressed as
$\notcontains(y, p.y')$ (e.g., \ziii solves $\indexof$ in this
way~\cite{z3}).

As mentioned above,
the problem is also interesting from the theoretical perspective. 
Although the positive version, $\contains$, can be easily encoded using word equations,
the negation is difficult. Its precise conversion to word equations would require universal quantification,
which is undecidable for word equations in general \cite{ganesh18}. 
The most systematic attempts at solving $\notcontains$ have been made in
\cite{parosh-notsubstring,ChenHHHL25}.
In~\cite{parosh-notsubstring}, the authors extend the flattening
underapproximating framework behind the solver \trau~\cite{Trau, Flatten} and
give a~precise solution for $\notcontains$ if all involved string variables are
constrained by flat languages (a~flat language here stands for a~finite union
of concatenations of iterations of words) and, moreover, if no string variable
appears multiple times, thus avoiding most of the difficulty of the problem. 
Our recent work \cite{ChenHHHL25}, on top of which we build here, proceeds in a
similar direction and removes the restriction of~\cite{parosh-notsubstring} on
multiple occurrences of variables, but still requires all
languages to be flat, which is a quite severe restriction.
Practical heuristics used in solvers generally solve only easy cases and quickly fail on more complex ones, cf.\ \cite{ChenHHHL25}, and do not give any guarantees. 
E.g., \cvcv{} translates $\notcontains$ into a universally
quantified disequality~\cite{ReynoldsHighlevel19}, which is in turn handled by
\cvcv's incomplete quantifier instantiation~\cite{niemetz21}.

\enlargethispage{3mm}

In this paper, we show decidability of a much more general kind of $\notcontains$ than \cite{ChenHHHL25,parosh-notsubstring}, namely of the form $\notcontains(\needle,
\haystack) \land \langconstr$
where $\needle$eedle and $\haystack$aystack
are string terms (sequences of symbols and variables) and
$\langconstr$ constrains variables by \emph{any regular language}.
The constraint is satisfied by an assignment to string variables
respecting~$\langconstr$ under which $\needle$ is not a factor (i.e.,
a~continuous
subword) of~$\haystack$ (i.e., if $\needle$eedle cannot be found in
$\haystack$aystack).

Our solution of the problem leads relatively deep into word combinatorics and automata theory. 
We rely on the result in~\cite{ChenHHHL25}
giving a decision procedure for a flat-language version of the problem.
The work~\cite{ChenHHHL25} uses an automata-based construction inspired by deciding
functional equivalence of streaming string
transducers~\cite{Alur11}.
Using a variation on automata Parikh images, it transforms the problem into an equisatisfiable Presburger arithmetic formula (which is decidable).
The general case with variables restricted by arbitrary regular languages, the subject of this paper, is solved by a~reduction to this flat-language fragment. 
The core idea of our proof is that we can always find fresh primitive words in non-flat languages that can be repeated an arbitrary number of times. The result of such a~repetition
is a word that can share with other variables only subwords of a bounded size, assuming all words assigned to variables are sufficiently long.
The reduction technically requires a dive into combinatorics on words and results on primitive words \cite{Lothaire02,makanin77,lyndon1962,Lothaire1997}, 
which are closely related to flat languages.
Our techniques shares traits with the work of Karhum{\"{a}}ki et al.~\cite{karhumaki00},
which constructs long primitive words to show that disjunctions of word equations
can be encoded into a~single equation.
First, for variables with non-flat languages occurring on both sides of the constraint, we show that we can replace each of them with a single fresh symbol.
This is because non-flat languages allow us to choose a sufficiently complex word for the variable $x$ that can be matched only with the value of $x$ on the other side ($\needle$~is the other side of~$\haystack$ and \emph{vice versa}).
For variables with non-flat languages that appear only in~$\haystack$, 
we show that after enumerating all possible assignments for them up to a~certain bound,
their languages can be underapproximated by flat languages while preserving satisfiability. 

\vspace{-2.0mm}
\section{Preliminaries}
\vspace{-1.0mm}

\subparagraph{Numbers.}
We use $\nat$ for natural numbers (including zero).
For $m,n \in \nat$,
their \emph{greatest common divisor} is denoted as~$\gcdof m n$ and
their \emph{least common multiple} is denoted as~$\lcmof m n$.

\tikzstyle{wordBNode} = [draw, scale=0.75, inner sep=1mm, text height=3mm, minimum height=7mm, baseline]

\subparagraph{Words.}
An \emph{alphabet}~$\Sigma$ is a finite non-empty set of \emph{symbols}.
Let~$\Sigma$ be fixed for the rest of the paper. 
A~(finite) word~$w$ over~$\Sigma$ is a~sequence of symbols $w = a_1 \ldots a_n$
from~$\Sigma$, where~$n$ is the \emph{length} of~$w$, denoted as~$|w|$.
The \emph{empty word} of the length~0 is denoted by~$\epsilon$ and a~concatenation
of two words~$u$ and~$v$ is denoted as $u \concat v$ 
(or shortly $uv$).
An \emph{iteration of a word} $w$ is defined as $w^0\defeq\epsilon$ and $w^{i+1} \defeq w^i \concat w$ for $i\geq 0$.
The set of all words over~$\Sigma$ is denoted as~$\Sigma^*$.
A~\emph{primitive word} cannot be written as $v^i$ for any $v$ and $i>1$, and we will use Greek letters $\alpha,\beta,\gamma,\ldots$ 
from the beginning of the alphabet
to denote primitive words.  
We denote the set of all primitive words $\primitiveWords$. 
A word $u$ is a \emph{factor} (i.e., a~continuous subword) of every word $v u
v'$.
%
%
%
Given two words $p u s$ and $p' u' s'$, we say that the factors $u$ and $u'$ have an overlap of size
$k \in \naturals$ if $\bigl| \{|p|+1,\ldots,|p|+|u|\} \cap \{|p'|+1,\ldots,|p'|+|u'|\} \bigr| = k$.
The overlap of $u$ and $u'$ in the words $pus$ and $p'u's'$ contains
a \emph{conflict} if there is a position $i$ with $|p| \le i < |pu|$ and $|p'| \le i < |p'u'|$
such that the words $pus$ and $p'u's'$ contain a different letter at position $i$.

\subparagraph{Languages.}
A~\emph{language}~$\lang$ over~$\Sigma$ is a~subset of~$\Sigma^*$.
We will sometimes abuse notation and, given a~word~$w \in \Sigma^*$, use~$w$ to
also denote the language~$\{w\}$. 
%
For two languages~$\lang_1$ and~$\lang_2$, we use $\lang_1 \concat
\lang_2$ (or just $\lang_1 \lang_2$) for their concatenation $\{uv \mid
u\in\lang_1, v \in \lang_2\}$.
A~\emph{bounded iteration of a~language} $\lang$ is defined as
$\lang^0 \defeq \{\epsilon\}$ and 
$\lang^{i+1} \defeq \lang^i \concat \lang$ for $i\geq 0$.
The \emph{(unbounded) iteration} is $\lang^* \defeq \bigcup_{i\geq 0} \lang^i$.
For a~word~$w$ we use $\prefixesOf{w}$ ($\suffixesOf{w}$) to denote the set of
prefixes (suffixes) of~$w$ and~$\factorsOf{w}$ to denote the set of all factors
of $w$. We lift the definitions to languages as usual.
A~language $\lang \subseteq \Sigma^*$ is \emph{flat} iff it can be
expressed as a~finite union
\begin{equation} \label{eq:defFlat}
\lang = \bigcup_{i=1}^N w_{i,1} \concat w_{i,2}^* \concat w_{i,3} \concat w_{i,4}^* \concat
w_{i,5} \concat \cdots \concat w_{i,\ell_i-1}^* \concat w_{i,\ell_i}
\end{equation}
where every $w_{i,j}$ s.t. $1\leq i\leq n,1\leq j \leq \ell_i$ is a word over $\Sigma$,
else it is \emph{non-flat}. 
Flatness of~$\lang$ can be  characterised by the absence of the so-called
``butterfly loops'':

\begin{fact}\label{lem:nonflatCharacterization}
A regular language $\lang \subseteq \Sigma^*$ is non-flat iff $p \{u, v\}^* s \subseteq \lang$ for some $p, s, u, v \in \Sigma^*$
with $u, v \not \in w^*$ for any word $w \in \alphabet^*$.
\end{fact}

\subparagraph{Automata.}
A \emph{(nondeterministic finite) automaton (NFA)} over $\Sigma$ is a tuple $\aut = (Q,\Delta,I,F)$ where $Q$ is a set of \emph{states},
$\Delta$ is a set of \emph{transitions} of the form $\transition q a r$ with $q,r\in
Q$ and $a\in\Sigma$, $I\subseteq Q$ is the set of \emph{initial states}, and $F\subseteq Q$
is the set of \emph{final states}. 
A run of $\aut$ over a word $w = a_1 \ldots a_n$ from state $q_0$ to state~$q_n$ 
is a sequence of
 transitions
 $\transition {q_0} {a_1} {q_1}$,
 $\transition {q_1} {a_2} {q_2}$, $\ldots$,
 $\transition {q_{n-1}} {a_n} {q_n}$ from $\Delta$. The empty sequence is a~run with $q_0 = q_n$ over $\epsilon$. 
We denote by $q_0 \move{w}_\aut q_n$ that $\aut$ has such a~run, from where we drop the subscript $\aut$ if it is clear from the context. 
The run is \emph{accepting} if $q_0 \in I$ and $q_n\in F$, and the
\emph{language} of~$\aut$ is $\langOf \aut \defeq \{w \in \Sigma^* \mid q\move{w} r, q\in I,
r\in F\}$.
Languages accepted by NFAs are called \emph{regular}.
$\aut$~is a~\emph{deterministic finite automaton} (DFA) if $|I| = 1$ and for
every symbol $a\in \Sigma$ and every pair of transitions $\transition{q_1} a
{r_1}$ and $\transition{q_2} a {r_2}$ in~$\Delta$ it holds that if $q_1 =q_2$
then $r_1 = r_2$.

\subparagraph{The \notcontains constraint.}
Let~$\vars$ be a~set of \emph{(string) variables}. 
A~\emph{term} is a word $t\in(\vars\cup\Sigma)^*$ over variables and
symbols.
A~\notcontains constraint is a formula $\varphi \defeq \notcontains(\needle,
\haystack) \land \langconstr$,
where $\needle$ and $\haystack$ (for $\needle$eedle and $\haystack$aystack;
$\varphi$ holds if we cannot find $\needle$ within $\haystack$)
are terms
and $\langconstr \defeq \bigwedge_{x\in\vars} x\in \lang(x)$
associates every variable~$x$ with a~regular language $\lang_x$.
%
An \emph{assignment} is a~function $\sigma\colon \vars \to \Sigma^*$, i.e., it
assigns strings to variables.
We use $\sigma \variant \{x_1 \mapsto w_1, \ldots, x_n \mapsto w_n\}$ to denote
the assignment obtained from~$\sigma$ by substituting the values of
variables~$x_1, \ldots, x_n$ to $w_1, \ldots, w_n$ respectively.
We lift~$\sigma$ to terms so that for $a \in \Sigma$, we let $\sigma(a) \defeq a$,
and for terms $t,t'$, we let $\sigma(t\concat t') \defeq \sigma(t)\concat\sigma(t')$. 
We then say that $\sigma$ satisfies~$\varphi$, written $\sigma \models
\varphi$, if $\sigma(x) \in \lang_x$ for every $x\in\vars$ and
$\sigma(\haystack)$ cannot be written as $u \concat \sigma(\needle) \concat v$ for any $u,v\in\Sigma^*$,
i.e., $\sigma(\needle)$ is not a factor of $\sigma(\haystack)$. 
%
We call a~variable~$z$ \emph{two-sided} if it occurs in both~$\needle$ and~$\haystack$.
Moreover, we use $\flatVars$ to denote the set of variables $x$ occurring in
$\varphi$ s.t.\ $\lang_x$ is a flat language.

Given a~term~$t$, a~variable~$x\in \vars$, and a~term~$t_s$, we use $t[x/t_s]$
to denote the term obtained by substituting every occurrence of the
variable~$x$ in~$t$ by the term~$t_s$.
Moreover, we use $\varsOf t$ to denote the set of variables with at least one
occurrence in the term~$t$.

\begin{theorem}[{\cite[{Theorem 7.5}]{ChenHHHL25}}]
Satisfiability of the $\notcontains$ constraint is $\clNP$-hard.
\end{theorem}

\vspace{-0.0mm}
\subsection{Normalization}\label{sec:normalization}
\vspace{-0.0mm}

A~variable~$z$ is
\emph{flat (non-flat)} if the language~$\lang_z$
associated with~$z$ is flat (non-flat), respectively, and \emph{finite} if its
corresponding language is finite.
Moreover, a~variable is called \emph{decomposed} if its language 
can be represented by a DFA having a~single initial, single final state, 
and containing exactly one nontrivial \emph{maximal strongly connected component} (SCC)
and no other SCCs.
We say that~$\varphi$ is \emph{normalized} if it contains an occurrence of at least one variable,
does not contain any finite
variable, and all of its variables are decomposed.
Any $\notcontains$ constraint can be transformed into a~disjunction of
normalized constraints, as shown by the following lemma.

\begin{restatable}{lemma}{lemNormalization}\label{lem:normalization}
Let $\varphi \defeq \notcontains(\needle, \haystack) \land \langconstr$.
Then $\varphi$ can be transformed to an equisatisfiable
disjunction $\bigvee_{1 \leq i \leq n} \notcontains(\needle_i, \haystack_i)
\land \langconstrof{\lang_i}$ of normalized constraints or the
formula~$\mathit{true}$.
\end{restatable}

Due to the previous lemma, in the rest of the paper we will focus on solving
a~single normalized $\notcontains$ constraint.

\medskip
In the paper, we will also make use of the following result showing decidability of
$\notcontains$ with only flat variables.
\begin{lemma}[\cite{ChenHHHL25}] \label{lemma:liaDecidability}
    Satisfiability of $\notcontains(\needle, \haystack) \land \langconstr$
    where~$\lang_x$ is flat for any $x \in \vars$ is decidable in \clNExpTime.
\end{lemma}
\begin{proof}[Proof sketch]
    We can reduce~$\varphi$ into an equisatisfiable Presburger arithmetic formula
    $\psi$ based on Parikh images of runs of the NFAs for the variables in~$\varphi$.
    Decidability of~$\varphi$ follows from decidability of Presburger arithmetic.
    See~\cite{ChenHHHL25} for details.
\end{proof}

The crucial fact that \cref{lemma:liaDecidability} depends on is that there is
a~one-to-one mapping between runs in NFAs of flat languages and their Parikh
images; this mapping fundamentally breaks for non-flat languages so one cannot
directly extend this technique to the non-flat case.

\vspace{-0.0mm}
\subsection{Lemmas in Our Toolbox} \label{sec:label}
\vspace{-0.0mm}

We introduce fundamental lemmas from the area of combinatorics on words that will be used
throughout the rest of the paper. The following lemma will be useful to guarantee
the existence of conflicts (i.e., non-matching positions) in sufficiently large overlaps of two words $\alpha^M$ and $\beta^N$ 
for some primitive words $\alpha, \beta \in \alphabet^*$ and large constants $M, N \in \naturals$.
Intuitively, we will control the choice of $\alpha$ and $\beta$,
and, thus, guarantee that $\alpha$ and $\beta$ cannot be powers of the same word,
essentially applying the contraposition of the following lemma.

\begin{lemma} \label{lemma:primRotations}
    Let $\alpha \in \alphabet^*$ be a primitive word, and let $p$ and $s$ be two words
    such that $\alpha = p s$. Then the word $s p$ is primitive.
\end{lemma}
\begin{proof}
    Assume that $\beta^k = s p$ for some $k \ge 2$. Then we have $s = \beta^l u$ and $p = v \beta^m$
    for $u \in \prefixesOf{\beta}$, $v \in \suffixesOf{\beta}$ such that $\beta = u v$
    and $l + m + 1 = k$. Thus, we have
    \begin{equation}
        \alpha = v \beta^m \beta^l u = v (uv)^m (uv)^l u = (vu)^{l+m+1}
    \end{equation}
    and so the word $\alpha$ is not primitive, a contradiction.
\end{proof}

\begin{lemma}[{\cite[Proposition~1.2.1 (Fine and Wilf)]{Lothaire02}}] \label{lemma:fineAndWilf}
    Let $x$ and $y$ be two words. If the words $x^k$ and $y^l$, for any $k,l \in \nat$ share 
    a~common prefix of the length at least $|x| + |y| - \gcdof{|x|}{|y|}$,
    then $x$ and $y$ are powers of the same word.
\end{lemma}
Using \cref{lemma:primRotations,lemma:fineAndWilf}, we provide the following corollary that shows
existence of conflicts between arbitrary overlaps of repetitions of primitive words
of a~sufficient size.

\begin{restatable}{corollary}{corInfixFineWilf}\label{cor:infixFineWilf}
Let $u = \alpha^M$ and $v = \beta^N$ be two words where $\alpha, \beta \in \primitiveWords$,
with $|\alpha| \neq |\beta|$ and $M, N \in \naturals$. Then any overlap between $u$ and $v$
of the size at least $|\alpha|+|\beta|-\gcdof{|\alpha|}{|\beta|}$ contains a conflict.
\end{restatable}

A natural approach to showing that an assignment $\sigma$ satisfies $\varphi$ is to
show that $\sigma(\haystack)$ cannot be written as $\sigma(\haystack) = p\concat \sigma(\needle) \concat s$
for any choice of words $p$ and $s$. Therefore, one would have to consider all prefixes $p$,
infixes $u$, and corresponding suffixes $s$ with $|u| = |\sigma(\needle)|$ and show that
$\sigma(\haystack) = p u s$ implies $u \neq \sigma(\needle)$. Note that the choice of the
prefix $p \in \prefixesOf{\sigma(\haystack)}$ uniquely determines $u$ and $s$, and, therefore,
we can only refer to different prefixes when showing $\sigma \models \varphi$. The following
lemma reduces the number of prefixes we have to consider if we have information about primitive
words that are factors of $\sigma(\needle)$ and $\sigma(\haystack)$.

\begin{lemma}[{\cite[Proposition~12.1.3]{Lothaire02}}] \label{lemma:primitiveAlignment}
    Let $\alpha \in \Sigma^*$ be a primitive word, and let $\alpha^2 = x \alpha y$ for some
    words $x, y \in \Sigma^*$. Then either $x = \epsilon$ or $y = \epsilon$, but not both.
\end{lemma}

We will use the next lemma as a recipe for constructing words $w_z \in \lang_z$ for
non-flat $\lang_z$ such that $w_z$ has as a factor a primitive word that 
is sufficiently long for our proofs.

\begin{lemma}[{\cite{lyndon1962}}] \label{lemma:primitiveDensity}
    Let $x^K = y^L z^M$ such that $x, y$, and $z$ are string variables
    and $K, L$ and $M$ are integers such that $K, L, M \ge 2$. Then any solution
    of the equation has the form $x = \alpha^k$, $y = \alpha^l$, and $z = \alpha^m$
    for some word $\alpha$ and numbers $k, l, m \in \naturals$.
\end{lemma}

We provide the following corollary
to give insight into how we use \cref{lemma:primitiveDensity} to construct factors that are primitive words
of a suitable length.
\begin{corollary} \label{cor:primWordConstruction}
    Given two words $u$ and $v$ such that for any word $w$ it holds that $u, v \not \in w^*$,
    we have that any word $\alpha = u^L v^M$ for $L, M \ge 2$ is primitive.
\end{corollary}

\begin{proof}
    By contradiction. Assume that $\alpha$ is not primitive, i.e., $\alpha = t^K = u^L v^M$
    for some $t$ and $K, L, M \ge 2$. Applying \cref{lemma:primitiveDensity}, we see that $u = w^l$
    and $z = w^m$ for some $w$, which contradicts the assumptions of the corrolary.
\end{proof}

\newcommand{\SCCs}[0]{\mathcal{S}}

\vspace{-0.0mm}
\subsection{Easy Fragments}\label{sec:easy}
\vspace{-2.0mm}

Before we establish our main result giving the decidability of the hardest
fragment of $\notcontains$, we first describe what we consider \emph{easy
fragments} and how to deal with them.
We assume a~normalized $\notcontains(\needle, \haystack) \land \langconstr$
constraint.
\begin{enumerate}
  \item  
    \emph{The formula is solvable by length abstraction.}
    This fragment contains formulae that can be solved easily by making the
    $\needle$eedle longer than the $\haystack$aystack.
    Suppose $\needle = t_1 \ldots t_m$ and $\haystack = s_1 \ldots s_n$ where
    every $t_i$ and $s_j$ is either a string variable $x\in\vars$ or a~symbol
    $a\in\Sigma$.
    We can then create a~Presburger arithmetic formula~$\varphi_\ell$ over
    \emph{length variables} $\{x_\ell \mid x \in \vars\}$ such that
    $\varphi_\ell\defeq \sum_{1\leq i\leq m} \ell_i > \sum_{1 \leq j \leq n}
    \ell_j \land \Psi$.
    In the formula, $\ell_i$~and~$\ell_j$ are either~1 (if~$t_i, s_j \in \Sigma$)
    or the length variable~$x_\ell$ (if $t_i, s_j = x$), and~$\Psi$ is
    a~formula constraining the possible values for the length variables
    (obtained, e.g., using the Parikh images of the variables' languages).
    If~$\varphi_\ell$ is satisfiable, so is the original $\notcontains$.
    %
  \item  \emph{All variables are flat.}
    In this case, we can use \cref{lemma:liaDecidability}.
\end{enumerate}

\vspace{-4.0mm}
\section{Overview}\label{sec:label}
\vspace{-2.0mm}

We now move to our main result: deciding a~\emph{hard} instance of
$\varphi \defeq \notcontains(\needle, \haystack) \land \langconstr$.
We can classify normalized $\notcontains$ constraints (cf.\ \cref{sec:normalization}) that do not fall in the
fragments of \cref{sec:easy} based on the occurrences of non-flat variables as
follows:
\begin{enumerate}
  \item  constraints where a~non-flat variable~$x$ occurs both in~$\needle$ and~$\haystack$ and
    \label{item:twosidedprob}
  \item  constraints where all (and at least one) non-flat variables occur only in~$\haystack$.
    \label{item:haystackOnly}
\end{enumerate}
Note that the above not included cases of
\begin{inparaenum}[(a)]
  \item  all variables being flat and
  \item  a~non-flat variable being only in~$\needle$
\end{inparaenum}
are covered in \cref{sec:easy}. In particular, if there is a~variable~$x$
that only occurs in $\needle$, then $\lang_x$ is infinite due to our normalization.
Therefore, such a~constraint can be solved by making $\needle$ longer than $\haystack$.

We distinguish the classes~(\ref{item:twosidedprob}) and~(\ref{item:haystackOnly}) above
since for~(\ref{item:twosidedprob}), the string substituted for some occurrence of~$x$ in
$\sigma(\haystack)$ may overlap with the string for an occurrence of $x$ in
$\sigma(\needle)$.
We deal with the class~(\ref{item:twosidedprob}) by substituting two-sided non-flat
variables~$x$ with fresh symbols.
In \cref{sec:twoSided}, we show that if there is a~model~$\sigma$ of the resulting
$\notcontains$, we can obtain a~model~$\sigma'$ of the original
constraint~$\varphi$ from~$\sigma$ by assigning~$\sigma'(x)$ to a~long-enough
word that ensures a~mismatch for every overlap of $\sigma'(x)$ in
$\sigma'(\haystack)$ and $\sigma'(x)$ in $\sigma'(\needle)$.
By doing this, we reduce~(\ref{item:twosidedprob}) to
either~(\ref{item:haystackOnly}) or $\notcontains$ over flat variables
(potentially with no variables at all).

For deciding the class~(\ref{item:haystackOnly}), given in detail in
\cref{sec:singleSided}, we construct an equisatisfiable formula that uses flat
underapproximations of languages associated with the remaining (as some might
have been removed at step~(\ref{item:twosidedprob})) non-flat variables present
in~$\haystack$.
Our result is based on the observation that long words in a~non-flat language
may have a~richer structure compared to long words one can construct using flat languages.
Therefore, it is unlikely that a~flat variable~$z$ should have a~large
conflict-free overlap with a non-flat variable $x$ in an assignment that
assigns these two variables sufficiently long words.
In particular, we prove that the original language of~$x$ can be
underapproximated by a~flat language while preserving equisatisfiability.
After this step, the resulting constraint can be decided using
\cref{lemma:liaDecidability}.

\vspace{-2.0mm}
\section{Removing Two-Sided Non-Flat Variables} \label{sec:twoSided}
\vspace{-1.0mm}

In this section, we will show how to transform a~normalized constraint
$\notcontains(\needle, \haystack) \land \langconstr$ with an occurrence of
a~two-sided non-flat variable~$x$ into a~constraint without occurrences of the
variable~$x$.
The resulting constraint after removing all two-sided non-flat variables can
then be solved either by reduction to Presburger arithmetic
(\cref{lemma:liaDecidability}; if no non-flat variables remain in the
constraint) or by the procedure in \cref{sec:singleSided} (if there are still
non-flat variables left in~$\haystack$).
The main result of this section is the following theorem.

\begin{theorem} \label{lemma:twoSidedRemoval}
    Let $\varphi \defeq \notcontains(\needle, \haystack) \land \langconstr$ be
    a~constraint over the
    alphabet~$\alphabet$ and let $z\in \varsOf{\needle} \cap
    \varsOf{\haystack}$ be a~non-flat variable.
    Then the formula
    $\varphi_\sep \defeq \notcontains(\needle[z/\sep], \haystack[z/\sep]) \land
    \langconstr$ with $\sep \notin \Sigma \cup \vars$ is equisatisfiable to~$\varphi$.
\end{theorem}
\vspace{-1mm}

The proof of the theorem is given below.
It is based on the observation that assigning long words to two-sided variables
necessarily causes some occurrences of the same variable to overlap.
Since these variables are non-flat, we can construct long words with a~rich
internal structure that will guarantee that any sufficiently long overlap
necessarily contains a~conflict.

\newcommand{\rmone}[0]{\mathrm{I}}
\newcommand{\rmtwo}[0]{\mathrm{II}}
\newcommand{\rmthree}[0]{\mathrm{III}}
\newcommand{\rmfour}[0]{\mathrm{IV}}


Before we give the proof, let us formally introduce the concept of words that
do not allow conflict-free overlaps of two occurrences of the same word larger
than a certain bound.

\vspace{-1mm}
\begin{definition}[$\ell$-aligned word]
Let $w$ be a word and $\ell \in \nat$.
We say that $w$ is \emph{$\ell$-aligned} 
if
for all~$p \in \Sigma^*$ such that $1 \leq |p|\leq |w| - \ell$,
$w$~is not a~prefix of $pw$.
\end{definition}
\vspace{-1mm}

Intuitively, $w$~is $\ell$-aligned if it cannot overlap with itself on
a~prefix/suffix of the length larger than or equal to~$\ell$ (except~$|w|$).
For example, the word $w = abaa$ is $2$-aligned since for no non-empty
word~$p$ of the length at most $|w| - 2 = 2$ it holds that $w$ is a prefix
of~$pw$.
On the other hand, $w = abaa$ is \emph{not} $1$-aligned since for $p =
aba$ of the length~3, it holds that $w$ is a~prefix of $pw = abaabaa$.

\smallskip

\vspace{-3.0mm}
\subsection{Proof of \cref{lemma:twoSidedRemoval}}\label{sec:label}
\vspace{-1.0mm}

If $\varphi$ is satisfiable then so is $\varphi_\sep$.
To see this, take a model $\sigma \models \varphi$ and replace the assignment of $z$ to $\sep$,
producing $\sigma'$. Then, there will be conflicts of $\sep$ and some non-$\sep$ symbol when
checking whether $\sigma'$ is a model of $\varphi_\sep$.
Alternatively, it might be possible to align $\sigma'(\needle)$ with $\sigma'(\haystack)$ in a manner
such that every $\sep$ in $\sigma'(\needle)$ matches some $\sep$ in $\sigma'(\haystack)$.
In such a case, if $\sigma'$ fails to be a model of $\varphi_\sep$ we reach a contradiction
with $\sigma$ being a~model~$\varphi$.

For the other direction, assume that~$\varphi_\sep$ is satisfiable,
which means there is a~model~$\sigma'$ of~$\varphi_\sep$.
Next, we will show how to construct a word $w_z$ s.t.\ $\sigma
= \sigma' \cup \{ z \mapsto w_z \}$ is a~model of~$\varphi$.

    
Focusing on variable~$z$, we can write the two sides of the $\notcontains$
constraint
as $\haystack = \haystack_0 z_{\haystack, 1} \haystack_1 \cdots \haystack_{n-1} z_{\haystack, n} \haystack_{n}$ and 
$\needle = \needle_0 z_{\needle, 1} \needle_1 \cdots \needle_{m-1} z_{\needle, m} \needle_{m}$ where 
$\needle_i, \haystack_j \in (\vars'\cup \Sigma)^*$
for each $i$ and $j$ assuming $\vars' = \vars\setminus \{ z \}$. Moreover, we write the subscript $z_{S, k}$ to distinguish
$k$-th occurrence of $z$ in $S \in \{ \haystack, \needle \}$.
As $\lang_z$ is non-flat, we have that $p \{ u, v \}^* s \subseteq \lang_z$ for
some words $p, u, v$, and $s$ where $u$ and $v$ are not a power of the same
word (\cref{lem:nonflatCharacterization}).

The core of our proof is based on the following observation. Since $\sigma'$ is a model
of $\varphi_\sep$, the word $\sigma'(\needle[z/\sep])$
is not a factor of $\sigma'(\haystack[z/\sep])$. Therefore, given any sufficiently long
word $w_z \in \lang_z$, if the extension~$\sigma = \sigma' \cup \{z \to w_z\}$ 
fails to be a~model, then there must be at least one occurrence of the
word $\sigma(z)$ in~$\sigma(\needle)$ partially overlapping with an occurrence
of the word~$\sigma(z)$ in~$\sigma(\haystack)$, as shown in the picture below.
Thus, if we construct a word $w_z \in \lang_z$
that cannot partially overlap with itself, we get~$\sigma$ that is a model of~$\varphi$.

%

\tikzstyle{boundaryNode} = [draw, text height=3mm, minimum height=7mm]
\tikzstyle{wordScalingFactor} = [scale=0.8]

\vspace{-2mm}
\begin{center}
\begin{tikzpicture}
    \node[boundaryNode, wordScalingFactor, minimum width=15mm, anchor=west, gray] (l_block0) at (0.0, 0)  {$\cdots$};
    \node[boundaryNode, wordScalingFactor, minimum width=12mm, anchor=west] (l_block1) at (l_block0.east)  {$\sigma(\haystack_i)$};
    \node[boundaryNode, wordScalingFactor, minimum width=30mm, anchor=west] (l_block2) at (l_block1.east) {$w_z$};
    \node[boundaryNode, wordScalingFactor, minimum width=8mm,  anchor=west] (l_block3) at (l_block2.east) {$\sigma(\haystack_{i+1})$};
    \node[boundaryNode, wordScalingFactor, minimum width=34mm, anchor=west, gray] (l_block_4) at (l_block3.east)  {$\cdots$};

    \node[boundaryNode, wordScalingFactor, minimum width=24mm, anchor=west, gray] (r_block0) at (0.0, -0.7)  {$\cdots$};
    \node[boundaryNode, wordScalingFactor, minimum width=7mm,  anchor=west]       (r_block1) at (r_block0.east) {$\sigma(\needle_{j})$};
    \node[boundaryNode, wordScalingFactor, minimum width=30mm, anchor=west]       (r_block2) at (r_block1.east) {$w_z$};
    \node[boundaryNode, wordScalingFactor, minimum width=11mm, anchor=west]       (r_block3) at (r_block2.east) {$\sigma(\needle_{j+1})$};
    \node[boundaryNode, wordScalingFactor, minimum width=25mm, anchor=west, gray] (r_block4) at (r_block3.east) {$\cdots$};
\end{tikzpicture}
\end{center}
\vspace{-2mm}

    Let $\alpha = u^2 u^k v^2$ and $\beta = u^2 v^l v^2$ be two words where
    $k = \lcmof {|v|}{|u|} / |u|$ and $l = \lcmof {|v|}{|u|} / |v|$. 
    By invoking \cref{cor:primWordConstruction}, we see
    that both $\alpha$ and $\beta$ are primitive.


    \newcommand{\cinfix}[0]{\gamma}

    Note that we also have $|\alpha| = |\beta|$ (because $|\alpha| = 2|u| + k|u| + 2|v|$,
    $|\beta| = 2|u| + l|v| + 2|v|$ and from the definition of $l$ and $k$ we have 
    $k|u| = l|v|$). We now use these two primitive
    words~$\alpha$ and~$\beta$ to construct~$w_z$. Let $\cinfix \triangleq \alpha^{r} \beta^{r} \concat \alpha^{r} \beta^{r} \concat \alpha^{2r} \beta^{2r}$ and let $w_z \in \lang_z$ be the word
    $w_z \triangleq p \cinfix s$
    where $r \ge 2$ is the smallest number satisfying
    $r |\alpha| > M + |p| + |s|$ with
    $M = \max\{|\sigma(\haystack_i)|, |\sigma(\needle_j)| : 1 \leq i \leq n, 1
    \leq j \leq m\}$.
    We set $\sigma = \sigma' \cup \{z \mapsto w_z \}$.
    Let us now give two lemmas establishing the properties of~$w_z$'s infix~$\cinfix$.
    
    We constructed the infix $\cinfix$ of $w_z$ in a way so that it prevents conflict-free overlaps with itself
    as shown by the following lemma.

    \vspace{-2mm}
    \begin{restatable}{lemma}{lemAlign}\label{cl:align}
        The word $\cinfix$ is $(r+1)|\alpha|$-aligned.
    \end{restatable}
    \vspace{-2mm}

    The full proof of \cref{cl:align} can be found in
    \ifTR{}\cref{sec:proofsTwoSided}\else{}\cite{techrep}\fi, but the core
    of the argument lies in observing that in any overlap of $\cinfix$ with
    itself of size at least $(r+1)|\alpha|$,
    there is a factor $\alpha^2$ having an overlap with $\alpha$ of size $|\alpha|$, or,
    similarly for $\beta$. Therefore, one can apply \cref{lemma:primitiveAlignment}
    and limit overlaps that must be considered.

    The following lemma shows that long overlaps
    between two occurrences of $\cinfix$ are unavoidable
    when $\cinfix$ has a~sufficient length.
    The~$a^i$ in the lemma is used just to position the overlap within
    $\sigma(\haystack)$ and $\sigma(\needle)$.

    \vspace{-1mm}
    \begin{lemma}\label{cl:overlap}
        For each $0\leq i \leq |\sigma(\haystack)| {-} |\sigma(\needle)|$,
        every occurrence of $\cinfix$ in $\sigma(\haystack)$ has an
        overlap with some occurrence of~$\cinfix$ in $a^i \concat \sigma(\needle)$ of size at least $(r+1)|\alpha|$,
        where~$a$ is any symbol in~$\Sigma$.
    \end{lemma}
    \vspace{-3mm}
    \begin{proof}
        Let us assume an arbitrary occurrence of~$\cinfix$ in~$\sigma(\haystack)$.
        Since each $\haystack_i$ and $\haystack_{i+1}$ are separated by $z$
        (the same goes for $\needle_i$ and~$\needle_{i+1}$), it suffices to consider only 
        the pessimistic case, which is when an occurrence of~$\cinfix$
        in~$\needle$ matches the longest $\sigma(\haystack_i)$ with~$p$ and~$s$
        on both sides.
        The situation is schematically depicted below.
        \vspace{-2mm}
        \tikzstyle{wordBNode} = [draw, scale=0.8, inner sep=1mm, text height=3mm, minimum height=6mm, baseline]
        \begin{center}
        \begin{tikzpicture}
            \node[wordBNode, anchor=west, minimum width=40mm] (gamma_1_2) at (0,0) {$\cinfix$};
            \node[wordBNode, anchor=west, minimum width=3mm] (suffix_0)  at  (gamma_1_2.east) {$s$};
            \node[wordBNode, anchor=west, minimum width=12mm] (literal_0) at  (suffix_0.east) {$ \sigma(\haystack_i) $};
            \node[wordBNode, anchor=west, minimum width=5mm] (prefix_0) at  (literal_0.east) {$p$};
            \node[wordBNode, anchor=west, minimum width=40mm] (gamma_2_0) at (prefix_0.east) {$\cinfix$};

            \node[wordBNode, anchor=west, minimum width=5mm]  (r_prefix_0)  at (2.0, -0.8) {$p$};
            \node[wordBNode, anchor=west, minimum width=40mm] (r_gamma_1_0) at (r_prefix_0.east) {$\cinfix$};
            \node[wordBNode, anchor=west, minimum width=3mm]  (r_suffix_0)  at (r_gamma_1_0.east) {$s$};

            \node[gray] at ([xshift=4mm]prefix_0.east) {$o_2$};
            \draw[dashed, gray] ([yshift=3mm]r_gamma_1_0.north east) -- ([yshift=8mm]r_gamma_1_0.north east);
            \node[gray] at ([xshift=-4mm]suffix_0.west) {$o_1$};
            \draw[dashed, gray] ([yshift=3mm]r_gamma_1_0.north west) -- ([yshift=8mm]r_gamma_1_0.north west);
        \end{tikzpicture}
        \end{center}
        \vspace{-3mm}
        In the figure, $o_1$ and $o_2$ denote the overlaps on both sides.
        We show that the size of at least one overlap $o_1$
        and $o_2$ is greater than $(r + 1) |\alpha|$ by expressing the length~$|\cinfix|$ using
        its definition and the schematic above:
        \vspace{-3mm}
        \begin{align}
          &&r(4|\beta| + 4|\alpha|) &= |o_1| + |s| + |\sigma(\haystack_i)| + |p| + |o_2| &&\reasonof{def.\ of $\cinfix$}&&&&&&&&&&&&&&&&&&&&&&&&&&&&& \nonumber\\
          &\Rightarrow\hspace*{-10mm}&7r|\alpha| + r|\alpha| &\leq |o_1| + r|\alpha| + |o_2| &&\reasonof{since $|\alpha| = |\beta|$ and def.\ of $r$} \\
          &\Leftrightarrow\hspace*{-10mm}&7r|\alpha| &\leq |o_1| + |o_2|\nonumber\\[-8mm]\nonumber
        \end{align}
        %
        We have that $|o_1| + |o_2| \geq 7r|\alpha|$, and, thus,
        at least one of~$|o_1|$ and~$|o_2|$ is bigger than $3r|\alpha|$.
        Since $r \ge 2$, we have $3r|\alpha| \geq (r+1)|\alpha|$
        and hence $\cinfix$ has an overlap of the required size.
    \end{proof}

    \vspace{-2mm}

    It remains to show that $\sigma$ is a model of $\varphi$.
    For the sake of contradiction, 
    assume that $\sigma$ is not a model, meaning that $\sigma(\needle)$ is a factor of $\sigma(\haystack)$. From \cref{cl:align,cl:overlap} we have that each occurrence of $\cinfix$ in $\sigma(\needle)$ is 
    perfectly aligned with some $\cinfix$ in $\sigma(\haystack)$, which also means that 
    $w_z$'s are perfectly aligned. Furthermore, we have that $w_z$'s in $\sigma(\needle)$
    are aligned with consecutive $w_z$'s in $\sigma(\haystack)$, i.e., any
    $\sigma(z_{\needle, i})$ is aligned with some $\sigma(z_{\haystack, i + k})$ for some
    $0 \le k \le n - m$. If this were not the case and we had
    $\sigma(z_{\needle, 1})$ overlapping with $\sigma(z_{\haystack, 1 + k})$ while
    there were some $\sigma(z_{\needle, i})$ matching with $\sigma(z_{\haystack, i + k + l})$
    for $l \ge 1$, there would have to be some $\sigma(\needle_{j})$ with $ 1 \le j < i$
    with $|\sigma(\needle_{j})| > |\sigma(z)|$, which is a~contradiction
    with~$w_z$ being longer than any $|\sigma(\needle_j)|$ by construction.
    Hence, for $\sigma'' = \sigma \variant \{ z \mapsto \sep \}$,
    $\sigma''(\needle)$ is also a~factor of $\sigma''(\haystack)$, which is
    a~contradiction to $\sigma'$ being a~model of $\varphi_\sep$. 
    Therefore, \cref{lemma:twoSidedRemoval} holds.
    %
    %


\vspace{-2.0mm}
\section{$\Gamma$-Expansion and Prefix/Suffix Trees} \label{sec:tools}
\vspace{-2.0mm}

At this point, we are left with a~normalized $\notcontains(\needle, \haystack)
\land \langconstr$ constraint where all variables in~$\needle$ are flat.
If all variables in~$\haystack$ are also flat, we can use
\cref{lemma:liaDecidability} and obtain the result.
In the rest of the paper, we will deal with the case when $\haystack$ contains
at least one non-flat variable.
Before we give the proof in \cref{sec:singleSided}, in this section, we
introduce two concepts that will be used later: $\Gamma$-expansion on non-flat
variables and prefix/suffix trees.


\vspace{-2.0mm}
\subsection{$\Gamma$-Expansions on Non-Flat Variables}
\vspace{-2.0mm}


%
Intuitively, non-flat languages have words with a rich internal structure
compared to flat languages. To illustrate, let $x$ be a flat variable with the
language $\lang_x = \alpha^*$ for some word~$\alpha$ and let $z$ be
a~non-flat variable with the language~$\lang_z$.
Furthermore, let $w_x \in \lang_x$ and $w_z \in \lang_z$ be two sufficiently long words.
We inspect the case when $w_x$ and $w_z$ share some long common factor $u$.
Since $w_x \in \alpha^*$, we have $u = s \alpha^k p$ for some $s \in \suffixesOf{\alpha}$,
$p \in \prefixesOf{\alpha}$, and $k \in \naturals$. As $z$ is non-flat,
the run of~$\aut_z$ corresponding to the word~$w_z$ passes through states at
which one can make a~choice of which transition to take next. Since $u$
is long, we have to make a lot of ``right'' choices during the run of $\aut_z$
in order for achieve the common factor $u$, highlighting the difference
between the complexity of $\lang_x$ and $\lang_z$, and suggesting that
there is a way to pick $w_z$ to prevent long overlaps with flat variables
occurring in $\needle$.

Guided by this intuition, we introduce a tool called \emph{$\Gamma_z$-expansion}
of a non-flat variable~$z$. Given a prefix $p \in \prefixesOf{\lang_z}$ and a
suffix $s \in \suffixesOf{\lang_z}$, the $\Gamma_z$-expansion of $(p, s)$
is the word $\Gamma_z(p, s) = p w s \in \lang_z$ for
a particular $w$ such that only a prefix or a suffix of a~bounded length can
have long overlaps with (sufficiently long) words that belong to a flat language.
This tool will play an important
role in our proofs. Loosely speaking, if we start with a model $\sigma$ and we try to find an
alternative model $\sigma' = \sigma \variant \{ z \mapsto \Gamma_z(p, s)\}$,
then the possible reasons why $\sigma'$ fails to be a model are narrowed down
to the choice of $p$ and $s$.

\newcommand{\primBase}[0]{\mathrm{Base}}

In order to define the $\Gamma_z$-expansion, we first need some auxiliary definitions.
First, as a resulting of our normalization, the map
$\primBase \colon \flatVars \rightarrow 2^{\alphabet^*}$
maps any flat variable $x$ to a singleton containing the primitive word~$\alpha$ that
forms the basis of $\lang_x$, i.e., $\primBase(x) \triangleq \{ \alpha \}$
such that $\lang_x = (\alpha^k)^*$ for some $k \in \naturals$. We lift
the definition of $\primBase$ to a set $X$ of flat variables as $\primBase(X) \triangleq \bigcup_{x \in X} \primBase(x)$,
and to a string $s \in (\alphabet \cup \flatVars)^*$ as $\primBase(s) \triangleq \primBase(\varsOf{s} \cap \flatVars)$.

Second, given a variable $z \in \vars$ with $\aut_z = (Q, \alphabet, \Delta, I, F)$, we
define the function $\con_{z}\colon Q \times Q \rightarrow \alphabet^*$ to give the lexicographically
smallest word $\con_{z}(q, s) \triangleq w$ such that $q \move{w}_\aut s$.
Having auxiliary definitions in place, we are ready to define $\Gamma$-expansion
in the context of the formula
$\varphi = \notcontains(\needle, \haystack) \land \langconstr$ with $\needle \in (\flatVars \cup \alphabet)^*$.

\begin{definition}[$\Gamma$-expansion]
    Let $z \in \varsOf{\haystack}$ be a decomposed non-flat variable and $\aut_z = (Q, \alphabet, \Delta, I, F)$ be a~DFA s.t.~$\lang(\aut_z) = \lang_z$.
    Moreover, let $q_{u|v} \in Q$
    be a~state such that $q_{u|v} \move{u} q_{u|v}$ and $q_{u|v} \move{v} q_{u|v}$ with $u, v \not \in w^*$ for any word~$w$.
    Furthermore, let $p \in \prefixesOf{\lang_z}$ be some prefix and $q_p$ be a state such that $q_0 \move{p} q_p$ for some $q_0 \in I$. Similarly, let $s \in \suffixesOf{\lang_z}$
    be a suffix and $q_s$ be a state
    such that $q_s \move{s} q_f$ for some $q_f \in F$.

    Let $\gamma_z \triangleq u^{2+k} v^2$ for a minimal $k \in \naturals$ such that $\gamma_z > |\alpha|$ for any $\alpha \in \primBase(\needle)$.
    Given $K \in \naturals$, we define the \emph{$\Gamma^K_z$-expansion} of $(p, s)$ to be the word
    $\Gamma^K_z(p, s) \triangleq p \concat \con(q_p, q_{u|v}) \concat \gamma_z^K \concat \con(q_{u|v}, q_s) \concat s$.
\end{definition}

Intuitively, $\Gamma_z$-expansion takes a prefix $p$ and finds the~shortest word $\con(q_p, q_{u|v})$
that takes the automaton to the state $q_{u|v}$ in which we have the freedom to read the
words $u$ and $v$ in any suitable sequence. We loop through $q_{u|v}$ in a specific manner so that the resulting
factor $\gamma_z$ is primitive thanks to \cref{cor:primWordConstruction}. The situation with the~suffix
is symmetric. Note that in the following section, we use prefix/suffix variants
of the $\Gamma$-expansion defined as 
$\gammaPref{z}^K \triangleq  p \concat \con(q_p, q_{u|v}) \concat \gamma_z^K$ and
$\gammaSuf{z}^K \triangleq \gamma_z^K \concat \con(q_{u|v}, q_s) \concat s$.

The following lemma shows that $\Gamma_z$-expansion can be seen almost as introducing a~fresh
symbol~$\sep$ for the infix~$w$ connecting $p$ and $s$ into a word $p w s \in \lang_z$.
Intuitively, if we have an~assignment $\sigma$ that assigns sufficiently long words to all flat variables,
then we can find $K \in \naturals$
such that any large overlap between $\sigma(\needle)$ and $\sigma(z) = \Gamma^K_z(p, s)$ contains a~conflict.
We use $\pMaxLit$ to be the length of the
longest literal in~$\varphi$, $\pMaxAut$ to be the number of states of the largest
DFA specifying the language of some variable $x \in \vars$, and
$\pMaxPrim \triangleq \max \{|\alpha| : \alpha \in
\primBase(\needle) \cup \{ \gamma_z \}\}$.


\begin{restatable}{lemma}{lemGammaUseful} \label{lemma:gammaUseful}
    Let $z \in \vars$ be a decomposed non-flat variable,
    $p \in \prefixesOf{\lang_z}$, and $s \in \suffixesOf{\lang_z}$.
    Further, let
    $K \in \nat$ be such that $K|\gamma_z| \ge 4\pMaxPrim + 2\pMaxLit$ and
    let $\sigma$ be an assignment with
    \begin{inparaenum}[(i)]
        \item $|\sigma(x)| \ge 2\pMaxPrim$ for any flat variable $x$ and
        \item $\sigma(z) = \Gamma^K_{z}(p, s)$.
    \end{inparaenum}
    Every overlap between $\sigma(z)$ and $\sigma(\needle)$ of the size at
    least~$\max(|p|, |s|) + \pMaxAut + 2\pMaxLit + 2\pMaxPrim$ contains a conflict.
\end{restatable}

\begin{proof}[Proof sketch.]
    It suffices to observe that if the overlap of size at least $N$ necessarily contains
    an overlap between $\gamma^K_z$ and $\sigma(x) = \alpha^l$ for some flat variable
    $x$ with $\alpha \in \primBase(x)$. The existence of a conflict follows from \cref{cor:infixFineWilf}.
    The full proof can be found in \ifTR{}\cref{sec:proofsTools}\else{}\cite{techrep}\fi.
\end{proof}

\smallskip

Next, we show that $\Gamma_z$-expansion can be used to facilitate modularity in
our proofs, allowing us to search for a suitable prefix~$p$ and a suitable
suffix~$s$ separately. Searching for~$p$ and~$s$ separately requires subtle modifications
to $\varphi$, resulting in us searching for $p$ and $s$ in the context of the modified
formulae $\varphi_{\pref}$ and~$\varphi_{\suf}$, respectively.
If we find models of $\varphi_{\pref}$ and~$\varphi_{\suf}$ of a particular
form, we compose them into a~model of~$\varphi$. 

Let $z \in \varsOf{\haystack}$ be a~non-flat variable, and let $z_\pref$
and $z_\suf$ be two fresh variables with their languages restricted to
$z_\pref \in \prefixesOf{\mathcal{L}_z}$ and $z_\suf \in \suffixesOf{\mathcal{L}_z}$.
Let
$\varphi_\mathrm{Pref}$ and $\varphi_\mathrm{Suf}$ be formulae defined as
$\varphi_\mathrm{Pref} \triangleq \varphi[z/z_\pref \concat \sep]$ and $\varphi_\mathrm{Suf}  \triangleq \varphi[z/\sep\concat z_\suf]$
where $\sep$ is a fresh alphabet symbol.
Further,
let $\sigma^\pref \models \varphi_\pref$ and $\sigma^\suf \models \varphi_\suf$
be two models such that:
\begin{enumerate}
    \item $\sigma^\pref$ and $\sigma^\suf$ agree on the values of variables different than $z_\pref$ and $z_\suf$,
    \item $|\sigma^\pref(x)| > 2 \pMaxPrim \land |\sigma^\suf(x)| > 2 \pMaxPrim$ for any flat variable $x \in \flatVars$,
    \item $\sigma^\pref(z_\pref) = p \gamma_z^{K}$ and $\sigma^\suf(z_\suf) = \gamma_z^{L}s$ such that
        $n |\gamma_z| \ge 4\pMaxPrim + 2\pMaxLit$ for $n \in \{K, L\}$.
\end{enumerate}

\begin{restatable}{lemma}{lemSewing} \label{lemma:sewing}
  The assignment
    $\sigma^\pref \variant \{ z \mapsto p \gamma_z^{K} s\}$ is a~model of $\varphi$
    where $K = \min(K, L)$.
\end{restatable}
\begin{proof}[Proof sketch.]
    The idea behind the proof is that if $\sigma \not \models \varphi$, then
    $\sigma(z)$ would need to have a large conflict-free overlap with $\sigma(x)$
    for some flat variable $x \in \vars$. Applying \cref{cor:infixFineWilf}, we
    reach a contradiction.
    The full proof of \cref{lemma:sewing} can be found in
    \ifTR{}\cref{sec:proofsTools}\else{}\cite{techrep}\fi.
\end{proof}

\vspace{-4.0mm}
\subsection{Prefix (Suffix) Enumeration through Prefix (Suffix) Trees}
\vspace{-0.0mm}

Having defined $\Gamma^K$-expansion that acts similarly to inserting a fresh
symbol~$\sep$ between a~chosen $p \in \prefixesOf{\lang_z}$ and a suffix $s
\in \suffixesOf{\lang_z}$ of a non-flat variable $z \in \varsOf{\haystack}$, we can start enumerating prefixes $p \in \prefixesOf{\lang_z}$ (or
suffixes) up to a~certain bound, while searching for a~model. We introduce
the concept of prefix (suffix) trees that play a major role in our proofs.
Below, we give only the definition of a~prefix tree; a~suffix tree is
defined symmetrically.

\begin{definition}[Choice state]
    Let $\aut = (Q, \Delta, I, F)$ be a~DFA. We say that a state $q \in Q$ is a \emph{choice state}
    if $\big| \{ (q, a, r) \in \Delta : a\in \Sigma, r \in Q\} \big| > 1$. We write $C(\aut)$ to denote the set
    of all choice states of~$\aut$.
\end{definition}

\begin{definition}[Prefix tree]
    Let $z \in \vars$ be a variable with its language $\lang_z$ given by a~DFA $\aut_z = (Q_z, \Delta_z, \{q_0\}, F_z)$.
    We define $z$'s \emph{prefix tree} $T_z = (V_z, E_z, r_z, \stateLabelingFn_z,
    \edgeLabelingFn_z)$ as an (infinite finitely-branching) tree with
    vertices~$V_z$ rooted in $r_z \in V_z$ such that
    \begin{itemize}
        \item $\stateLabelingFn_z\colon V_z \rightarrow Q_z$ is a function that
            labels non-root vertices of $T_z$ with $\aut_z$'s choice
            states, i.e., $\stateLabelingFn_z(v) \in C(\aut_z)$
            for any $v \neq r_z$ and $\stateLabelOf{r_z} = q_0$,
        \item $E_z \subseteq V_z \times \Sigma^+ \times V_z$ is a~set of
          labelled edges such that
          $(v, a_1\ldots a_n, v') \in E_z$ iff there is a~run
          $\stateLabelingFn_z(v) \tran{a_1} q_1 \tran{a_2} \cdots \tran{a_n}
          \stateLabelingFn_z(v')$ in~$\aut_z$ where for all $0 < i < n$ it
          holds that $q_i \notin C(\aut)$.
        %
        \item $\edgeLabelingFn_z\colon V_z \times V_z \rightharpoonup \Sigma^*$
          is a function that maps any two vertices connected by an edge to the
          label on the edge, i.e., $\edgeLabelingFn_z(v, v') = w$ iff there
          exists an edge $(v,w,v') \in E_z$ and is undefined otherwise.
          %
    \end{itemize}
\end{definition}

Intuitively, vertices of the tree are labeled by $\mathcal{A}_z$'s choice
states $C(\aut_z)$, i.e., states in which we can choose between
multiple outgoing transitions along different alphabet symbols.
Vertices~$s$ and~$s'$
are connected by an edge in~$T_z$ if 
$\stateLabelOf{s'}$ is reachable from $\stateLabelOf{s}$ without passing through any choice
state.

A path $\pi$ in $T_z$ is a sequence of vertices $\pi = s_0 \dots s_n$ where 
$(s_i, w_i, s_{i+1}) \in E_z$ for any $0 \le i < n$. We lift the definition of $\edgeLabelingFn$
to paths as $\edgeLabelingFn(s_0 \ldots s_n) \triangleq \edgeLabelingFn\big( (s_0, s_1) \big)\concat \cdots \concat \edgeLabelingFn\big( (s_{n-1}, s_n) \big)$.

\begin{definition}[Dead-end vertex of a prefix tree]
    Let $\varphi = \notcontains(\needle, \haystack) \land \langconstr$
    and $T_z = (V, E, v_0, \stateLabelingFn, \edgeLabelingFn)$ be the prefix tree for~$z \in \vars$,
    and let $\sigma\colon (\vars \setminus \{z\}) \rightarrow \Sigma^*$ be a~partial
    assignment. 
    A vertex $v_n \in V$ is called a \emph{dead end} in~$T_z$ w.r.t.\ $\sigma$ if $\sigma' \not \models
    \varphi[z/z\sep]$ where
    $\sigma' \defeq \sigma \variant \{ z \mapsto \edgeLabelingFn(v_0 \ldots v_n)\}$ for $v_0 \ldots v_n$ being the (single) path between~$v_0$ and~$v_n$ in~$T_z$.
\end{definition}

Intuitively, dead-end vertices (and all vertices that are below them in the
prefix tree) are not interesting for obtaining a~$\notcontains$ model.
Consider, e.g., $\varphi \triangleq \notcontains(\texttt{ab}x, xz) \land \langconstr$
with $\lang_x = (\texttt{ab})^+$ and $\lang_z = (\letter{a} \{\letter{b}, \letter{c}\} \letter{c})^*$.
We have $\varphi[z/z\sep] = \notcontains(\needle', \haystack') = \notcontains(\letter{ab}x, xz\sep)$
and, thus, the vertex $v \in V_z$ corresponding to the prefix $\letter{abca}$ is a dead end in~$T_z$ w.r.t.
$\sigma = \{x \mapsto \letter{ab} \}$ since
$\sigma'(\needle') = \letter{abab}$ is a~factor of
$\sigma'(\haystack') = \letter{ababca}\sep$.

\begin{definition}[$H$-reaching path]
    Let $\pi = v_0 \dots v_n$ be a path in a prefix tree $T_z = (V, E, v_0,
    \stateLabelingFn, \edgeLabelingFn)$ and $H \in \nat$.
    We say that $\pi$ 
    is \emph{$H$-reaching} if $|\edgeLabelOf{v_0 \ldots v_n}| \ge H \ge |\edgeLabelOf{v_0 \ldots v_{n-1}}|$.
\end{definition}

In our proof, we explore all prefixes of words in a language up
to a certain bound~$H$.
As we have a~prefix tree with edges labelled with
words of (possibly) different lengths, stating that we have explored all
prefixes of the length precisely~$H$ is problematic.
Hence, the concept of
$H$-reaching paths is a relaxation allowing paths (prefixes) to slightly
vary in length.


\vspace{-2.0mm}
\section{Underapproximating Non-Flat Variables} \label{sec:singleSided}
\vspace{-2.0mm}

In this section, we give the main lemma allowing to underapproximate 
the language of non-flat variables with a~flat language.
Throughout this section we use three constants $\boundFlat, \boundAut, \boundGamma \in \naturals$ with the following semantics:
\begin{itemize}
    \item The constant $\boundAut$ is the length of prefixes (suffixes) of non-flat variables
        that we will enumerate in our proofs,
        searching a~model that is shorter w.r.t. some non-flat variable.
    \item The constant $\boundFlat$ is the minimal size of words assigned to flat variables occurring in $\needle$.
    \item $\boundGamma$ is used as the value of the parameter $K$ in every application of $\gammaPref{z}^K$ or 
        $\gammaSuf{z}^K$.
\end{itemize}

First, let us define parameters of $\notcontains(\needle, \haystack)$ that we use to define the above bounds.
Let $\pMaxLit$ be the length of the longest string literal in $\needle$ and $\haystack$, let $\pMaxAut$
be the largest number of states of a~DFA associated with some variable. Furthermore, let $\pMaxPrim$
be the length of the longest word in the set $W_\alpha(\needle) \cup W_\gamma$
where $W_\gamma$ is the set of the primitive words~$\gamma_z$ used to define the $\Gamma_z$-expansion
for every non-flat variable $z \in \varsOf{\haystack}$.

Since $\boundFlat$ and $\boundGamma$ depend on the value of $\boundAut$, we start by fixing
$\boundAut \triangleq 2 \pMaxPrim \pMaxAut + \pMaxLit \label{eq:boundAutDef}$.
Intuitively, for any non-flat variable $z \in \vars$, we set up $\boundAut$ in
a way so that if we consider all prefixes in $T_z$ up to the length $\pMaxPrim \pMaxAut$,
then $T_z$ will contain paths through any state $q \in Q_z$ since~$\aut_z$ is
a~single SCC.
After extending these paths up to the length~$\boundAut$, we can guarantee that~$T_z$
will contain all words read from any state $q$ of the length
at least~$\pMaxPrim$. Considering all possible words of the length~$|\beta|$ for some $\beta \in \primBase(\needle)$
readable from a state 
will be crucial later, as we will show that there can be only a~few such words
if we fail to find an alternative model $\sigma' \triangleq \sigma \variant \{z \mapsto w_z \}$
s.t.\ $|\sigma'(z)| < |\sigma(z)|$, assuming the existence of a model $\sigma$.

The remaining bounds $\boundFlat$ and $\boundGamma$ are defined as $\boundFlat  \triangleq \boundAut + 4 \pMaxPrim + \pMaxAut$ and
$\boundGamma \triangleq \boundFlat + 2 \pMaxPrim + 2 \pMaxLit$.
Ignoring some technical details and due to reasons that will be revealed shortly, we need $\boundFlat$ to be
slightly longer than $\boundAut$, so that when we later construct $\sigma'
\triangleq \sigma \variant \{z \mapsto p \}$ for some particular
prefix $|p| \le \boundAut + \pMaxAut$, we can establish some
of the string that precedes an~occurrence of $z$ in $\sigma_{|\vars \setminus \{z\}}(\haystack)$
in the case~$\sigma'$ fails to be a model.
Finally, $\boundGamma$ is set up so that together with $\boundFlat$ they allow
\cref{lemma:sewing} to be applied, where $\boundAut$ and $\boundFlat$
play the role of $K_0$ and $N_0$, respectively.

We remark that the exact values of $\boundAut$, $\boundFlat$, and $\boundGamma$ are
not important when reading the proof for the first time. It is sufficient to note
that $\boundAut < \boundFlat < \boundGamma$, and that the difference in sizes between
these bounds is sufficiently large.

\vspace{-3.0mm}
\subsection{Overcoming the Infinite by Equivalence with a Finite Index}
\vspace{-2.0mm}

Our procedure to decide $\varphi = \notcontains(\needle, \haystack)$ containing
non-flat variables in $\haystack$ originates in enumeration of partial
assignments $\eta \colon \varsOf{\needle} \rightarrow \alphabet^*$, since it is
easy to find suitable values for non-flat variables when $\needle$ is a literal
due to us fixing values of all variables in $\needle$. The problem is that
there is an infinite number of such assignments. Our key observation allowing
us prove that we can underapproximate non-flat languages using flat ones is
that the precise values of flat variables occurring in $\needle$ does not matter as
long as these variables have assigned sufficiently long words.
In general, however, a~model $\sigma \models \varphi$ might assign long words only
to a subset of flat variables. Therefore, in our decision procedure, we first guess
the set $X$ of flat variables that are assigned words shorter than $\boundFlat$. Since
there are finitely many such words, we have a finite number of possible choices $\tau\colon X \rightarrow \Sigma^*$
of values these variables can attain. We enumerate all possible valuations $\tau$,
and for every such a valuation $\tau$ we produce a new constraint $\varphi_\tau$
in which we replace every short variable $x \in X$ by the word $\tau(x)$.
The regular constraints restricting the remaining flat variables are modified
to permit only words longer than $\boundFlat$, allowing us to assume in our proofs
that flat variables in $\varphi_\tau$ have assigned sufficiently long words.

Let us formalize our observation that the precise length of words assigned to flat variables
does not matter as long as they are sufficiently long. Let $\eta \colon
\varsOf{\needle} \rightarrow \alphabet^*$ be a partial assignment. We define
the set $X_{< \lambda}$ as $X_{< \lambda}(\eta) \triangleq \big\{x \in
\varsOf{\needle} : |\eta(x)| < \lambda \big\}$. Given a~constant $\lambda \in
\naturals$, we say that two partial assignments~$\eta$ and~$\vartheta$ are
$\lambda$-equivalent denoted by $\eta \sim_{\lambda} \vartheta$ iff $\eta
\sim_\lambda \vartheta \defequiv \eta_{|X_{< \lambda}(\eta)} =
\vartheta_{|X_{< \lambda}(\vartheta)}$.

Clearly, $\sim_\lambda$ has a finite index and if there exists a model $\sigma$ of $\varphi$, then its
restriction $\sigma_{|\varsOf{\needle}}$ will fall into one of the equivalence classes induced by $\sim_\lambda$.
Setting $\lambda = \boundFlat$, we inspect all equivalence classes, checking
whether any of them contains a model. Given a representative $\eta$ of an
equivalence class, we replace all variables $x \in \varsOf{\needle}$
with $\eta(x)$ if $|\eta(x)| < \boundFlat$, producing a new constraint $\varphi_\eta$. Furthermore, we need to include
the fact that the remaining variables in $\needle$ have assigned long words.
Therefore, the languages of all variables $y \in \varsOf{\needle}$ such that $|\eta(y)| \ge \boundFlat$ will
have their their languages restricted to a~new language $\mathcal{L}'_y = \mathcal{L}_y \cap \{|w| \ge \boundFlat \mid w \in \Sigma^* \}$
in $\varphi_\eta$.
The resulting constraint $\varphi_\eta$ is clearly equisatisfiable to $\varphi$
with models restricted to be $\sim_{\boundFlat}$-equivalent to $\eta$.


Some of these instances can be decided without any additional work. In particular,
if we have an assignment $\eta$ such that $X_{< \boundFlat}(\eta) = \varsOf{\needle}$,
i.e., all variables occurring in $\needle$ are short, we fix values of
all variables in $\needle$, and, thus, the $\needle$eedle of $\varphi_\eta$ is a word.
The remaining instances with $X_{< \boundFlat}(\tau) \subset \varsOf{\needle}$
contain at least one occurrence of a~(long) variable in $\tau(\needle)$, and, thus,
their decidability requires investigation. 

\vspace{-2.0mm}
\subsection{Inspecting the Structure of Non-flat Variables in the Presence of Long Flat Variables}
\vspace{-1.0mm}

Throughout this section, we fix $\varphi$ be a $\notcontains$ instance resulting
from the previous section, i.e., $\varphi \triangleq \varphi_\eta = \notcontains(\needle, \haystack
\land \langconstrof{\lang ,\eta}$
for some equivalence class representative $\eta$
such that $\varsOf{\needle} \neq \emptyset$. We start by stating the key theorem for our decidability result.
\begin{theorem} \label{thm:constrflat}
    Let $z$ be a (decomposed) non-flat variable present in $\haystack$.
    There is a flat language $\langflatunderapp_z \subset \lang_z$ s.t.\ if there exists
    a model $\sigma \models \varphi$,
    then there exists a model $\sigma' \models \varphi[z/z\sep]$
    s.t.\ $\sigma' \triangleq \sigma \variant \{z \mapsto w_z \}$ for some word $w_z \in \langflatunderapp_z$.
\end{theorem}

Before presenting quite technical lemmas that allowed us to obtain the result,
let us derive some intuition on why the theorem holds. Assume that we have a
model $\sigma$ of $\varphi$ and we we pick some long prefix $p$ and a long suffix $s$
for the variable $z$, and we glue them together using $\Gamma_z$-expansion to
produce a word $w \triangleq \Gamma^{\boundGamma}_z(p, s)$ and an altered assignment
$\sigma' \triangleq \{ z \mapsto w \}$. The core of the theorem lies in analyzing
the situation when $\sigma'$ fails to be a model. By symmetry, we focus on the case
when our choice of the prefix $p$ is problematic. We have two possibilities.
\begin{itemize}
    \item There is a short prefix of $p$ due to which $\sigma'$ fails to be a model. We address this
        by systematically exploring the prefix tree of $z$ up to a certain bound,
        marking the vertices that correspond to such prefixes as dead ends.
    \item Our choice of $p$ does not cause $\sigma'$ to immediately fail to be a model, however,
        by applying $\Gamma_z$-expansion we introduce an infix due to which $\sigma' \not \models \varphi$.
        Since we assume that $\varphi$ results from a previous section, we know that all flat variables have assigned a long word, i.e.,
        $\sigma'(\needle)$ contains long factors of the form $\alpha^k$ for some $\alpha$ which
        forms the basis of a flat variable $x \in \varsOf{\needle}$ and $k \in \naturals$.
        $\Gamma_z$-expansion glues together a prefix and a suffix using a word $\gamma^K_z$
        where $|\gamma_z| \neq |\alpha|$ for any base $\alpha$. Therefore, we know that only a limited part of the infix
        introduced by $\Gamma_z$-expansion is problematic, otherwise we would have a long overlap
        between $\gamma^K_z$ and some factor $\alpha^k$ of $\sigma'(\needle)$.
        Thus, $p$ contains a long factor $\alpha^k$ for some $k \ge 1$ and a primitive $\alpha$
        word $\alpha$ that forms the base of a flat variable present in $\needle$. We carefully analyze
        the effect of such a factor on the structure of $\aut_z$.
\end{itemize}

We now provide an overview of lemmas that lead to \cref{thm:constrflat}. Since
these lemmas are quite technical, we accompany them with intuition and
only sketch their proofs. Full proofs can be found in
\ifTR{}\cref{sec:proofsSingleSided}\else{}\cite{techrep}\fi.
To simplify the presentation, we focus primarily on attempting to find a suitable
prefix of a non-flat variable, and hence, our results are formulated in the
context of a modified formula that contains a fresh alphabet symbol $\sep$.
Since the situation is symmetric for suffixes, we can use the properties
of $\Gamma_z$-expansion (\cref{lemma:sewing}) and glue together a suitable prefix and a suitable
suffix to produce an altered model.

We start with a technical lemma used frequently in our proofs.
The lemma shows that if we know that $\alpha$ is a factor of $\haystack$,
and we know that a~part of $\haystack$ in the proximity of the factor $\alpha$
is incompatible with $\alpha$, then we can show that a~large number of overlaps
between~$\sigma(\needle)$ and~$\sigma(\haystack)$ must
contain a~conflict if $\needle$ contains a large factor of the form $\alpha^N$.

\begin{restatable}{lemma}{lemConflictKaboom} \label{lemma:conflictKaboom}
    Let $\alpha$ and $\gamma$ be two primitive words, such that $|\alpha| \neq |\gamma|$.
    Let $\haystack = t \alpha u \gamma^{\boundGamma}$ and $\needle = v \alpha^{N} w$
    where $t, u, v$ and $w$ are (possibly empty) words such that $N |\alpha| > \boundFlat$ and $u < \boundFlat - 2\max(|\alpha|, |\gamma|)$.
    If
    \begin{enumerate*}
        \item the prefix $p$ of $\haystack$ of the size $|p| = |t| + |\alpha| + |w|$
            does not contain the word $\needle$, \label{cond1}
        \item $\alpha \not \in \prefixesOf{v_0}$ and $v_0 \not \in \prefixesOf{\alpha}$, \label{cond2} 
    \end{enumerate*}
    then $\needle$ is not a factor of $\haystack$.
\end{restatable}

\begin{proof}[Proof sketch.]
    Since $\needle$ contains a long factor $\alpha^N$ and $\haystack$ contains at least one $\alpha$,
    we can apply \cref{lemma:primitiveAlignment} to rule out a lot of overlaps that might be conflict-free.
    All of the remaining overlaps contain a conflict thanks to Condition~\ref{cond2}.
    The full proof is available in \ifTR{}\cref{sec:proofsSingleSided}\else{}\cite{techrep}\fi. 
\end{proof}
\vspace{-5mm}

\medskip
\begin{restatable}{lemma}{lemFlatify}\label{lemma:flatify}
    Let $x \in \vars$ be a flat variable with $\primBase(x) = \{ \alpha \}$, and let $z \in \varsOf{\haystack}$ be a (decomposed) non-flat variable.
    Let $\varphi$ $\varphi \triangleq \notcontains(\needle,
    \haystack) = \notcontains(\needle' x \alpha^M p W, \haystack)$ be formula with $p \neq
    \alpha$ being a prefix of $\alpha$ and $W = a_0 \dots a_n$ being a
    non-empty word such that the word $p a_0$ is not a prefix of $\alpha$.

    If there exists a model $\sigma$ with $\sigma(z)$ being of the form $\sigma(z) = s \alpha^k p W V$
    for some word $V$, $k \ge 1$, and a suffix $s$ of $\alpha$,
    then $\sigma \variant \big\{ z \mapsto \gammaPref{z}^{\boundGamma}(s \alpha^k p W) \big\}$
    is a model of $\varphi[z/z\sep]$.
\end{restatable}

Intuitively, the rightmost variable in $\needle$ is the flat variable $x$ with $\primBase(x) = \{\alpha\}$. To the right of $x$, there is a~literal with the prefix
$\alpha^M p$ that resembles the flat language $\lang_x$. Moreover, $\sigma(z)$
also starts with a prefix $s \alpha^k p$ resembling $\lang_x$, followed by $W$.
Thus, the prefix $s \alpha^k p W$ of the word $\sigma(z)$ mimics
the suffix of the right-hand side $\sigma(\needle)$. Hence, if we look
solely on the prefix of $\sigma(z)$ and the suffix of $\sigma(\needle)$, there are no
obvious conflicts.
\begin{center}
\tikzstyle{wordBNode} = [draw, scale=0.8, inner sep=1mm, text height=3mm, minimum height=7mm, baseline]
\begin{tikzpicture}
    \node[wordBNode, anchor=west, minimum width=8mm, gray] (a0) at (0, 0) {$\dots$};
    \node[wordBNode, anchor=west, minimum width=4mm]  (a1) at (a0.east) {$s_\alpha$};
    \node[wordBNode, anchor=west, minimum width=10mm] (a2) at (a1.east) {$\dots$};
    \node[wordBNode, anchor=west, minimum width=12mm] (a3) at (a2.east) {$\alpha$};
    \node[wordBNode, anchor=west, minimum width=4mm]  (a4) at (a3.east) {$p_\alpha$};
    \node[wordBNode, anchor=west, minimum width=8mm]  (a5) at (a4.east) {$W$};
    \node[wordBNode, anchor=west, minimum width=12mm]  (a6) at (a5.east) {$\dots$};

    \node[wordBNode, anchor=west, minimum width=8mm, gray] (r0) at (-0.50, -0.7) {$ \dots  $};
    \node[wordBNode, anchor=west, minimum width=12mm]      (r1) at (r0.east)     {$ \alpha  $};
    \node[wordBNode, anchor=west, minimum width=10mm]       (r2) at (r1.east)     {$ \dots  $};
    \node[wordBNode, anchor=west, minimum width=12mm]      (r3) at (r2.east)     {$ \alpha  $};
    \node[wordBNode, anchor=west, minimum width=4mm]      (r4) at (r3.east)     {$ p_\alpha  $};
    \node[wordBNode, anchor=west, minimum width=8mm]       (r5) at (r4.east)     {$ W $};

    \node[] at (-2, 0)    {$\sigma(\haystack)$};
    \node[] at (-2, -0.7) {$\sigma(\needle)$};
    
    \coordinate (z_end) at ($(a5.north east)+(0.2, 0)$);
    \coordinate (center) at ($(a1.north west)!0.5!(z_end)$);
    \node[scale=0.7] (sigma_z) at ($(center)+(0, 0.5)$) {$\sigma(z)$};
    \def\shift{0.05}
    \draw (a1.north west) edge[->, out=40,in=-130] ($(sigma_z.south)+(-\shift, 0)$);
    \draw (z_end) edge[->, out=140,in=-50] ($(sigma_z.south)+(\shift, 0)$);
\end{tikzpicture}
\end{center}
However, $\sigma \models \varphi$, and, therefore, there must be a conflict
outside of $z$ when considering the above alignment.  The rest of the proof can
be found in \ifTR{}\cref{sec:proofsSingleSided}\else{}\cite{techrep}\fi.

Next, we derive a lemma formalizing that we can restrict languages of non-flat variables to flat ones,
producing an~equisatisfiable instance. Stating a symmetric lemma for suffixes, and applying
\cref{lemma:sewing} would give us the entire proof of \cref{thm:constrflat}.
\begin{restatable}{lemma}{lemConstrFlat}\label{lemma:constrflat}
    Let $x$ be the rightmost flat variable with $\primBase(x) = \{\alpha\}$, and let $z$ be a non-flat
    variable occurring in $\haystack$.
    Let $\varphi$ be a formula $\varphi \triangleq \notcontains(\needle,
    \haystack) \land \langconstr = \notcontains(\needle' x \alpha^M p W,
    \haystack) \land \langconstr$
    where $\needle' \in (\vars \cup \alphabet)^*$, 
    $M \ge 0$, $p \neq \alpha$ is a prefix of $\alpha$,
    and $W = a_0 \dots a_n$ is a non-empty word such that $p a_0$ is not a prefix of $\alpha$.
    There exists a flat language $\langflatunderapp_z \subset \lang_z$ s.t.\ if there is
    a~model $\sigma \models \varphi$, then
    $\sigma \variant \{z \mapsto w_z \} \models \varphi[z/z\sep]$ for some word $w_z \in \langflatunderapp_z$.
\end{restatable}

The intuition behind \cref{lemma:constrflat} is the same as the intuition behind \cref{thm:constrflat}.
We note that the lemma requires the word $W$ to be non-empty. The case for when $W = \varepsilon$ has a~similar,
but simpler proof.

\smallskip

\tikzstyle{wordBNode} = [draw, scale=0.8, inner sep=1mm, text height=3mm, minimum height=7mm, baseline, fill=white, draw=black]

\newcommand{\prefl}[0]{u}
\vspace{-3mm}
\begin{proof}[Proof sketch]
    We assume the existence of a model $\sigma \models \varphi$, and we
    systematically explore the prefix tree $T_z$ of the variable $z$ in a
    breadth-first fashion up to the bound $\boundAut$, searching for the word
    $w_z$. When inspecting any prefix $\prefl$, we check whether $\sigma \variant \{
    z \mapsto \prefl \}$ is a model of $\varphi[z/z\sep]$, and if not we mark the
    last vertex corresponding to $\prefl$ as a dead end and we do not explore it
    further. At the end of the exploration, we inspect the set
    $\mathcal{P}_{\ge \boundAut}$ of all $\boundAut$-reaching paths in $T_z$.
    If there are no such paths, we know that $|\sigma(z)| < \boundAut$, and
    hence, $w_z$ can be found in the finite (flat) language $
    \{ w \in \mathcal{L}_z \mid |w| < \boundAut \}$. Alternatively, we check
    for every path $\pi$ in $\mathcal{P}_{\ge \boundAut}$ whether $\sigma \variant \{z
    \mapsto \gammaPref{z}^{\boundGamma}(\prefl_\pi) \} \models \varphi[z/z\sep]$
    where $\prefl_\pi$ is the prefix corresponding to $\pi$. Since, the number of
    possible $\boundAut$-reaching paths is finite, we have that $w_z$ can be found in the flat language
    $\{ \gammaPref{z}^{\boundGamma}(p_\pi) \mid \pi \in \mathcal{P}_{\ge \boundAut}\}$
    in the case that $\gammaPref{z}$-expansion of some $\prefl_\pi$ leads to a~model of $\varphi[z/z\sep]$.

    Next, it might be that none of the paths in $\mathcal{P}_{\ge \boundAut}$ can be $\gammaPref{z}$-expanded
    into a model.
    Let $x$ be the rightmost variable in $\mathcal{N}$ with $\primBase(x) = \{ \alpha \}$.
    Recall that none of the paths in $\mathcal{P}_{\ge \boundAut}$ contain
    a dead-end, and that all variables in $\mathcal{N}$ are flat, and, thus, they have assigned long words.
    Combined with the properties of $\gammaPref{z}$-expansion, we
    know the reason why $\sigma' \triangleq \sigma \variant \{z \mapsto \gammaPref{z}^{\boundGamma}(p_\pi) \}$
    fails to be a model, i.e., we almost accurately know the position of $\sigma'(\needle)$ in
    $\sigma'(\haystack)$. We show that all paths in $\mathcal{P}_{\ge \boundAut}$ share
    the same prefix of the form $s \alpha^M p$ for some large $k \in \naturals$
    and $s \in \suffixesOf{\alpha}$ and $p \in \prefixesOf{\alpha}$.
    Since $z$ is non-flat, and $\boundAut$ is larger than the number of states of $\aut_z$,
    we have opportunities to diverge from the shared prefix $s \alpha^M p$ in $T_z$.
    We show that diverging must immediately lead to a dead-end vertex, and in such a case $\sigma(z)$ has a prefix $s \alpha^M p W$.
    Hence, we apply \cref{lemma:conflictKaboom} and obtain that $\sigma \variant \{z \mapsto w_{z,M} \}$
    is a model of $\varphi[z/z\sep]$ where $w_{z, M} = \gammaPref{z}^{\boundGamma}(s \alpha^M p W)$. Note that $w_{z, M}$ depends
    on an unknown integer $M$, however, the language containing all $w_{z, M}$s for every possible choice of $k$ is flat.
    Alternatively, not-diverging from the path implies that $\sigma(z) \in s \alpha^* p'$ for some $p' \in \prefixesOf{\alpha}$,
    which is again a flat language.
\end{proof}

\vspace{-2mm}
A careful analysis of the proof of \cref{lemma:constrflat} reveals that the
lemma, and, therefore, \cref{thm:constrflat}, is not effective in a sense that
one cannot directly construct~$\langflatunderapp_z$. However, we can obtain a
decision procedure at the cost of a producing larger flat language.
We construct $T_z$ and all paths up to the bound $\boundAut$ without having
$\sigma$ available, losing the ability to mark dead-end vertices. In the
resulting flat language $\langflatunderapp_z$ we include all words shorter than
$\boundAut$, and $\gammaPref{z}$-expansions of all $\boundAut$-reaching paths.
The remaining parts of $\langflatunderapp_z$ that correspond to the situation
when no $\boundAut$-reaching paths can be $\gammaPref{z}$-expanded into a model
can be computed from $\aut_z$ without requiring access to the original model $\sigma$.
For details, we refer the reader to the full proof of \cref{lemma:constrflat} in
\ifTR{}\cref{sec:proofsSingleSided}\else{}\cite{techrep}\fi{}.

\SetKwComment{tcc}{$\triangleright~$}{}
\SetKwComment{tcp}{$\triangleright~$}{}

\vspace{-2.0mm}
\section{Decision Procedure}\label{sec:label}
\vspace{-1.0mm}

Finally, we summarize the approach described in previous sections
into a decision procedure for $\notcontains$. The (nondeterministic)
algorithm is shown in \cref{alg:decproc}. 
%
In~the algorithm, for a negated containment $\varphi$ and a (partial) assignment~$\sigma$, 
we use~$\sigma(\varphi)$ to denote the $\notcontains$ predicate obtained from $\varphi$ replacing 
variables whose assignment is defined with the corresponding assignment.

\newcommand\lPrefLang[0]{\lang^\mathrm{Pref}_x}
\newcommand\lPrefPrefixes[0]{\mathrm{P}_{\mathrm{inc}}}
\newcommand\lPrefWords[0]{\mathrm{P}_{w}}

\newcommand\lSufLang[0]{\lang^\mathrm{Suf}_x}
\newcommand\lSufPrefixes[0]{\mathrm{S}_{\mathrm{inc}}}
\newcommand\lSufWords[0]{\mathrm{S}_{w}}

\begin{algorithm}[t]\small
    \SetAlgoLined
    \KwIn{$\varphi \defeq \notcontains(\needle, \haystack) \land \langconstr$}
    \KwOut{Satisfiability of $\varphi$}
     
    \smallskip
    Normalize $\varphi$ into $\bigvee_i \varphi_i$ and \textbf{Guess} $\varphi_i$ \tcc*[r]{\cref{sec:normalization}}
    \lIf(\tcc*[f]{\cref{sec:easy}}){$\varphi_i$ is easy}{\Return
      $\mathit{solve}(\varphi_i)$}
    Remove all two-sided $x \in \vars \setminus \flatVars$ from $\varphi_i$ \tcc*[r]{\cref{lemma:twoSidedRemoval}}
    \lIf(\tcc*[f]{\cref{sec:tools}}){$\varphi_i$ is easy}{\Return
      $\mathit{solve}(\varphi_i)$}
    \textbf{Guess} $P \subseteq \flatVars$ \tcc*[r]{vars.\ with assignments longer than $\boundFlat$}
    \textbf{Guess} assignment $\sigma$ s.t. $\forall x \in \flatVars\setminus P\colon |\sigma(x)| < \boundFlat$\;
    Let $\varphi' := \sigma(\varphi_i) = \notcontains (\needle', \haystack')$\; 

    \ForEach{$x \in (\vars \setminus \flatVars) \cap \varsOf{\haystack'}$}{
        Constrain prefix/suffix trees $T_x$/$S_x$ up to bound $\boundAut$\;
        Constrain prefix language $\lPrefLang$ \tcc*[r]{\cref{lemma:constrflat}}
        Constrain suffix language $\lSufLang$ \tcc*[r]{symmetric to \cref{lemma:constrflat}}
        Construct flat $\lang_x' := \glue(\lPrefLang, \lSufLang)$\;
    }
    \Return $\mathit{solve}(\varphi' \land \{ x \mapsto \lang_x' \mid x \in \vars\setminus\flatVars \})$ \tcc*[r]{\cref{lemma:liaDecidability}}
    \caption{Decision Procedure for $\notcontains$}
    \label{alg:decproc}
\end{algorithm}

The set $\lPrefLang$ ($\lSufLang$) contains prefixes (suffixes) of words from
$\lang_z$ that might be used to find an alternative model. The
$\glue(\lPrefLang, \lSufLang)$ procedure glues together prefixes and suffixes,
resulting in a language consisting of entire words (not just prefixes or suffixes) from~$\lang_z$.
The procedure partitions the language $\lang^\mathrm{Pref}_x$ into $\lang^\mathrm{Pref}_x = \lPrefWords \cup \lPrefPrefixes$
such that $\lPrefPrefixes$ consists of words resulting from an application of $\gammaPref{z}$-expansion. 
Intuitively, the words in $\lPrefWords$ are words from $\lang_x$
whereas~$\lPrefPrefixes$ are only prefixes that need to be completed
into full words from $\lang_x$ by concatenating suitable suffixes.
We decompose $\lang^{\mathrm{Suf}}_x$ in the same way into
$\lang^{\mathrm{Suf}}_x = \lSufWords \cup \lSufPrefixes$.
The procedure then returns
    $ \lang'_x = \lPrefWords \cup \lSufWords \cup
    \{ p \gamma^{\boundGamma} s \mid (p \gamma^{\boundGamma}, \gamma^{\boundGamma} s) \in \lPrefPrefixes \times \lSufPrefixes \}. $

\begin{theorem}[Soundness]
    If \cref{alg:decproc} terminates with an assignment $\sigma$, then $\sigma \models \varphi$.
\end{theorem}
\begin{proof}
    Follows from the fact that $\mathcal{L}'_x \subseteq \mathcal{L}_x$ for every non-flat variable $x$ found in $\mathcal{H}'$.
\end{proof}

\begin{theorem}[Completeness]
    If $\varphi$ is satisfiable, then \cref{alg:decproc} terminates with an~assignment $\sigma$ such that $\sigma \models \varphi$.
    Otherwise \cref{alg:decproc} terminates with the answer $\mathsf{UNSAT}$.
\end{theorem}
\begin{proof}
    Correctness for two-sided non-flat variables follows from \cref{lemma:twoSidedRemoval}. For remaining non-flat variables, correctness
    follows from \cref{lemma:constrflat}. Finally, correctness of the $\glue$ procedure follows from \cref{lemma:sewing}. 
\end{proof}

\vspace{-4.0mm}
\begin{theorem}\label{thm:mainComplexity}
A constraint $\notcontains(\needle, \haystack) \land \langconstr$ is decidable
in~$\clExpSpace$.
\end{theorem}

\begin{proof}[Proof sketch.]
Decidability follows from the analysis of the decision procedure in
\cref{alg:decproc} given above. As for $\clExpSpace$ membership, first notice
that languages of non-flat variables occurring only in $\mathcal{H}$ are
replaced with flat languages of polynomial size due to the bounds $\boundAut$
and $\boundGamma$. \cref{alg:decproc} then uses \cref{lemma:liaDecidability},
bringing the complexity of the procedure to $\clNExpTime$. However, obtaining
a full model that includes all two-sided non-flat variables brings the procedure
to $\clExpspace$ as the length of the word to assigned to a two-sided non-flat
variable doubles with each such a variable (cf. \cref{lemma:twoSidedRemoval}).
\end{proof}

\vspace{-4.0mm}
\subsection{Chain-Free Word Equations with $\notcontains$}\label{sec:label}
\vspace{-1.0mm}

After establishing the decidability of a~single $\notcontains$ predicate, we
immediately obtain decidability of string fragments that permit the so-called
\emph{monadic decomposition}~\cite{VeanesBNB17,ChenHLRW19}, i.e., expressing
the set of solutions as a~finite disjunction of regular membership constraints
$\bigwedge_{x\in\vars} x \in \lang_x$.
These include fragments such as the \emph{straight-line}
fragment~\cite{AnthonyTowards2016} or the more expressive \emph{chain-free}
fragment of word equations~\cite{ChainFree} (note that~\cite{ChainFree}
considers also other predicates).
We can therefore easily establish the following theorem.
\begin{theorem}\label{thm:chainFreeNotContains}
Formula $W \land \notcontains(\needle, \haystack) \land \langconstr$
  where~$W$ is a~conjunction of chain-free word equations is decidable.
\end{theorem}

\vspace{-4.0mm}
\section{Future Work}\label{sec:label}
\vspace{-3.0mm}

This paper shows that chain-free word equations with regular constraints and
a~single instance of the $\notcontains$ predicate are decidable.
There are several possible future work directions.
First, we wish to investigate the fragment where the number of $\notcontains$ constraints is not limited to a~single one.
Another direction is examining combinations of~$\notcontains$ with other
predicates, such as length constraints or disequalities.
The technique for combining these from~\cite{ChenHHHL25} based on the reduction
of the constraints to reasoning over Parikh images of finite automata is not
directly applicable here.
Also, the resulting complexity of our procedure is \clExpSpace, however, we have hints
that the problem might in fact be solvable in $\clNP$.

\newcommand{\ackPhdTalent}[0]{
\noindent
The work of Michal Hečko, a~Brno Ph.D.\ Talent Scholarship
\raisebox{-6pt}{\protect\includegraphics[height=17pt]{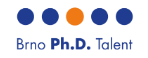}}
Holder, is funded by the Brno City Municipality.\xspace
}

\vspace{-4.0mm}
\section*{Acknowledgements}
\vspace{-3.0mm}
We thank the anonymous reviewers for careful reading of the paper and their
suggestions that greatly improved its quality.
This work was supported by 
the Czech Ministry of Education, Youth and Sports ERC.CZ project LL1908,
the Czech Science Foundation project 25-18318S, and
the FIT BUT internal project FIT-S-23-8151.
\ackPhdTalent

\bibliographystyle{plainurl}
\bibliography{literature}

\begin{thebibliography}{10}

\bibitem{Flatten}
Parosh~Aziz Abdulla, Mohamed~Faouzi Atig, Yu{-}Fang Chen, Bui~Phi Diep,
  Luk{\'{a}}{\v{s}} Hol{\'{\i}}k, Ahmed Rezine, and Philipp R{\"{u}}mmer.
\newblock Flatten and conquer: {A} framework for efficient analysis of string
  constraints.
\newblock In Albert Cohen and Martin~T. Vechev, editors, {\em Proceedings of
  the 38th {ACM} {SIGPLAN} Conference on Programming Language Design and
  Implementation, {PLDI} 2017, Barcelona, Spain, June 18-23, 2017}, pages
  602--617. {ACM}, 2017.
\newblock \href {https://doi.org/10.1145/3062341.3062384}
  {\path{doi:10.1145/3062341.3062384}}.

\bibitem{Trau}
Parosh~Aziz Abdulla, Mohamed~Faouzi Atig, Yu{-}Fang Chen, Bui~Phi Diep,
  Luk{\'{a}}{\v{s}} Hol{\'{\i}}k, Ahmed Rezine, and Philipp R{\"{u}}mmer.
\newblock Trau: {SMT} solver for string constraints.
\newblock In Nikolaj~S. Bj{\o}rner and Arie Gurfinkel, editors, {\em 2018
  Formal Methods in Computer Aided Design, {FMCAD} 2018, Austin, TX, USA,
  October 30 - November 2, 2018}, pages 1--5. {IEEE}, 2018.
\newblock \href {https://doi.org/10.23919/FMCAD.2018.8602997}
  {\path{doi:10.23919/FMCAD.2018.8602997}}.

\bibitem{parosh-notsubstring}
Parosh~Aziz Abdulla, Mohamed~Faouzi Atig, Yu-Fang Chen, Bui Phi Diep,
  Luk{\'a}{\v{s}} Hol{\'i}k, Denghang Hu, Wei-Lun Tsai, Zhillin Wu, and Di-De
  Yen.
\newblock Solving not-substring constraint with flat abstraction.
\newblock In {\em Programming Languages and Systems}, pages 305--320, Cham,
  2021. Springer International Publishing.

\bibitem{ChainFree}
Parosh~Aziz Abdulla, Mohamed~Faouzi Atig, Bui~Phi Diep, Luk{\'{a}}{\v{s}}
  Hol{\'{\i}}k, and Petr Jank{\r{u}}.
\newblock Chain-free string constraints.
\newblock In Yu{-}Fang Chen, Chih{-}Hong Cheng, and Javier Esparza, editors,
  {\em Automated Technology for Verification and Analysis - 17th International
  Symposium, {ATVA} 2019, Taipei, Taiwan, October 28-31, 2019, Proceedings},
  volume 11781 of {\em Lecture Notes in Computer Science}, pages 277--293.
  Springer, 2019.
\newblock \href {https://doi.org/10.1007/978-3-030-31784-3\_16}
  {\path{doi:10.1007/978-3-030-31784-3\_16}}.

\bibitem{aiswarya22}
C.~Aiswarya, Soumodev Mal, and Prakash Saivasan.
\newblock On the satisfiability of context-free string constraints with
  subword-ordering.
\newblock In {\em Proceedings of the 37th Annual ACM/IEEE Symposium on Logic in
  Computer Science}, LICS '22, New York, NY, USA, 2022. Association for
  Computing Machinery.
\newblock \href {https://doi.org/10.1145/3531130.3533329}
  {\path{doi:10.1145/3531130.3533329}}.

\bibitem{Alur11}
Rajeev Alur and Pavol \v{C}ern\'{y}.
\newblock Streaming transducers for algorithmic verification of single-pass
  list-processing programs.
\newblock {\em SIGPLAN Not.}, 46(1):599–610, January 2011.
\newblock \href {https://doi.org/10.1145/1925844.1926454}
  {\path{doi:10.1145/1925844.1926454}}.

\bibitem{tinelli-hotsos16}
Clark~W. Barrett, Cesare Tinelli, Morgan Deters, Tianyi Liang, Andrew Reynolds,
  and Nestan Tsiskaridze.
\newblock Efficient solving of string constraints for security analysis.
\newblock In William~L. Scherlis and David Brumley, editors, {\em Proceedings
  of the Symposium and Bootcamp on the Science of Security, Pittsburgh, PA,
  USA, April 19-21, 2016}, pages 4--6. {ACM}, 2016.
\newblock \href {https://doi.org/10.1145/2898375.2898393}
  {\path{doi:10.1145/2898375.2898393}}.

\bibitem{BerzishDGKMMN23}
Murphy Berzish, Joel~D. Day, Vijay Ganesh, Mitja Kulczynski, Florin Manea,
  Federico Mora, and Dirk Nowotka.
\newblock Towards more efficient methods for solving regular-expression heavy
  string constraints.
\newblock {\em Theor. Comput. Sci.}, 943:50--72, 2023.
\newblock \href {https://doi.org/10.1016/j.tcs.2022.12.009}
  {\path{doi:10.1016/j.tcs.2022.12.009}}.

\bibitem{Z3str3RE}
Murphy Berzish, Mitja Kulczynski, Federico Mora, Florin Manea, Joel~D. Day,
  Dirk Nowotka, and Vijay Ganesh.
\newblock An {SMT} solver for regular expressions and linear arithmetic over
  string length.
\newblock In Alexandra Silva and K.~Rustan~M. Leino, editors, {\em Computer
  Aided Verification - 33rd International Conference, {CAV} 2021, Virtual
  Event, July 20-23, 2021, Proceedings, Part {II}}, volume 12760 of {\em
  Lecture Notes in Computer Science}, pages 289--312. Springer, 2021.
\newblock \href {https://doi.org/10.1007/978-3-030-81688-9\_14}
  {\path{doi:10.1007/978-3-030-81688-9\_14}}.

\bibitem{Z3str4}
{Berzish, Murphy}.
\newblock {\em Z3str4: A solver for theories over strings}.
\newblock PhD thesis, University of Waterloo, 2021.
\newblock URL: \url{http://hdl.handle.net/10012/17102}.

\bibitem{BTV09}
Nikolaj~S. Bj{\o}rner, Nikolai Tillmann, and Andrei Voronkov.
\newblock Path feasibility analysis for string-manipulating programs.
\newblock In Stefan Kowalewski and Anna Philippou, editors, {\em Tools and
  Algorithms for the Construction and Analysis of Systems, 15th International
  Conference, {TACAS} 2009, Held as Part of the Joint European Conferences on
  Theory and Practice of Software, {ETAPS} 2009, York, UK, March 22-29, 2009.
  Proceedings}, volume 5505 of {\em Lecture Notes in Computer Science}, pages
  307--321. Springer, 2009.
\newblock \href {https://doi.org/10.1007/978-3-642-00768-2\_27}
  {\path{doi:10.1007/978-3-642-00768-2\_27}}.

\bibitem{BlahoudekCCHHLS23}
František Blahoudek, Yu{-}Fang Chen, David Chocholat{\'{y}}, Vojtěch Havlena,
  Luk{\'{a}}š Hol{\'{\i}}k, Ondřej Leng{\'{a}}l, and Juraj S{\'{\i}}č.
\newblock Word equations in synergy with regular constraints.
\newblock In Marsha Chechik, Joost{-}Pieter Katoen, and Martin Leucker,
  editors, {\em Formal Methods - 25th International Symposium, {FM} 2023,
  L{\"{u}}beck, Germany, March 6-10, 2023, Proceedings}, volume 14000 of {\em
  Lecture Notes in Computer Science}, pages 403--423. Springer, 2023.
\newblock \href {https://doi.org/10.1007/978-3-031-27481-7\_23}
  {\path{doi:10.1007/978-3-031-27481-7\_23}}.

\bibitem{AnthonyReplaceAll2018}
Taolue Chen, Yan Chen, Matthew Hague, Anthony~W. Lin, and Zhilin Wu.
\newblock What is decidable about string constraints with the replaceall
  function.
\newblock {\em Proc. {ACM} Program. Lang.}, 2({POPL}):3:1--3:29, 2018.
\newblock \href {https://doi.org/10.1145/3158091} {\path{doi:10.1145/3158091}}.

\bibitem{AnthonyRegex2022}
Taolue Chen, Alejandro Flores{-}Lamas, Matthew Hague, Zhilei Han, Denghang Hu,
  Shuanglong Kan, Anthony~W. Lin, Philipp R{\"{u}}mmer, and Zhilin Wu.
\newblock Solving string constraints with regex-dependent functions through
  transducers with priorities and variables.
\newblock {\em Proc. {ACM} Program. Lang.}, 6({POPL}):1--31, 2022.
\newblock \href {https://doi.org/10.1145/3498707} {\path{doi:10.1145/3498707}}.

\bibitem{AnthonyInteger2020}
Taolue Chen, Matthew Hague, Jinlong He, Denghang Hu, Anthony~Widjaja Lin,
  Philipp R{\"{u}}mmer, and Zhilin Wu.
\newblock A decision procedure for path feasibility of string manipulating
  programs with integer data type.
\newblock In Dang~Van Hung and Oleg Sokolsky, editors, {\em Automated
  Technology for Verification and Analysis - 18th International Symposium,
  {ATVA} 2020, Hanoi, Vietnam, October 19-23, 2020, Proceedings}, volume 12302
  of {\em Lecture Notes in Computer Science}, pages 325--342. Springer, 2020.
\newblock \href {https://doi.org/10.1007/978-3-030-59152-6\_18}
  {\path{doi:10.1007/978-3-030-59152-6\_18}}.

\bibitem{AnthonyComplex2019}
Taolue Chen, Matthew Hague, Anthony~W. Lin, Philipp R{\"{u}}mmer, and Zhilin
  Wu.
\newblock Decision procedures for path feasibility of string-manipulating
  programs with complex operations.
\newblock {\em Proc. {ACM} Program. Lang.}, 3({POPL}):49:1--49:30, 2019.
\newblock \href {https://doi.org/10.1145/3290362} {\path{doi:10.1145/3290362}}.

\bibitem{ChenHLRW19}
Taolue Chen, Matthew Hague, Anthony~W. Lin, Philipp R{\"{u}}mmer, and Zhilin
  Wu.
\newblock Decision procedures for path feasibility of string-manipulating
  programs with complex operations.
\newblock {\em Proc. {ACM} Program. Lang.}, 3({POPL}):49:1--49:30, 2019.
\newblock \href {https://doi.org/10.1145/3290362} {\path{doi:10.1145/3290362}}.

\bibitem{chen2023solving}
Yu-Fang Chen, David Chocholat{\'y}, Vojt{\v{e}}ch Havlena, Luk{\'a}{\v{s}}
  Hol{\'\i}k, Ond{\v{r}}ej Leng{\'a}l, and Juraj S{\'\i}{\v{c}}.
\newblock Solving string constraints with lengths by stabilization.
\newblock {\em Proceedings of the ACM on Programming Languages},
  7(OOPSLA2):2112--2141, 2023.

\bibitem{ChenCHHLS24}
Yu{-}Fang Chen, David Chocholat{\'{y}}, Vojt\v{e}ch Havlena, Luk{\'{a}}\v{s}
  Hol{\'{\i}}k, Ond\v{r}ej Leng{\'{a}}l, and Juraj S{\'{\i}}\v{c}.
\newblock {Z3-Noodler}: An automata-based string solver.
\newblock In Bernd Finkbeiner and Laura Kov{\'{a}}cs, editors, {\em Tools and
  Algorithms for the Construction and Analysis of Systems - 30th International
  Conference, {TACAS} 2024, Held as Part of the European Joint Conferences on
  Theory and Practice of Software, {ETAPS} 2024, Luxembourg City, Luxembourg,
  April 6-11, 2024, Proceedings, Part {I}}, volume 14570 of {\em Lecture Notes
  in Computer Science}, pages 24--33. Springer, 2024.
\newblock \href {https://doi.org/10.1007/978-3-031-57246-3\_2}
  {\path{doi:10.1007/978-3-031-57246-3\_2}}.

\bibitem{ChenHHHL25}
Yu{-}Fang Chen, Vojtěch Havlena, Michal Hečko, Lukáš Holík, and Ondřej
  Lengál.
\newblock A uniform framework for handling position constraints in string
  solving.
\newblock {\em Proc. {ACM} Program. Lang.}, 9(PLDI), 2025.
\newblock URL: \url{https://dx.doi.org/10.1145/3729273}, \href
  {https://doi.org/10.1145/3729273} {\path{doi:10.1145/3729273}}.

\bibitem{DayEKMNP19}
Joel~D. Day, Thorsten Ehlers, Mitja Kulczynski, Florin Manea, Dirk Nowotka, and
  Danny~B{\o}gsted Poulsen.
\newblock On solving word equations using {SAT}.
\newblock In {\em {RP}'19}, volume 11674 of {\em LNCS}, pages 93--106.
  Springer, 2019.
\newblock \href {https://doi.org/10.1007/978-3-030-30806-3\_8}
  {\path{doi:10.1007/978-3-030-30806-3\_8}}.

\bibitem{day23}
Joel~D. Day, Vijay Ganesh, Nathan Grewal, and Florin Manea.
\newblock On the expressive power of string constraints.
\newblock {\em Proc. ACM Program. Lang.}, 7(POPL), January 2023.
\newblock \href {https://doi.org/10.1145/3571203} {\path{doi:10.1145/3571203}}.

\bibitem{ganesh18}
Joel~D. Day, Vijay Ganesh, Paul He, Florin Manea, and Dirk Nowotka.
\newblock The satisfiability of extended word equations: The boundary between
  decidability and undecidability.
\newblock {\em CoRR}, abs/1802.00523, 2018.
\newblock URL: \url{http://arxiv.org/abs/1802.00523}, \href
  {https://arxiv.org/abs/1802.00523} {\path{arXiv:1802.00523}}.

\bibitem{z3}
Leonardo~Mendon{\c{c}}a de~Moura and Nikolaj~S. Bj{\o}rner.
\newblock {Z3}: An efficient {SMT} solver.
\newblock In C.~R. Ramakrishnan and Jakob Rehof, editors, {\em Tools and
  Algorithms for the Construction and Analysis of Systems, 14th International
  Conference, {TACAS} 2008, Held as Part of the Joint European Conferences on
  Theory and Practice of Software, {ETAPS} 2008, Budapest, Hungary, March
  29-April 6, 2008. Proceedings}, volume 4963 of {\em Lecture Notes in Computer
  Science}, pages 337--340. Springer, 2008.
\newblock \href {https://doi.org/10.1007/978-3-540-78800-3\_24}
  {\path{doi:10.1007/978-3-540-78800-3\_24}}.

\bibitem{karhumaki00}
Juhani Karhum{\"{a}}ki, Filippo Mignosi, and Wojciech Plandowski.
\newblock The expressibility of languages and relations by word equations.
\newblock {\em J. {ACM}}, 47(3):483--505, 2000.
\newblock \href {https://doi.org/10.1145/337244.337255}
  {\path{doi:10.1145/337244.337255}}.

\bibitem{cvc4_string14}
Tianyi Liang, Andrew Reynolds, Cesare Tinelli, Clark~W. Barrett, and Morgan
  Deters.
\newblock A {DPLL(T)} theory solver for a theory of strings and regular
  expressions.
\newblock In Armin Biere and Roderick Bloem, editors, {\em Computer Aided
  Verification - 26th International Conference, {CAV} 2014, Held as Part of the
  Vienna Summer of Logic, {VSL} 2014, Vienna, Austria, July 18-22, 2014.
  Proceedings}, volume 8559 of {\em Lecture Notes in Computer Science}, pages
  646--662. Springer, 2014.
\newblock \href {https://doi.org/10.1007/978-3-319-08867-9\_43}
  {\path{doi:10.1007/978-3-319-08867-9\_43}}.

\bibitem{tinelli-fmsd16}
Tianyi Liang, Andrew Reynolds, Nestan Tsiskaridze, Cesare Tinelli, Clark~W.
  Barrett, and Morgan Deters.
\newblock An efficient {SMT} solver for string constraints.
\newblock {\em Formal Methods Syst. Des.}, 48(3):206--234, 2016.
\newblock \href {https://doi.org/10.1007/s10703-016-0247-6}
  {\path{doi:10.1007/s10703-016-0247-6}}.

\bibitem{tinelli-frocos16}
Tianyi Liang, Nestan Tsiskaridze, Andrew Reynolds, Cesare Tinelli, and Clark~W.
  Barrett.
\newblock A decision procedure for regular membership and length constraints
  over unbounded strings.
\newblock In Carsten Lutz and Silvio Ranise, editors, {\em Frontiers of
  Combining Systems - 10th International Symposium, FroCoS 2015, Wroclaw,
  Poland, September 21-24, 2015. Proceedings}, volume 9322 of {\em Lecture
  Notes in Computer Science}, pages 135--150. Springer, 2015.
\newblock \href {https://doi.org/10.1007/978-3-319-24246-0\_9}
  {\path{doi:10.1007/978-3-319-24246-0\_9}}.

\bibitem{AnthonyTowards2016}
Anthony~Widjaja Lin and Pablo Barcel{\'{o}}.
\newblock String solving with word equations and transducers: Towards a logic
  for analysing mutation {XSS}.
\newblock In Rastislav Bod{\'{\i}}k and Rupak Majumdar, editors, {\em
  Proceedings of the 43rd Annual {ACM} {SIGPLAN-SIGACT} Symposium on Principles
  of Programming Languages, {POPL} 2016, St. Petersburg, FL, USA, January 20 -
  22, 2016}, pages 123--136. {ACM}, 2016.
\newblock \href {https://doi.org/10.1145/2837614.2837641}
  {\path{doi:10.1145/2837614.2837641}}.

\bibitem{Lothaire1997}
M.~Lothaire, editor.
\newblock {\em Combinatorics on Words}.
\newblock Cambridge Mathematical Library. Cambridge University Press, 2
  edition, 1997.

\bibitem{Lothaire02}
M.~Lothaire.
\newblock {\em Algebraic Combinatorics on Words}.
\newblock Cambridge University Press, 2002.

\bibitem{nfa2sat23}
Kevin Lotz, Amit Goel, Bruno Dutertre, Benjamin Kiesl-Reiter, Soonho Kong,
  Rupak Majumdar, and Dirk Nowotka.
\newblock Solving string constraints using sat.
\newblock In Constantin Enea and Akash Lal, editors, {\em Computer Aided
  Verification}, pages 187--208, Cham, 2023. Springer Nature Switzerland.

\bibitem{LuSJDM024}
Zhengyang Lu, Stefan Siemer, Piyush Jha, Joel~D. Day, Florin Manea, and Vijay
  Ganesh.
\newblock Layered and staged {Monte Carlo} tree search for {SMT} strategy
  synthesis.
\newblock In {\em Proceedings of the Thirty-Third International Joint
  Conference on Artificial Intelligence, {IJCAI} 2024, Jeju, South Korea,
  August 3-9, 2024}, pages 1907--1915. ijcai.org, 2024.
\newblock URL: \url{https://www.ijcai.org/proceedings/2024/211}.

\bibitem{lyndon1962}
R.~C. Lyndon and M.~P. Sch{\"u}tzenberger.
\newblock {The equation $a^M=b^Nc^P$ in a free group.}
\newblock {\em Michigan Mathematical Journal}, 9(4):289 -- 298, 1962.
\newblock \href {https://doi.org/10.1307/mmj/1028998766}
  {\path{doi:10.1307/mmj/1028998766}}.

\bibitem{makanin77}
G~S Makanin.
\newblock The problem of solvability of equations in a free semigroup.
\newblock {\em Mathematics of the USSR-Sbornik}, 32(2):129, feb 1977.
\newblock URL: \url{https://dx.doi.org/10.1070/SM1977v032n02ABEH002376}, \href
  {https://doi.org/10.1070/SM1977v032n02ABEH002376}
  {\path{doi:10.1070/SM1977v032n02ABEH002376}}.

\bibitem{nielsen1917}
Jakob Nielsen.
\newblock Die isomorphismen der allgemeinen, unendlichen gruppe mit zwei
  erzeugenden.
\newblock {\em Mathematische Annalen}, 78(1):385--397, 1917.

\bibitem{niemetz21}
Aina Niemetz, Mathias Preiner, Andrew Reynolds, Clark~W. Barrett, and Cesare
  Tinelli.
\newblock Syntax-guided quantifier instantiation.
\newblock In Jan~Friso Groote and Kim~Guldstrand Larsen, editors, {\em Tools
  and Algorithms for the Construction and Analysis of Systems - 27th
  International Conference, {TACAS} 2021, Held as Part of the European Joint
  Conferences on Theory and Practice of Software, {ETAPS} 2021, Luxembourg
  City, Luxembourg, March 27 - April 1, 2021, Proceedings, Part {II}}, volume
  12652 of {\em Lecture Notes in Computer Science}, pages 145--163. Springer,
  2021.
\newblock \href {https://doi.org/10.1007/978-3-030-72013-1\_8}
  {\path{doi:10.1007/978-3-030-72013-1\_8}}.

\bibitem{cvc422}
Andres N{\"{o}}tzli, Andrew Reynolds, Haniel Barbosa, Clark~W. Barrett, and
  Cesare Tinelli.
\newblock Even faster conflicts and lazier reductions for string solvers.
\newblock In Sharon Shoham and Yakir Vizel, editors, {\em Computer Aided
  Verification - 34th International Conference, {CAV} 2022, Haifa, Israel,
  August 7-10, 2022, Proceedings, Part {II}}, volume 13372 of {\em Lecture
  Notes in Computer Science}, pages 205--226. Springer, 2022.
\newblock \href {https://doi.org/10.1007/978-3-031-13188-2\_11}
  {\path{doi:10.1007/978-3-031-13188-2\_11}}.

\bibitem{plandowski99}
Wojciech Plandowski.
\newblock Satisfiability of word equations with constants is in {PSPACE}.
\newblock In {\em 40th Annual Symposium on Foundations of Computer Science,
  {FOCS} '99, 17-18 October, 1999, New York, NY, {USA}}, pages 495--500. {IEEE}
  Computer Society, 1999.
\newblock \href {https://doi.org/10.1109/SFFCS.1999.814622}
  {\path{doi:10.1109/SFFCS.1999.814622}}.

\bibitem{ReynoldsHighlevel19}
Andrew Reynolds, Andres N{\"{o}}tzli, Clark~W. Barrett, and Cesare Tinelli.
\newblock High-level abstractions for simplifying extended string constraints
  in {SMT}.
\newblock In Isil Dillig and Serdar Tasiran, editors, {\em Computer Aided
  Verification - 31st International Conference, {CAV} 2019, New York City, NY,
  USA, July 15-18, 2019, Proceedings, Part {II}}, volume 11562 of {\em Lecture
  Notes in Computer Science}, pages 23--42. Springer, 2019.
\newblock \href {https://doi.org/10.1007/978-3-030-25543-5\_2}
  {\path{doi:10.1007/978-3-030-25543-5\_2}}.

\bibitem{cvc420}
Andrew Reynolds, Andres N{\"{o}}tzli, Clark~W. Barrett, and Cesare Tinelli.
\newblock Reductions for strings and regular expressions revisited.
\newblock In {\em 2020 Formal Methods in Computer Aided Design, {FMCAD} 2020,
  Haifa, Israel, September 21-24, 2020}, pages 225--235. {IEEE}, 2020.
\newblock \href {https://doi.org/10.34727/2020/isbn.978-3-85448-042-6\_30}
  {\path{doi:10.34727/2020/isbn.978-3-85448-042-6\_30}}.

\bibitem{cvc417}
Andrew Reynolds, Maverick Woo, Clark~W. Barrett, David Brumley, Tianyi Liang,
  and Cesare Tinelli.
\newblock Scaling up {DPLL(T)} string solvers using context-dependent
  simplification.
\newblock In Rupak Majumdar and Viktor Kuncak, editors, {\em Computer Aided
  Verification - 29th International Conference, {CAV} 2017, Heidelberg,
  Germany, July 24-28, 2017, Proceedings, Part {II}}, volume 10427 of {\em
  Lecture Notes in Computer Science}, pages 453--474. Springer, 2017.
\newblock \href {https://doi.org/10.1007/978-3-319-63390-9\_24}
  {\path{doi:10.1007/978-3-319-63390-9\_24}}.

\bibitem{Rungta22}
Neha Rungta.
\newblock A billion {SMT} queries a day (invited paper).
\newblock In Sharon Shoham and Yakir Vizel, editors, {\em Computer Aided
  Verification - 34th International Conference, {CAV} 2022, Haifa, Israel,
  August 7-10, 2022, Proceedings, Part {I}}, volume 13371 of {\em Lecture Notes
  in Computer Science}, pages 3--18. Springer, 2022.
\newblock \href {https://doi.org/10.1007/978-3-031-13185-1\_1}
  {\path{doi:10.1007/978-3-031-13185-1\_1}}.

\bibitem{Kudzu}
Prateek Saxena, Devdatta Akhawe, Steve Hanna, Feng Mao, Stephen McCamant, and
  Dawn Song.
\newblock A symbolic execution framework for {JavaScript}.
\newblock In {\em 31st {IEEE} Symposium on Security and Privacy, S{\&}P 2010,
  16-19 May 2010, Berleley/Oakland, California, {USA}}, pages 513--528. {IEEE}
  Computer Society, 2010.
\newblock \href {https://doi.org/10.1109/SP.2010.38}
  {\path{doi:10.1109/SP.2010.38}}.

\bibitem{trauc}
Trauc string constraints benchmark collection, 2020.
\newblock URL: \url{https://github.com/plfm-iis/trauc_benchmarks}.

\bibitem{VeanesBNB17}
Margus Veanes, Nikolaj~S. Bj{\o}rner, Lev Nachmanson, and Sergey Bereg.
\newblock Monadic decomposition.
\newblock {\em J. {ACM}}, 64(2):14:1--14:28, 2017.
\newblock \href {https://doi.org/10.1145/3040488} {\path{doi:10.1145/3040488}}.

\end{thebibliography}

\ifTR
\clearpage
\appendix

\crefalias{section}{appendix}

\tikzstyle{wordBNode} = [draw, scale=0.8, inner sep=1mm, text height=3mm, minimum height=7mm, baseline]


\vspace{-0.0mm}
\section{Proofs for \cref{sec:twoSided}}\label{sec:proofsTwoSided}
\vspace{-0.0mm}

\lemNormalization*

\begin{proof}
The transformation proceeds as follows:
\begin{enumerate}
\item  First, we decompose the language~$\lang_x$ of every variable~$x \in
  \vars$ into a~finite union
  \begin{equation}
    \bigcup_{i=1}^{N_x} w^x_{i,1} \concat \K_{i,2}^{x} \concat w^x_{i,3}
    \concat \K_{i,4}^{x} \concat w^x_{i,5} \concat \cdots \concat
    \K_{i,\ell_i-1}^{x} \concat w^x_{i,\ell_i}
  \end{equation}
  where every $w^x_{i,j}$ is a word over $\Sigma$ and every $\K_{i,j}^x
  \nsubseteq \{\epsilon\}$ is either 
  \begin{inparaenum}[(i)]
    \item a regular language represented by a DFA having single initial, single final state composed of a single nontrivial SCC in the case $x$ is non-flat, or
    \item a language of the form $w^*$ where $w\in\Sigma^+$ in the case $x$ is a flat variable.
  \end{inparaenum}
  The decomposition of non-flat variables can be done by, e.g., modifying the DFA of~$\lang_x$ into an NFA (but with all SCCs deterministic) created by copying SCCs 
  such that every (non-trivial) SCC has exactly one state that serves
  as the input and one as the output port of the SCC (accordingly also for the DFA's
  final states). The nondeterminism may occurr only in transitions connecting SCCs. The non-trivial SCCs in the obtained NFA
  would correspond to the~$\K$'s and the sequences of states connecting them
  to~$w$'s.
  The union would be over all paths from an initial to a final state in the
  SCC graph of the resulting NFA.
  
  

\item  Second, for every variable $x \in \vars$, we create the disjunction
  $\bigvee_{i=1}^{N_x} \notcontains(\needle^x_i, \haystack^x_i) \land
  \langconstrof{\lang^x_i}$
  such that every occurrence of~$x$ in~$\needle$ and~$\haystack$ is
  substituted by the term $w_{i,1}^x \concat X_{i,2} \concat w_{i,3}^x
  \concat \cdots X_{i,\ell_i-1} \concat w^x_{i, \ell_i}$
  where $X_{i,2}, \ldots, X_{i, \ell_i-1}$ are fresh variables whose
  languages in~$\lang_i^x$ are set to $\K_{i,2}^{x}, \ldots, \K_{i, \ell_i-1}^{x}$
  respectively.
  We remove~$x$ from~$\vars$ and proceed recursively for the remaining
  variables, obtaining the resulting disjunction of $\notcontains$
  constraints with all variables decomposed and no finite variable.
\item  Third, any $\notcontains$ in the resulting disjunction that does not
  include a~variable is evaluated and removed if the result
  is $\mathit{false}$.
  Otherwise, $\mathit{true}$ is returned.
\end{enumerate}
It is easy to see that the result is equisatisfiable to~$\varphi$.
\end{proof}


\corInfixFineWilf*

\begin{proof}
Consider an arbitrary overlap of $u$ and $v$ of the required size.
\begin{center}
\begin{tikzpicture}
    \node[wordBNode, minimum width=6mm, anchor=west] (a0) at (0, 0)     {$\cdots$};
    \node[wordBNode, minimum width=8mm, anchor=west] (a1) at (a0.east)  {$\alpha$};
    \node[wordBNode, minimum width=8mm, anchor=west] (a2) at (a1.east)  {$\alpha$};
    \node[wordBNode, minimum width=8mm, anchor=west] (a3) at (a2.east)  {$\alpha$};
    \node[wordBNode, minimum width=8mm, anchor=west] (a4) at (a3.east)  {$\alpha$};
    \node[wordBNode, minimum width=6mm, anchor=west] (a5) at (a4.east)  {$\cdots$};

    \node[wordBNode, minimum width=6mm, anchor=west]  (b0) at (0.2, -0.7)  {$\cdots$};
    \node[wordBNode, minimum width=10mm, anchor=west] (b1) at (b0.east)  {$\beta$};
    \node[wordBNode, minimum width=10mm, anchor=west] (b2) at (b1.east)  {$\beta$};
    \node[wordBNode, minimum width=10mm, anchor=west] (b3) at (b2.east)  {$\beta$};
    \node[wordBNode, minimum width=6mm, anchor=west]  (b4) at (b3.east)  {$\cdots$};

    \draw[red, thick] ($(a1.north)+(0.1, 0.1)$) rectangle ($(b3.south)+(0.1,-0.1)$);
\end{tikzpicture}
\end{center}
The subword $w_u$ of $u$ that belongs to this overlap can be written as $s \alpha^k p$
for some $k \in \naturals$, $s \in \suffixesOf{\alpha}$ and $p \in \prefixesOf{\alpha}$.
Applying \cref{lemma:primRotations}, we can rearrange $w_u$ into $w_u = (\alpha')^k p_\alpha'$
for some $\alpha' \in \primitiveWords$ and $p_\alpha' \in \prefixesOf{\alpha'}$. Applying the
same argument to the subword $w_v$ of $v$ corresponding to the overlap, we have
$w_v = (\beta')^l p'_\beta$ for $\beta' \in \primitiveWords$ and $p'_\beta \in \prefixesOf{\beta'}$.
Assume that there is no conflict between $w_u$ and $w_v$.
Invoking \cref{lemma:fineAndWilf}, we have that $\alpha' = \gamma^O$ and $\beta' = \gamma^P$
for some $O, P \ge 1$. Since $\beta'$ is primitive and $|\alpha'| > |\beta'|$, we have
that $O \ge 2$, a contradiction with $\alpha'$ being primitive.
\end{proof}

\lemAlign*
\begin{proof}
    Let $u = p \gamma s$ and $u' = p' \gamma' s'$ be two words for $p, p', s, s' \in \Sigma^*$
    and $\gamma = \gamma'$ such that $\gamma$ has an overlap with $\gamma'$ of size
    at least $(r+1)|\alpha|$.
    First, we show that any such an overlap of $\gamma$ and $\gamma'$
    implies the existence of a decomposition of $u$ and $u'$ into
    $u = p \gamma_p B^2 \gamma_s s$ and $u' = p' \gamma_p' B \gamma_s' s'$
    where $\gamma_p, \gamma_p' \in \prefixesOf{\gamma}$, $\gamma_s, \gamma_s' \in \suffixesOf{\gamma}$,
    $B \in \{\alpha, \beta \}$ 
    with $|\gamma_p| = k |\alpha|$ and $|\gamma_p'| = k' |\alpha|$
    for some $k, k' \in \naturals$
    such that $B^2$ has an overlap with $B$ of the size $|\alpha|$.
    Graphically, we want to show that whenever we consider 
    $(r+1)|\alpha|$-sized (or larger) overlaps of $\gamma$ with itself, we can
    observe the following situation:
    \begin{center}
    \begin{tikzpicture}
        \def\bBlockWidth{10}   
        \def\dotsBlockWidth{20}  

        \node[boundaryNode, wordScalingFactor, minimum width=\dotsBlockWidth mm, anchor=west]
            (l_block0) at (0, 0)  {$\dots$};
        \node[boundaryNode, wordScalingFactor, minimum width=\bBlockWidth mm, anchor=west]
            (l_block1) at (l_block0.east)  {$B$};
        \node[boundaryNode, wordScalingFactor, minimum width=\bBlockWidth mm, anchor=west]
            (l_block1) at (l_block1.east)  {$B$};
        \node[boundaryNode, wordScalingFactor, minimum width=\dotsBlockWidth mm, anchor=west]
            (l_block2) at (l_block1.east)  {$\dots$};

        \node[boundaryNode, wordScalingFactor, minimum width=\dotsBlockWidth mm, anchor=west]
            (l_block0) at (0.5, -0.7)  {$\dots$};
        \node[boundaryNode, wordScalingFactor, minimum width=\bBlockWidth mm, anchor=west]
            (l_block1) at (l_block0.east)  {$B$};
        \node[boundaryNode, wordScalingFactor, minimum width=\dotsBlockWidth mm, anchor=west]
            (l_block1) at (l_block1.east)  {$\dots$};
    \end{tikzpicture}
    \end{center}
    
    \tikzstyle{hl} = [fill=blue!20]
    \tikzstyle{hlr} = [fill=red!20]
    \tikzstyle{hlt} = [fill=teal!20]
    We show the existence the decomposition directly,
    starting with an overlap between $\gamma$ and $\gamma'$ of the size
    $|\gamma|$ and gradually decreasing its size down to $(r+1)|\alpha|$. We
    highlight overlapping parts of $\gamma$ that witness
    that the claim holds. Starting with an overlap of the size $|\gamma|$,
    we can find an overlap, e.g., in the $\beta^{2r}$ part of the word and the
    claim holds.
    \begin{center}
    \begin{tikzpicture}
        \def\kBlockWidth{10}
        \def\kkBlockWidth{20}
        \node[boundaryNode, wordScalingFactor, minimum width=\kBlockWidth mm, anchor=west]  (l_block0) at (0, 0)               {$\alpha^{r}$};
        \node[boundaryNode, wordScalingFactor, minimum width=\kBlockWidth mm, anchor=west]  (l_block1) at (l_block0.east)      {$\beta^{r}$};
        \node[boundaryNode, wordScalingFactor, minimum width=\kBlockWidth mm, anchor=west]  (l_block1) at (l_block1.east)      {$\alpha^{r}$};
        \node[boundaryNode, wordScalingFactor, minimum width=\kBlockWidth mm, anchor=west]  (l_block2) at (l_block1.east)      {$\beta^{r}$};
        \node[boundaryNode, wordScalingFactor, minimum width=\kkBlockWidth mm, anchor=west] (l_block3) at (l_block2.east)      {$\alpha^{2r}$};
        \node[boundaryNode, hl, wordScalingFactor, minimum width=\kkBlockWidth mm, anchor=west] (l_block4) at (l_block3.east)  {$\beta^{2r}$};

        \node[boundaryNode, wordScalingFactor, minimum width=\kBlockWidth mm, anchor=west]  (l_block0) at (0, -0.7)            {$\alpha^{r}$};
        \node[boundaryNode, wordScalingFactor, minimum width=\kBlockWidth mm, anchor=west]  (l_block1) at (l_block0.east)      {$\beta^{r}$};
        \node[boundaryNode, wordScalingFactor, minimum width=\kBlockWidth mm, anchor=west]  (l_block1) at (l_block1.east)      {$\alpha^{r}$};
        \node[boundaryNode, wordScalingFactor, minimum width=\kBlockWidth mm, anchor=west]  (l_block2) at (l_block1.east)      {$\beta^{r}$};
        \node[boundaryNode, wordScalingFactor, minimum width=\kkBlockWidth mm, anchor=west] (l_block3) at (l_block2.east)      {$\alpha^{2r}$};
        \node[boundaryNode, hl, wordScalingFactor, minimum width=\kkBlockWidth mm, anchor=west] (l_block4) at (l_block3.east)  {$\beta^{2r}$};
    \end{tikzpicture}
    \end{center}
    Decreasing the overlap to the size $|\gamma| - (2r-1)|\alpha| + 1$, we have that the overlap between the two occurrences of $\beta^{2r}$
    has the size $|\alpha|-1$ and
    we no longer can find a witness between the two occurrences of the factor $\beta^{2r}$.
    \begin{center}
    \begin{tikzpicture}
        \def\kBlockWidth{10}
        \def\kkBlockWidth{20}
        \node[boundaryNode, wordScalingFactor, minimum width=\kBlockWidth mm, anchor=west]  (l_block0) at (0, 0)           {$\alpha^{r}$};
        \node[boundaryNode, wordScalingFactor, minimum width=\kBlockWidth mm, anchor=west]  (l_block1) at (l_block0.east)  {$\beta^{r}$};
        \node[boundaryNode, wordScalingFactor, minimum width=\kBlockWidth mm, anchor=west]  (l_block1) at (l_block1.east)  {$\alpha^{r}$};
        \node[boundaryNode, wordScalingFactor, minimum width=\kBlockWidth mm, anchor=west]  (l_block2) at (l_block1.east)  {$\beta^{r}$};
        \node[boundaryNode, hl, wordScalingFactor, minimum width=\kkBlockWidth mm, anchor=west] (l_block3) at (l_block2.east)  {$\alpha^{2r}$};
        \node[boundaryNode, wordScalingFactor, minimum width=\kkBlockWidth mm, anchor=west] (l_block4) at (l_block3.east)  {$\beta^{2r}$};

        \node[boundaryNode, wordScalingFactor, minimum width=\kBlockWidth mm, anchor=west]  (r_block0) at (1.3, -0.7)        {$\alpha^{r}$};
        \node[boundaryNode, wordScalingFactor, minimum width=\kBlockWidth mm, anchor=west]  (r_block1) at (r_block0.east)  {$\beta^{r}$};
        \node[boundaryNode, hl, wordScalingFactor, minimum width=\kBlockWidth mm, anchor=west]  (r_block1) at (r_block1.east)  {$\alpha^{r}$};
        \node[boundaryNode, wordScalingFactor, minimum width=\kBlockWidth mm, anchor=west]  (r_block2) at (r_block1.east)  {$\beta^{r}$};
        \node[boundaryNode, wordScalingFactor, minimum width=\kkBlockWidth mm, anchor=west] (r_block3) at (r_block2.east)  {$\alpha^{2r}$};
        \node[boundaryNode, wordScalingFactor, minimum width=\kkBlockWidth mm, anchor=west] (r_block4) at (r_block3.east)  {$\beta^{2r}$};

        \coordinate (r0) at ($(l_block4.east)$);
        \coordinate (r1) at ($(r_block4.east)+(0, 0.7)$);
        \draw (r0) edge[<->] node[below,scale=0.6,xshift=1mm] {$(2r-1)|\alpha|+1$} (r1);
    \end{tikzpicture}
    \end{center}
    If we, however, focus on $\alpha^{2r}$ and $\alpha^{r}$ highlighted in blue above, we have that their common overlap has the size
    $2r|\alpha| + 2r|\alpha| + (2r - 1)|\alpha| + 1 - r|\alpha|
    - 2r|\alpha| - 2r|\alpha| = |\alpha|(r - 1) + 1$ (computed as the
    difference of the positions of the top left and bottom right blue
    borders).
    Since $r\ge2$, we have that the the overlap has the size at least
    $|\alpha|$ and the claim holds.

    \begin{center}
    \begin{tikzpicture}
        \def\kBlockWidth{25}
        \def\kkBlockWidth{50}
        \node[boundaryNode, wordScalingFactor, minimum width=\kkBlockWidth mm, anchor=west, hl]  (a0) at (0, 0)      {$\alpha^{2r}$};

        \node[boundaryNode, hl, wordScalingFactor, minimum width=\kBlockWidth mm, anchor=west]  (b0) at (-0.8, -0.7)        {$\alpha^{r}$};
        \node[boundaryNode, wordScalingFactor, minimum width=\kBlockWidth mm, anchor=west]  (b1) at (b0.east)    {$\beta^{r}$};
        \node[boundaryNode, wordScalingFactor, minimum width=\kkBlockWidth mm, anchor=west]  (b2) at (b1.east){$\alpha^{2r}$};
        
        \coordinate (r0) at ($(a0.south east)+(0, -1)$);
        \coordinate (r1) at ($(b1.south east)+(0, -0.3)$);
        \coordinate (r2) at ($(b0.south east)+(0, -0.3)$);
        \coordinate (r3) at ($(b0.south west)+(0, -0.3)$);
        \coordinate (r4) at ($(a0.south west)+(0, -1)$);
        
        \draw (r0) edge[<->] node[below,scale=0.6] {$|\alpha|-1$} (r1);
        \draw (r1) edge[<->] node[below,scale=0.6] {$r|\alpha|$} (r2);
    \end{tikzpicture}
    \end{center}

    Decreasing the overlap size further to $|\gamma| - |\alpha|(3r + 1) -1$, we can no longer find
    $B^2$ and $B$ between $\alpha^{2r}$ and $\alpha^r$ as above since their shared overlap has the size $|\alpha| - 1$.
    However, we can then find a witness of our claim between $\beta^{2r}$ and $\beta^r$.
    \begin{center}
    \begin{tikzpicture}
        \def\kBlockWidth{10}
        \def\kkBlockWidth{20}
        \node[boundaryNode, wordScalingFactor, minimum width=\kBlockWidth mm, anchor=west]  (l_block0) at (0, 0)           {$\alpha^{r}$};
        \node[boundaryNode, wordScalingFactor, minimum width=\kBlockWidth mm, anchor=west]  (l_block1) at (l_block0.east)  {$\beta^{r}$};
        \node[boundaryNode, wordScalingFactor, minimum width=\kBlockWidth mm, anchor=west]  (l_block1) at (l_block1.east)  {$\alpha^{r}$};
        \node[boundaryNode, wordScalingFactor, minimum width=\kBlockWidth mm, anchor=west]  (l_block2) at (l_block1.east)  {$\beta^{r}$};
        \node[boundaryNode, wordScalingFactor, minimum width=\kkBlockWidth mm, anchor=west] (l_block3) at (l_block2.east)  {$\alpha^{2r}$};
        \node[boundaryNode, hl, wordScalingFactor, minimum width=\kkBlockWidth mm, anchor=west] (l_block4) at (l_block3.east)  {$\beta^{2r}$};

        \node[boundaryNode, wordScalingFactor, minimum width=\kBlockWidth mm, anchor=west]  (l_block0) at (3.1, -0.7)        {$\alpha^{r}$};
        \node[boundaryNode, wordScalingFactor, minimum width=\kBlockWidth mm, anchor=west]  (l_block1) at (l_block0.east)  {$\beta^{r}$};
        \node[boundaryNode, wordScalingFactor, minimum width=\kBlockWidth mm, anchor=west]  (l_block1) at (l_block1.east)  {$\alpha^{r}$};
        \node[boundaryNode, hl, wordScalingFactor, minimum width=\kBlockWidth mm, anchor=west]  (l_block2) at (l_block1.east)  {$\beta^{r}$};
        \node[boundaryNode, wordScalingFactor, minimum width=\kkBlockWidth mm, anchor=west] (l_block3) at (l_block2.east)  {$\alpha^{2r}$};
        \node[boundaryNode, wordScalingFactor, minimum width=4 mm, anchor=west] (l_block4) at (l_block3.east)  {$\dots$};
    \end{tikzpicture}
    \end{center}
    We continue showing the existence of $B^2$ and $B$ sharing an~overlap of the size $|\alpha|$
    by alternating between examining overlaps of $\alpha^{2r}$ with $\alpha^r$
    and $\beta^{2r}$ with $\beta^{r}$. Eventually, we reach the overlap of the size $|\alpha|(r+1)$.

    For the purpose of reaching a contradiction, assume
    that there is $u = p \gamma s$ and $u' = p' \gamma' s'$ such that
    the overlap between $\gamma$ and $\gamma'$ is conflict-free
    and it has a size at least $(r+1)|\alpha|$.
    Let us denote by $O$ and $O'$ the factors of $\gamma$ and $\gamma'$,
    respectively, located at the shared positions in the overlap.

    As both $\alpha$ and $\beta$ are primitive, we can apply
    \cref{lemma:primitiveAlignment} to the above reasoning and conclude that
    the factor $B$ overlaps trivially with itself. In other words, when we have
    $B^2$ overlapping with $B$, then the overlap must look in the following way
    in order for the overlap of $B^2$ and $B$ to contain no conflict.
    \begin{center}
    \begin{tikzpicture}
        \def\bBlockWidth{10}   
        \def\dotsBlockWidth{20}  

        \node[boundaryNode, wordScalingFactor, minimum width=\dotsBlockWidth mm, anchor=west]
            (l_block0) at (0, 0)  {$\dots$};
        \node[boundaryNode, wordScalingFactor, minimum width=\bBlockWidth mm, anchor=west]
            (l_block1) at (l_block0.east)  {$B$};
        \node[boundaryNode, wordScalingFactor, minimum width=\bBlockWidth mm, anchor=west]
            (l_block1) at (l_block1.east)  {$B$};
        \node[boundaryNode, wordScalingFactor, minimum width=\dotsBlockWidth mm, anchor=west]
            (l_block2) at (l_block1.east)  {$\dots$};

        \node[boundaryNode, wordScalingFactor, minimum width=\dotsBlockWidth mm, anchor=west]
            (l_block0) at (0.0, -0.7)  {$\dots$};
        \node[boundaryNode, wordScalingFactor, minimum width=\bBlockWidth mm, anchor=west]
            (l_block1) at (l_block0.east)  {$B$};
        \node[boundaryNode, wordScalingFactor, minimum width=\dotsBlockWidth mm, anchor=west]
            (l_block2) at (l_block1.east)  {$\dots$};
        \node[anchor=west] at ($(l_block2.east)+(1.0, 0)$) {or};

        \node[boundaryNode, wordScalingFactor, minimum width=\dotsBlockWidth mm, anchor=west]
            (l_block0) at (0.8, -1.4)  {$\dots$};
        \node[boundaryNode, wordScalingFactor, minimum width=\bBlockWidth mm, anchor=west]
            (l_block1) at (l_block0.east)  {$B$};
        \node[boundaryNode, wordScalingFactor, minimum width=\dotsBlockWidth mm, anchor=west]
            (l_block2) at (l_block1.east)  {$\dots$};
    \end{tikzpicture}
    \end{center}

    Since $|\alpha| = |\beta|$, we have $|p| = \ell |\alpha|$ and
    $p' = \ell' |\alpha|$ for some $\ell, \ell' \in \naturals$. 
    Therefore, $|O| = |O'| = n |\alpha|$ for $n \in \naturals$, and
    the shared factors have the form $O = O' = C_1 \cdots C_n$ for $C_i \in \{\alpha, \beta \}$, $1 \le i \le n$.
    We introduce the notation $|O|_{\beta} = \{ C_i \mid C_i = \beta, 1 \le i \le |O|/|\alpha| \}$.
    We show that for any non-trivial overlap of $\gamma$ and $\gamma'$ of size at least $(r+1)|\alpha|$, 
    we have $|O|_\beta \neq |O'|_\beta$ except for the case when $|O| = \frac{1}{2} |\gamma|$, which we 
    handle separately.

    Our argument is direct and it is based on counting the number of
    $\alpha$'s and $\beta$'s in the overlap. Observe that when we write
    $\gamma$ as $\gamma = C_1 \cdots C_{8r}$ where $C_i \in \{\alpha, \beta\}$,
    $1 \le i \le 8r$, we have $|\gamma|_\alpha = |\gamma|_\beta$.
    Roughly speaking, the number of $\alpha$'s and the number of $\beta$'s in $\gamma$ are the same.

    \begin{center}
    \begin{tikzpicture}
        \def\kBlockWidth{10}
        \def\kkBlockWidth{20}
        \node[boundaryNode, wordScalingFactor, minimum width=\kBlockWidth mm, anchor=west]  (l_block0) at (0, 0)           {$\alpha^{r}$};
        \node[boundaryNode, wordScalingFactor, minimum width=\kBlockWidth mm, anchor=west]  (l_block1) at (l_block0.east)  {$\beta^{r}$};
        \node[boundaryNode, wordScalingFactor, minimum width=\kBlockWidth mm, anchor=west]  (l_block1) at (l_block1.east)  {$\alpha^{r}$};
        \node[boundaryNode, wordScalingFactor, minimum width=\kBlockWidth mm, anchor=west]  (l_block2) at (l_block1.east)  {$\beta^{r}$};
        \node[boundaryNode, wordScalingFactor, minimum width=\kkBlockWidth mm, anchor=west] (l_block3) at (l_block2.east)  {$\alpha^{2r}$};
        \node[boundaryNode, wordScalingFactor, minimum width=\kkBlockWidth mm, anchor=west] (l_block4) at (l_block3.east)  {$\beta^{2r}$};

        \node[boundaryNode, wordScalingFactor, minimum width=\kBlockWidth mm, anchor=west]  (r_block0) at (0.3, -0.7)        {$\alpha^{r}$};
        \node[boundaryNode, wordScalingFactor, minimum width=\kBlockWidth mm, anchor=west]  (r_block1) at (r_block0.east)  {$\beta^{r}$};
        \node[boundaryNode, wordScalingFactor, minimum width=\kBlockWidth mm, anchor=west]  (r_block1) at (r_block1.east)  {$\alpha^{r}$};
        \node[boundaryNode, wordScalingFactor, minimum width=\kBlockWidth mm, anchor=west]  (r_block2) at (r_block1.east)  {$\beta^{r}$};
        \node[boundaryNode, wordScalingFactor, minimum width=\kkBlockWidth mm, anchor=west] (r_block3) at (r_block2.east)  {$\alpha^{2r}$};
        \node[boundaryNode, hl, wordScalingFactor, minimum width=\kkBlockWidth mm, anchor=west] (r_block4) at (r_block3.east)  {$\beta^{2r}$};

        \coordinate (r0) at ($(l_block4.east)$);
        \coordinate (r1) at ($(r_block4.east)+(0, 0.7)$);
        \draw (r0) edge[<->] node[above,scale=0.6] {$|\alpha|$} (r1);
    \end{tikzpicture}
    \end{center}
    Starting with the largest non-trivial overlap shown above, we have $|O|_{\beta} > |O'|_{\beta}$.
    As we decrease the size of the overlap by multiples of $|\alpha|$, we see that $|O|_{\beta} > |O'|_{\beta}$
    until we reach an overlap of size $|O| = \frac{1}{2}|\gamma|$ in which case we have $O \neq O'$.
    \begin{center}
    \begin{tikzpicture}
        \def\kBlockWidth{10}
        \def\kkBlockWidth{20}
        \node[boundaryNode, wordScalingFactor, minimum width=\kBlockWidth mm, anchor=west]  (l_block0) at (0, 0)           {$\alpha^{r}$};
        \node[boundaryNode, wordScalingFactor, minimum width=\kBlockWidth mm, anchor=west]  (l_block1) at (l_block0.east)  {$\beta^{r}$};
        \node[boundaryNode, wordScalingFactor, minimum width=\kBlockWidth mm, anchor=west]  (l_block1) at (l_block1.east)  {$\alpha^{r}$};
        \node[boundaryNode, wordScalingFactor, minimum width=\kBlockWidth mm, anchor=west]  (l_block2) at (l_block1.east)  {$\beta^{r}$};
        \node[boundaryNode, wordScalingFactor, minimum width=\kkBlockWidth mm, anchor=west] (l_block3) at (l_block2.east)  {$\alpha^{2r}$};
        \node[boundaryNode, wordScalingFactor, minimum width=\kkBlockWidth mm, anchor=west] (l_block4) at (l_block3.east)  {$\beta^{2r}$};

        \node[boundaryNode, wordScalingFactor, minimum width=\kBlockWidth mm, anchor=west]  (r_block0) at (3.25, -0.7)        {$\alpha^{r}$};
        \node[boundaryNode, wordScalingFactor, minimum width=\kBlockWidth mm, anchor=west]  (r_block1) at (r_block0.east)  {$\beta^{r}$};
        \node[boundaryNode, wordScalingFactor, minimum width=\kBlockWidth mm, anchor=west]  (r_block1) at (r_block1.east)  {$\alpha^{r}$};
        \node[boundaryNode, wordScalingFactor, minimum width=\kBlockWidth mm, anchor=west]  (r_block2) at (r_block1.east)  {$\beta^{r}$};
        \node[boundaryNode, wordScalingFactor, minimum width=\kkBlockWidth mm, anchor=west] (r_block3) at (r_block2.east)  {$\alpha^{2r}$};
        \node[boundaryNode, wordScalingFactor, minimum width=\kkBlockWidth mm, anchor=west] (r_block4) at (r_block3.east)  {$\beta^{2r}$};
    \end{tikzpicture}
    \end{center}
    The rest of the argument proof follows by symmetry as if there is an
    overlap such that $(r+1)|\alpha| \le |O| < \frac{1}{2}|\gamma|$ with
    $|O|_\beta = |O'|_\beta$, then the factors of $\gamma$ and $\gamma'$ that are
    not part of the overlap would also contain the same number of $\beta$'s.
\end{proof}

\vspace{-0.0mm}
\section{Proofs for \cref{sec:tools}}\label{sec:proofsTools}
\vspace{-0.0mm}

\lemGammaUseful*

\begin{proof}
First, we show that any overlap between $\sigma(z)$ and $\sigma(\needle)$ of the size at least $N$
contains an overlap of the size at least~$2 M_\alpha$ between~$\gamma^K_z$ and~$\sigma(x)$ for some flat variable $x \in \flatVars$.

There are two cases to consider based on whether there is a~literal $W$ in $\needle$
surrounded by variables sharing an overlap with $\gamma^K_z$ of the
size~$|W|$.
If there is such a literal, then let us decompose $\needle$ into $\needle = w x W y w'$ for some
words $w, w' \in (\vars \cup \alphabet)^*$ and $x, y \in \flatVars$.
The size of the overlap between $\sigma(x)$ and $\gamma^K_z$ 
together with the size of the overlap between~$\sigma(y)$ and~$\gamma^K_z$ 
is
\begin{center}
\begin{tikzpicture}
    \def\alphaXWidth{8}
    \def\alphaYWidth{8}

    \node[wordBNode, anchor=west, minimum width=11mm] (a0)  at (0, 0)    {$\dots$};
    \node[wordBNode, anchor=west, minimum width=8mm] (a1)  at (a0.east) {$\gamma_z$};
    \node[wordBNode, anchor=west, minimum width=8mm] (a2)  at (a1.east) {$\gamma_z$};
    \node[wordBNode, anchor=west, minimum width=8mm] (a3)  at (a2.east) {$\gamma_z$};
    \node[wordBNode, anchor=west, minimum width=10mm] (a4)  at (a3.east) {$\dots$};
    \node[wordBNode, anchor=west, minimum width=8mm] (a5)  at (a4.east) {$\gamma_z$};
    \node[wordBNode, anchor=west, minimum width=8mm] (a6)  at (a5.east) {$\gamma_z$};
    \node[wordBNode, anchor=west, minimum width=8mm] (a7)  at (a6.east) {$\gamma_z$};
    \node[wordBNode, anchor=west, minimum width=11mm] (a8)  at (a7.east) {$\dots$};

    \node[fill=white, wordBNode, anchor=west, minimum width=11 mm] (b0) at (-0.7, -0.8) {$ \dots $};
    \node[fill=white, wordBNode, anchor=west, minimum width=\alphaXWidth mm] (b1) at (b0.east) {$ \alpha_x $};
    \node[fill=white, wordBNode, anchor=west, minimum width=\alphaXWidth mm] (b2) at (b1.east) {$ \alpha_x $};
    \node[fill=white, wordBNode, anchor=west, minimum width=\alphaXWidth mm] (b3) at (b2.east) {$ \alpha_x $};
    \node[fill=white, wordBNode, anchor=west, minimum width=20 mm]           (b4) at (b3.east) {$ W $};
    \node[fill=white, wordBNode, anchor=west, minimum width=\alphaYWidth mm] (b5) at (b4.east) {$ \alpha_y $};
    \node[fill=white, wordBNode, anchor=west, minimum width=\alphaYWidth mm] (b6) at (b5.east) {$ \alpha_y $};
    \node[fill=white, wordBNode, anchor=west, minimum width=\alphaYWidth mm] (b7) at (b6.east) {$ \alpha_y $};
    \node[fill=white, wordBNode, anchor=west, minimum width=11 mm] (b8) at (b7.east) {$ \dots $};
    
    \coordinate (c0) at ($(a0.south east)+(0, -1.3 )$);
    \coordinate (c1) at ($(a8.south west)+(0, -1.3 )$);
    \draw[teal] (c0) edge[<->] node[below,scale=0.8,yshift=1mm] {$K|\gamma_z|$} (c1);

    \coordinate (c2) at ($(b4.south west)+(0, -0.35)$);
    \coordinate (c3) at ($(b4.south east)+(0, -0.35)$);
    \draw[blue] (c2) edge[<->] node[above,scale=0.8,yshift=-1mm] {$|W|$} (c3);
    
    \begin{pgfonlayer}{bg}
        \draw[dotted, thick, blue] (b4.south west) -- (c2);
        \draw[dotted, thick, blue] (b4.south east) -- (c3);

        \draw[dotted, thick, teal] (a0.south east) -- (c0);
        \draw[dotted, thick, teal] (a8.south west) -- (c1);

        \draw[very thick, gray] ($(a1.north west)+(-0.01, 0.01)$) rectangle ($(a7.south east)+( 0.01,-0.01)$);
        \node[anchor=south, scale=0.8] at (a4.north) {$\gamma^K_z$};
    \end{pgfonlayer}
\end{tikzpicture} 
\end{center}
\begin{equation}
    K|\gamma_z| - |W| \ge 4 M_\alpha + 2 M_{\mathrm{lit}} - |W| > 4 M_\alpha \, .
\end{equation}
Therefore, either $\sigma(x)$ or $\sigma(y)$ have an overlap of the size~$2M_\alpha$ with $\gamma^K_z$.

Alternatively, none of the literals surrounded by variables overlaps with $\gamma^K_v$.
Due to the choice of~$N$, we see that there must be a variable $y \in \flatVars$ in $\needle$
such that $\sigma(y)$ overlaps with $\gamma^K_v$. If we denote the size of the maximal overlap between $\sigma(y)$
and $\gamma^K_z$ by $O$, we have roughly the following situation:

\begin{center}
\tikzstyle{wordBNode} = [draw, scale=0.8, inner sep=1mm, text height=3mm, minimum height=7mm, baseline]
\begin{tikzpicture}
    \node[wordBNode, anchor=west, minimum width=8mm] (a0)  at (0, 0)    {$p_z$};
    \node[wordBNode, anchor=west, minimum width=6mm] (a1)  at (a0.east) {$U$};
    \node[wordBNode, anchor=west, minimum width=8mm] (a2)  at (a1.east) {$\gamma_z$};
    \node[wordBNode, anchor=west, minimum width=8mm] (a3)  at (a2.east) {$\gamma_z$};
    \node[wordBNode, anchor=west, minimum width=8mm] (a4)  at (a3.east) {$\gamma_z$};
    \node[wordBNode, anchor=west, minimum width=8mm] (a5)  at (a4.east) {$\gamma_z$};
    \node[wordBNode, anchor=west, minimum width=8mm] (a6)  at (a5.east) {$\gamma_z$};
    \node[wordBNode, anchor=west, minimum width=8mm] (a7)  at (a6.east) {$\dots$};

    \node[fill=white, wordBNode, anchor=west, minimum width=5mm]  (b0) at (-1.5, -0.8) {$ \dots $};
    \node[fill=white, wordBNode, anchor=west, minimum width=7mm]  (b1) at (b0.east) {$ \alpha_x $};
    \node[fill=white, wordBNode, anchor=west, minimum width=15mm] (b2) at (b1.east) {$ W_0 $};
    \node[fill=white, wordBNode, anchor=west, minimum width=9mm]  (b3) at (b2.east) {$ \alpha_y $};
    \node[fill=white, wordBNode, anchor=west, minimum width=9mm]  (b4) at (b3.east) {$ \alpha_y $};
    \node[fill=white, wordBNode, anchor=west, minimum width=9mm]  (b5) at (b4.east) {$ \alpha_y $};
    \node[fill=white, wordBNode, anchor=west, minimum width=12mm] (b6) at (b5.east) {$ W_1 $};
    
    \coordinate (c0) at ($(a0.south west)+(0, -1.1)$);
    \coordinate (c1) at ($(a1.south east)+(0, -1.1)$);
    \coordinate (c2) at ($(b5.south east)+(0, -0.3)$);
    \coordinate (c3) at ($(b6.south east)+(0, -0.3)$);
    
    \draw[blue] (c1) edge[<->] node[below] {$O$} (c2);
    
    \begin{pgfonlayer}{bg}
        \draw[dotted, thick, blue] (a1.south east) -- (c1);
        \draw[dotted, thick, blue] (b5.south east) -- (c2);
    \end{pgfonlayer}
\end{tikzpicture} 
\end{center}
for $U = \con(q_0, q_{u|v})$ and
some literals $W_0, W_1 \in \alphabet^*$, or we have a symmetrical situation with a prefix of $\sigma(\needle)$.

Since $|U| \le |Q|$, we have 
\begin{equation}
O \ge N - |p| - |Q| - |W_0| - |W_1|\,.
\end{equation}
Substituting for $N$ its definition and noting that $\pMaxLit \ge |W_0|$ and $\pMaxLit \ge |W_1|$,
we have $O \ge 2\pMaxPrim$.

We have shown that any overlap of the size at least~$N$ contains an overlap between $\gamma^K_z$
and $\sigma(x)$ for some flat variable $x \in \flatVars$.
Using \cref{cor:infixFineWilf},
we have that this overlap contains a conflict since $|\gamma_z| > |\alpha|$ for any
$\alpha \in \primBase(\needle)$.
\end{proof}

\lemSewing*

\begin{proof}
    For the sake of contradiction assume that $\sigma \not \models \varphi$.
    Since $\sigma^\pref$, $\sigma^\suf$, and $\sigma$ agree on all variables except $z$,
    it must be that there is an overlap between $\sigma(\needle)$ and $\sigma(z)$ containing
    no conflicts. Since $\sigma^\pref \models \varphi^\pref$ and $\sigma^\suf \models \varphi^\suf$,
    it must be that this overlap includes the word $\gamma^K$. Note that $\needle$ contains only flat
    variables, and we have $|\sigma(x)| > 2 \pMaxPrim$ for any flat variable $x \in \flatVars$.
    Since, $K |\gamma_z| \ge 4\pMaxPrim + 2\pMaxLit$, we have that the overlap
    between $\sigma(z)$ and $\sigma(\needle)$ contains an overlap between $\gamma_z^K$
    and $\sigma(x)$ of the size at least $2 \pMaxPrim$ for some flat variable $x \in \flatVars$.
    Invoking \cref{cor:infixFineWilf}, we have that the overlap contains a conflict. Contradiction.
\end{proof}

\vspace{-0.0mm}
\section{Proofs for \cref{sec:singleSided}}\label{sec:proofsSingleSided}
\vspace{-0.0mm}

\lemConflictKaboom*

\begin{proof}
We will take a direct approach, showing conflicts between $\haystack$ and $p' \needle$
while incrementally increasing the length of the prefix $p' \in \prefixesOf{\haystack}$.
Condition~\ref{cond1} allows us to start with $|p'| = |p| - |\needle| + 1$ where $p$
is the prefix from Condition~\ref{cond1}.
\begin{center}
\begin{tikzpicture}
    \node[wordBNode, minimum width=6mm, anchor=west] (a0) at (0, 0)     {$\cdots$};
    \node[wordBNode, minimum width=10mm, anchor=west] (a1) at (a0.east)  {$\alpha$};
    \node[wordBNode, minimum width=16mm, anchor=west] (a2) at (a1.east) {$v_0$};
    \node[wordBNode, minimum width=6mm, anchor=west] (a3) at (a2.east)  {$\cdots$};

    \node[wordBNode, minimum width=6mm, anchor=west]  (b0) at (-0.7, -0.7){$\cdots$};
    \node[wordBNode, minimum width=10mm, anchor=west] (b1) at (b0.east)  {$\alpha$};
    \node[wordBNode, minimum width=10mm, anchor=west] (b2) at (b1.east)  {$\alpha$};
    \node[wordBNode, minimum width=13mm, anchor=west]  (b3) at (b2.east)  {$w$};
\end{tikzpicture}
\end{center}
Hence, we have $\alpha$ in $\haystack$ overlapping with a block $\alpha^2$ in $\needle$.
Invoking \cref{lemma:primitiveAlignment}, we have that for this choice of $p$, we have a conflict.

The first situation when we cannot apply \cref{lemma:primitiveAlignment} to show the existence
of a conflict is when $|p'| = |p| - |\needle| + |\alpha|$. But in such a case we have
\begin{center}
\begin{tikzpicture}
    \node[wordBNode, minimum width=6mm, anchor=west] (a0) at (0, 0)     {$\cdots$};
    \node[wordBNode, minimum width=10mm, anchor=west] (a1) at (a0.east)  {$\alpha$};
    \node[wordBNode, minimum width=16mm, anchor=west] (a2) at (a1.east) {$v_0$};
    \node[wordBNode, minimum width=6mm, anchor=west] (a3) at (a2.east)  {$\cdots$};

    \node[wordBNode, minimum width=6mm, anchor=west]  (b0) at (-0.0, -0.7){$\cdots$};
    \node[wordBNode, minimum width=10mm, anchor=west] (b1) at (b0.east)  {$\alpha$};
    \node[wordBNode, minimum width=10mm, anchor=west] (b2) at (b1.east)  {$\alpha$};
    \node[wordBNode, minimum width=13mm, anchor=west]  (b3) at (b2.east)  {$w$};
\end{tikzpicture}
\end{center}
and thanks to Condition~\ref{cond2}, we have that there is a conflict between $v_0$ and $\alpha$
from $\needle$.

Inspecting different choices of $p'$, we alternate between applying \cref{lemma:primitiveAlignment}
and Condition~\ref{cond2}, until we exhaust all possible $p'$s or we can no longer apply any of these
two conditions. In such a case, we have
\begin{center}
\begin{tikzpicture}
    \node[wordBNode, minimum width=6mm, anchor=west] (a0) at (0, 0)     {$\cdots$};
    \node[wordBNode, minimum width=10mm, anchor=west] (a1) at (a0.east)  {$\alpha$};
    \node[wordBNode, minimum width=16mm, anchor=west] (a2) at (a1.east) {$v_0$};
    \node[wordBNode, minimum width=9mm, anchor=west] (a3) at (a2.east)  {$\gamma$};
    \node[wordBNode, minimum width=9mm, anchor=west] (a4) at (a3.east)  {$\gamma$};
    \node[wordBNode, minimum width=9mm, anchor=west] (a5) at (a4.east)  {$\gamma$};
    \node[wordBNode, minimum width=9mm, anchor=west] (a6) at (a5.east)  {$\dots$};

    \node[wordBNode, minimum width=13mm, anchor=west]  (b0) at (-0.4, -0.7){$u_1$};
    \node[wordBNode, minimum width=10mm, anchor=west] (b1) at (b0.east)  {$\alpha$};
    \node[wordBNode, minimum width=10mm, anchor=west] (b2) at (b1.east)  {$\dots$};
    \node[wordBNode, minimum width=10mm, anchor=west]  (b3) at (b2.east)  {$\alpha$};
    \node[wordBNode, minimum width=10mm, anchor=west]  (b4) at (b3.east)  {$\alpha$};
    \node[wordBNode, minimum width=10mm, anchor=west]  (b5) at (b4.east)  {$\alpha$};
    \node[wordBNode, minimum width=13mm, anchor=west]  (b6) at (b5.east)  {$w$};

    \draw[red, thick] ($(a3.north west)+( 0.05, 0.1)$) rectangle ($(b5.south)+( 0.2,-0.1)$);
\end{tikzpicture}
\end{center}
Due to the condition $|v_0| < \boundFlat - 2\max(|\alpha|, |\gamma|)$ together with $\boundFlat < N|\alpha|$,
we have that there is an overlap between $\gamma^{\boundGamma}$ and $\alpha^{\boundFlat}$
of size at least $2 \max(|\alpha|, |\gamma|) > |\alpha| + |\gamma|$. Therefore,
we invoke \cref{cor:infixFineWilf} and see that such an overlap contains a conflict.
Taking a~longer $p'$ only increases the size of the overlap between $\gamma^{\boundGamma}$ and $\alpha^{\boundFlat}$,
and, thus, we obtain the lemma.
\end{proof}

\lemFlatify*


\begin{proof}
Applying the definition of $\gammaPref{z}^{\boundGamma}$, we have
    $\gammaPref{z}^{\boundGamma}(s_\alpha \alpha^k p_\alpha W) = s_\alpha \alpha^k p_\alpha W \concat U_z \gamma^{\boundGamma}$.
As $\sigma(z)$ and $\sigma'(z)$ share a common prefix, we know that if
$\sigma'$ fails to be a model then its due to the suffix $U_z \gamma^{\boundGamma}$ introduced by $\Gamma$-expansion
(otherwise we reach a contradiction with $\sigma$ being a model of $\varphi$). Therefore,
we have Condition~\ref{cond1} of \cref{lemma:conflictKaboom} satisfied. 
The requirements
on $p_\alpha$ being maximal set forth by the lemma  imply $p_\alpha a_0 \not \in \prefixesOf{\alpha}$
where $a_0$ is the first letter of $W$. Hence we have $p_\alpha W \not \in \prefixesOf{\alpha}$
and $\alpha \not \in \prefixesOf{p_\alpha W}$, and we also have Condition~\ref{cond2}
of \cref{lemma:conflictKaboom} satisfied. Invoking \cref{lemma:conflictKaboom}, we obtain the proof.
\end{proof}

\lemConstrFlat*

\begin{proof}
    Let $\varphi' \triangleq \varphi[z/z\sep]$, and let $T_z$ be the prefix tree for variable $z$.
    Starting from the root, we explore paths in $T_z$ in a breadth-first fashion up the bound $\boundAut$.
    For any path $\pi$, we first check whether $\sigma \variant \{ z \mapsto \edgeLabelingFn(\pi)\}
    \models \varphi'$. If the modified assignment fails to be a model of $\varphi'$,
    we mark the last vertex $s_n$ of $\pi$
    as a dead end, and we do not explore any longer paths passing through $s_n$. Let $\shortPathsSet$
    be the set of all paths explored in this process. Clearly, if there are no $\boundAut$-reaching
    paths in $\shortPathsSet$, then $\sigma(z) \in \shortPrefixes$ where
    $\shortPrefixes \triangleq \{ w \in \lang_z \mid |w| \le \boundAut \} \cap \lang_z$.
    
    Alternatively, there are some $\boundAut$-reaching paths in
    $\shortPathsSet$. Let $\boundReachingPaths \subseteq \shortPathsSet$
    be the set of $\boundAut$-reaching paths in $\shortPathsSet$.
    The question is what if for every path
    $\pi \in \boundReachingPaths$, we have that the altered assignment
    $\sigma_{w_\pi} \triangleq \sigma \variant \{z \mapsto \gammaPref{z}^{\boundGamma}(\edgeLabelingFn(z))\} $ fails to be a model of $\varphi'$.

    Let $\pi \in \boundReachingPaths$ be such a path, and we investigate why
    $\sigma_{w_\pi}$ fails to be a~model of $\varphi'$. Since $\piSigma$ is not
    a model, we have that $\sigma_{w_{\pi}}(\haystack)$ must contain an
    occurrence of the word $\piSigma(\needle)$. Furthermore, since $\sigma$ and
    $\piSigma$ agree on the valuation of all variables except for $z$, we have
    that the occurrence of the word $\piSigma(\needle)$
    contains a part of $\piSigma(z)$. Due to the way in which we explore $T_z$, there are no dead-end vertices
    in $\pi$, and, hence, we know that any overlap in which the last letter of $\piSigma(\needle)$ occurs
    before the last letter of $\edgeLabelingFn(\pi)$ contains a~conflict, i.e., the following is impossible:
    \begin{center}
    \begin{tikzpicture}
        \node[wordBNode, anchor=west, minimum width=8mm, draw=gray] (a0) at (0, 0) {$\dots$};
        \node[wordBNode, anchor=west, minimum width=42mm]      (a1) at (a0.east) {$\edgeLabelingFn(\pi)$};
        \node[wordBNode, anchor=west, minimum width=8mm, draw=gray] (a2) at (a1.east) {$\dots$};

        \node[wordBNode, anchor=west, minimum width=8mm, draw=gray] (b0) at (-0.7, -0.7)    {$ \dots  $};
        \node[wordBNode, anchor=west, minimum width=12mm]      (b1) at (b0.east) {$ \alpha  $};
        \node[wordBNode, anchor=west, minimum width=12mm]      (b2) at (b1.east) {$ \alpha  $};
        \node[wordBNode, anchor=west, minimum width=5mm]      (b3) at (b2.east) {$ p_\alpha  $};
        \node[wordBNode, anchor=west, minimum width=18mm]      (b4) at (b3.east) {$ W $};
    \end{tikzpicture}
    \end{center}
    Therefore, we know that the last letter of $W$ must be located somewhere in the suffix introduced
    by applying $\gammaPref{z}^{\boundGamma}$. Recall, that we have fixed values of all variables in $\needle$ that are short,
    and therefore, the remaining variables have assigned words longer than
    $\boundFlat$. Thanks to the properties of $\Gamma$-expansion (\cref{lemma:gammaUseful}), we have
    that the last letter of $\piSigma(\needle)$ cannot be located too far off to the right, as we would
    have the following.
    \begin{center}
    \begin{tikzpicture}
        \node[wordBNode, anchor=west, minimum width=8mm, draw=gray] (a0) at (0, 0) {$\dots$};
        \node[wordBNode, anchor=west, minimum width=35mm] (a1) at (a0.east) {$\edgeLabelingFn(\pi)$};
        \node[wordBNode, anchor=west, minimum width=8mm] (a2) at (a1.east) {$U_z$};
        \node[wordBNode, anchor=west, minimum width=12mm] (a3) at (a2.east) {$\gamma_z$};
        \node[wordBNode, anchor=west, minimum width=12mm] (a4) at (a3.east) {$\gamma_z$};
        \node[wordBNode, anchor=west, minimum width=12mm] (a5) at (a4.east) {$\gamma_z$};
        \node[wordBNode, anchor=west, minimum width=8mm, draw=gray] (a6) at (a5.east) {$\dots$};

        \node[wordBNode, anchor=west, minimum width=20mm, draw=gray] (r0) at (1.2, -0.7)    {$ \dots  $};
        \node[wordBNode, anchor=west, minimum width=12mm]       (r1) at (r0.east) {$ \alpha  $};
        \node[wordBNode, anchor=west, minimum width=12mm]       (r2) at (r1.east) {$ \alpha  $};
        \node[wordBNode, anchor=west, minimum width=12mm]       (r3) at (r2.east) {$ \alpha  $};
        \node[wordBNode, anchor=west, minimum width=12mm]       (r4) at (r3.east) {$ \alpha  $};
        \node[wordBNode, anchor=west, minimum width=5mm]        (r5) at (r4.east) {$ p_\alpha  $};
        \node[wordBNode, anchor=west, minimum width=18mm]       (r6) at (r5.east) {$ W $};
    \end{tikzpicture}
    \end{center}
    Invoking \cref{cor:infixFineWilf}, we see that such an overlap contains a~conflict.

    Perhaps more importantly, we have chosen $\boundAut$ such that $\boundFlat
    - \boundAut > \pMaxLit$ (cf. \cref{eq:boundAutDef}), i.e., the word
    $\edgeLabelingFn(\pi)$ is longer than any literal $w \in \alphabet^*$ in
    $\varphi$. As the last letter of $W$ must be located somewhere after the
    last letter of $\edgeLabelingFn(\pi)$, $\edgeLabelingFn(\pi)$ starts with a
    prefix $p \in \factorsOf{\alpha^*}$ of the length $|p| = \boundAut - |W|$.
    Graphically,
    we have the following.
    \begin{center}
    \begin{tikzpicture}
        \node[wordBNode, anchor=west, minimum width=8mm, draw=gray] (a0) at (0, 0) {$\dots$};
        \node[wordBNode, anchor=west, minimum width=42mm] (a1) at (a0.east) {$\edgeLabelingFn(\pi)$};
        \node[wordBNode, anchor=west, minimum width=8mm]  (a2) at (a1.east) {$U_z$};
        \node[wordBNode, anchor=west, minimum width=12mm] (a3) at (a2.east) {$\gamma_z$};
        \node[wordBNode, anchor=west, minimum width=12mm] (a4) at (a3.east) {$\gamma_z$};
        \node[wordBNode, anchor=west, minimum width=8mm, draw=gray] (a5) at (a4.east) {$\dots$};

        \node[wordBNode, anchor=west, minimum width=5mm, draw=gray] (r0) at (-0.2, -0.7)    {$ \dots  $};
        \node[wordBNode, anchor=west, minimum width=12mm]       (r1) at (r0.east) {$ \alpha  $};
        \node[wordBNode, anchor=west, minimum width=12mm]       (r2) at (r1.east) {$ \alpha  $};
        \node[wordBNode, anchor=west, minimum width=12mm]       (r3) at (r2.east) {$ \alpha  $};
        \node[wordBNode, anchor=west, minimum width=5mm]        (r4) at (r3.east) {$ p_\alpha  $};
        \node[wordBNode, anchor=west, minimum width=18mm]       (r5) at (r4.east) {$ W $};

        \coordinate (botLeft) at ($(a1.south west)+(0, -1.1)$);
        \coordinate (botRight) at ($(r4.south east)+(0, -0.4)$);

        \draw[dashed, blue, thick] (a1.south west) -- (botLeft);
        \draw[dashed, blue, thick] ($(r4.south east)+(0, 0.7)$) -- (botRight);
        \draw (botLeft) edge[<->, blue] node[below] {$p$} (botRight);
    \end{tikzpicture}
    \vspace*{-5mm}
    \end{center}
    As $\pi$ was an arbitrary $\boundAut$-reaching path,
    we have that every $\pi_i \in \boundReachingPaths$ is labeled with a prefix $p_i \in \factorsOf{\alpha^*}$
    of size $|p_i| = \boundAut - |W|$, otherwise we would reach a contradiction with there being
    no $\boundAut$-reaching paths that can be $\Gamma$-expanded into a model.

    In fact, we show that all $\boundAut$-reaching paths share the same prefix $p$. To see this,
    we first decompose $\haystack$ into 
    $\haystack = \haystack_0 \concat z_1 \concat \cdots \concat z_n \concat \haystack_n$
    for $\haystack_i \in (\vars \setminus \{z\} \cup \alphabet)^*$ for any $0 \le i \le n$, and where
        the subscript $z_j$ denotes $j$-th occurrence of the variable $z$. Recall that $\piSigma \not \models \varphi'$, and thus, there is an occurrence $z_j$ 
    that contains a part of $\piSigma(\needle)$. As we have $\boundFlat - \boundAut \ge 2\pMaxLit$,
    we know that a suffix $s$ of $\sigma'(\haystack_{j-1})$ of size $|s| = 2|\alpha|$ contains
    $\alpha$ as a~factor. Therefore, we have the following.

    \begin{center}
    \begin{tikzpicture}
        \node[wordBNode, anchor=west, minimum width=8mm, draw=gray] (a0) at (0, 0) {$\dots$};
        \node[wordBNode, anchor=west, minimum width=20mm] (a1) at (a0.east) {$s$};
        \node[wordBNode, anchor=west, minimum width=42mm] (a2) at (a1.east) {$\edgeLabelingFn(\pi)$};
        \node[wordBNode, anchor=west, minimum width=8mm]  (a3) at (a2.east) {$U_z$};
        \node[wordBNode, anchor=west, minimum width=12mm] (a4) at (a3.east) {$\gamma_z$};
        \node[wordBNode, anchor=west, minimum width=8mm, draw=gray] (a5) at (a4.east) {$\dots$};

        \node[wordBNode, anchor=west, minimum width=5mm, draw=gray] (r0) at ( 0.6, -0.7)    {$ \dots $};
        \node[wordBNode, anchor=west, minimum width=12mm, draw=red, thick]       (r1) at (r0.east) {$ \alpha  $};
        \node[wordBNode, anchor=west, minimum width=12mm]       (r2) at (r1.east) {$ \alpha  $};
        \node[wordBNode, anchor=west, minimum width=9mm]       (r3) at (r2.east) {$ \dots  $};
        \node[wordBNode, anchor=west, minimum width=12mm]       (r4) at (r3.east) {$ \alpha  $};
        \node[wordBNode, anchor=west, minimum width=5mm]        (r5) at (r4.east) {$ p_\alpha  $};
        \node[wordBNode, anchor=west, minimum width=18mm]       (r6) at (r5.east) {$ W $};
    \end{tikzpicture}
    \end{center}
    
    Therefore, if we have two $\boundAut$-reaching paths $\pi$ and $\pi'$ labeled
    with words starting with maximal $s, s' \in \suffixesOf{\alpha}$, respectively, such that $s \neq s'$,
    while both $\piSigma$ and $\sigma_{w_{\pi'}}$ are failing to be models, we would have w.l.o.g. the following
    \begin{center}
    \begin{tikzpicture}
        \def\alphaWidth{12}

        \node[wordBNode, anchor=west, minimum width=8mm, draw=gray] (a0) at (0, 0) {$\dots$};
        \node[wordBNode, anchor=west, minimum width=20mm] (a1) at (a0.east) {$ $};
        \node[wordBNode, anchor=west, minimum width=42mm] (a2) at (a1.east) {$ $};
        \node[wordBNode, anchor=west, minimum width=8mm]  (a3) at (a2.east) {$U_z$};
        \node[wordBNode, anchor=west, minimum width=12mm] (a4) at (a3.east) {$\gamma_z$};
        \node[wordBNode, anchor=west, minimum width=8mm, draw=gray] (a5) at (a4.east) {$\dots$};

        \node[wordBNode, anchor=west, minimum width=5mm, draw=gray] (r0) at ( 0.6, -0.7)    {$ \dots $};
        \node[wordBNode, anchor=west, minimum width=\alphaWidth mm]       (r1) at (r0.east) {$ \alpha  $};
        \node[wordBNode, anchor=west, minimum width=\alphaWidth mm]       (r2) at (r1.east) {$ \alpha  $};
        \node[wordBNode, anchor=west, minimum width=\alphaWidth mm]       (r3) at (r2.east) {$ \alpha  $};
        \node[wordBNode, anchor=west, minimum width=9mm]       (r4) at (r3.east) {$ \dots  $};
        \node[wordBNode, anchor=west, minimum width=5mm]        (r5) at (r4.east) {$ p_\alpha  $};
        \node[wordBNode, anchor=west, minimum width=18mm]       (r6) at (r5.east) {$ W $};

        \node[wordBNode, anchor=west, minimum width=10mm] (s) at (a2.west) {$s_\alpha$}; 
        \node[wordBNode, anchor=west, minimum width=\alphaWidth mm] (i0) at (s.east) {$\alpha$}; 
        \node[wordBNode, anchor=west, minimum width=\alphaWidth mm] (i1) at (i0.east) {$\alpha$}; 

        \node[wordBNode, anchor=west, minimum width=\alphaWidth mm] at ($(a1.west)+(0.5, 0)$) {$\alpha$}; 

        \draw[draw=black, thick] (a1.north west) rectangle (a1.south east);
        \draw[draw=red, thick] (a2.north west) rectangle (a2.south east);

        \node[scale=0.8, red] at ($(a2.north)+(0, 0.2)$) {$\edgeLabelingFn(\pi)$};
        \node[scale=0.8] at ($(a1.north)+(0, 0.2)$) {$s$};

        \coordinate (r1_mid) at ($(a0.south)!0.5!(r0.north)$);
        \node[red] at ($(r1_mid)+(-1.0, 0)$) {$\pi$};

        \node[wordBNode, anchor=west, minimum width=8mm, draw=gray] (a0) at (0, -2) {$\dots$};
        \node[wordBNode, anchor=west, minimum width=20mm] (a1) at (a0.east) {$ $};
        \node[wordBNode, anchor=west, minimum width=42mm] (a2) at (a1.east) {$ $};
        \node[wordBNode, anchor=west, minimum width=8mm]  (a3) at (a2.east) {$U_z$};
        \node[wordBNode, anchor=west, minimum width=12mm] (a4) at (a3.east) {$\gamma_z$};
        \node[wordBNode, anchor=west, minimum width=8mm, draw=gray] (a5) at (a4.east) {$\dots$};

        \node[wordBNode, anchor=west, minimum width=5mm, draw=gray] (r0) at ( 0.6, -2.7)    {$ \dots $};
        \node[wordBNode, anchor=west, minimum width=\alphaWidth mm]       (r1) at (r0.east) {$ \alpha  $};
        \node[wordBNode, anchor=west, minimum width=\alphaWidth mm, fill=blue!30]       (r2) at (r1.east) {$ \alpha  $};
        \node[wordBNode, anchor=west, minimum width=\alphaWidth mm, fill=blue!30]       (r3) at (r2.east) {$ \alpha  $};
        \node[wordBNode, anchor=west, minimum width=9mm]       (r4) at (r3.east) {$ \dots  $};
        \node[wordBNode, anchor=west, minimum width=5mm]        (r5) at (r4.east) {$ p_\alpha  $};
        \node[wordBNode, anchor=west, minimum width=18mm]       (r6) at (r5.east) {$ W $};

        \node[wordBNode, anchor=west, minimum width=7mm] (s) at (a2.west) {$s_\alpha$}; 
        \node[wordBNode, anchor=west, minimum width=\alphaWidth mm, fill=blue!30] (i0) at (s.east) {$\alpha$}; 
        \node[wordBNode, anchor=west, minimum width=\alphaWidth mm] (i1) at (i0.east) {$\alpha$}; 

        \node[wordBNode, anchor=west, minimum width=\alphaWidth mm] at ($(a1.west)+(0.5, 0)$) {$\alpha$}; 

        \draw[draw=black, thick] (a1.north west) rectangle (a1.south east);
        \draw[draw=blue, thick] (a2.north west) rectangle (a2.south east);

        \node[scale=0.8, blue] at ($(a2.north)+(0, 0.2)$) {$\edgeLabelingFn(\pi')$};
        \node[scale=0.8] at ($(a1.north)+(0, 0.2)$) {$s$};

        \coordinate (r2_mid) at ($(a0.south)!0.5!(r0.north)$);
        \node[blue] at ($(r2_mid)+(-1.0, 0)$) {$\pi'$};
    \end{tikzpicture}
    \end{center}
    Invoking \cref{lemma:primitiveAlignment}, we have that $\pi'$ is a model, contradicting that
    none of the $\boundAut$-reaching paths can be $\Gamma$-expanded into a~model.

    We are now interested in the path $\rho$ corresponding
    to the prefix $p$ shared by all $\boundAut$-reaching paths in $\shortPathsSet$,
    examining what happens if we deviate from $\rho$.
    Let $\rho$ be the following sequence of vertices
    $\rho = s_0, \dots, s_n$.
    We now prove a series of claims that derive structural properties of $\aut_z$ by inspecting
    the possibility of picking a~path passing through an alternative successor of some vertex $s_j$ for $i \le j \le n$.
    Since we have $\boundAut < \boundFlat$, we have that all flat variables from $\needle$
    have assigned a word that is longer than any of the prefixes $p_z \in \prefixesOf{z}$
    we will consider in the following. Combined with the fact that $\varphi'$ contains a~fresh alphabet symbol
    $\sep$, we conclude that we only have to consider overlaps with $\needle$'s rightmost variable
    $x$.
    Recall, that the lemma statement assumes that the $\needle$eedle of $\varphi$ can
    be written as $\needle = \needle' x \alpha^M p W$ where $p \neq \alpha$ is a prefix of $\alpha$
    and $p a_0 \not \in \prefixesOf{\alpha}$ with $a_0$ being the first letter of $W$.
    In the follow-ing, we write $s_j$ to denote a vertex of $\rho$ for $0 \le j \le n$,
    and  by $s'_{j+1}$ we denote a~vertex adjacent to $s_j$ not taken by $\rho$.
    
    \begin{claim} \label{claim:succCount}
        If $\rho$ contains a vertex $s_j$ with $s'_{j+1}$ such that the words
        $\edgeLabelingFn(s_j, s'_{j+1})$ and $W$ differ at some position, then
        any sequence $\chi = s_0, \dots, s_j, s'_{j+1}, \dots, s'_{m}$ of a
        suitable length is a $\boundAut$-reaching path in $T_z$ with $\chiSigma
        \models \varphi'$.
    \end{claim}

    \begin{proof}
        We first establish that Point~\ref{cond1} of \cref{lemma:conflictKaboom} holds for the
        word $\edgeLabelingFn(\chi)$, implying that none of the vertices in $\chi$ would be marked
        as dead ends.

        We have shown that $\alpha \in \factorsOf{\sigma(\haystack_i)}$, and, thus,
        let $S$ be a~minimal suffix of $\sigma(\haystack_i)$ such that $\alpha \in \factorsOf{S}$.

        First, before deviating from $\rho$, we had $S \concat \edgeLabelingFn(s_0, s_{j+1}) \in \factorsOf{\alpha^+}$.
        Due to the construction of $T_z$, we have that $\edgeLabelingFn(s_j, s'_{j+1})$ starts with a different
        letter than $\edgeLabelingFn(s_j, s_{j+1})$, and, therefore, $S \concat \edgeLabelingFn(s_0, s_{j+1}) \not \in
        \factorsOf{\alpha^+}$. Graphically, if we consider the following in which we $\alpha$ in $S$ perfectly aligned some $\alpha$ from $\sigma(\needle)$ due to \cref{lemma:primitiveAlignment}
        \begin{center}
        \begin{tikzpicture}
            \node[wordBNode, anchor=east, minimum width=8mm, draw=gray] (a0) at (0, 0) {$\dots$};
            \node[wordBNode, anchor=west, minimum width=20mm, thick]    (a1) at (a0.east) {$ $};
            \node[wordBNode, anchor=west, minimum width=40mm]           (a2) at (a1.east) {$\edgeLabelingFn(s_0, \dots, s_j)$};
            \node[wordBNode, anchor=west, minimum width=30mm]           (a3) at (a2.east) {$\textcolor{red}{\edgeLabelingFn(s_j, s'_{j+1})} $};
            \node[wordBNode, anchor=west, minimum width=8mm, draw=gray] (a4) at (a3.east) {$\dots$};

            \node[wordBNode, anchor=west, minimum width=8mm, draw=gray] (b0) at (-0.3, -0.7)    {$ \dots  $};
            \node[wordBNode, anchor=west, minimum width=12mm]      (b1) at (b0.east) {$ \alpha  $};
            \node[wordBNode, anchor=west, minimum width=12mm]      (b2) at (b1.east) {$ \alpha  $};
            \node[wordBNode, anchor=west, minimum width=12mm]      (b3) at (b2.east) {$ \alpha  $};
            \node[wordBNode, anchor=west, minimum width=12mm]      (b4) at (b3.east) {$ \alpha  $};
            \node[wordBNode, anchor=west, minimum width=12mm] (b5) at (b4.east) {$ \textcolor{red}{\alpha}  $};
            \node[wordBNode, anchor=west, minimum width=12mm]      (b6) at (b5.east) {$ \alpha  $};
            \node[wordBNode, anchor=west, minimum width=5mm]       (b7) at (b6.east) {$ p_\alpha  $};
            \node[wordBNode, anchor=west, minimum width=5mm]       (b8) at (b7.east) {$ W  $};

            \node[] at ($(a1.north)+(0, 0.2)$) {$ S $};
            \node[wordBNode, minimum width=12mm, anchor=west, minimum height=5.5mm] at ($(a1.west)+(0.35, 0)$) {$ \alpha $};
        \end{tikzpicture}
        \end{center}
        then there would be a conflict between $\alpha$ and a prefix of $\edgeLabelingFn(s_j, s'_{j+1})$.
        Thus, no $\alpha$ from $\sigma(\needle)$ can overlap with $\edgeLabelingFn(s_j, s'_{m})$,
        and it remains for us to investigate the remaining case when there is no $\alpha$ overlapping
        with $\edgeLabelingFn(s_j, \dots, s'_{m})$. In such a case, we have the following.
        \begin{center}
        \tikzstyle{wordBNode} = [draw, scale=0.8, inner sep=1mm, text height=3mm, minimum height=7mm, baseline]
        \begin{tikzpicture}
            \node[wordBNode, anchor=west, minimum width=50mm] (l_alpha0) at (0, 0)          {$\edgeLabelingFn(s_0, \dots, s_j)$};
            \node[wordBNode, anchor=west, minimum width=30mm] (l_alpha1) at (l_alpha0.east) {$\edgeLabelingFn(s_j, \dots, s'_{m})$};
            \node[wordBNode, anchor=east, minimum width=8mm, gray] (dots) at (l_alpha0.west) {$\dots$};
            \node[wordBNode, anchor=west, minimum width=8mm, gray] (dots) at (l_alpha1.east) {$\dots$};

            \node[wordBNode, anchor=west, minimum width=8mm, gray] (r_alpha0) at (-1.0, -0.7)    {$ \dots  $};
            \node[wordBNode, anchor=west, minimum width=12mm]      (r_alpha1) at (r_alpha0.east) {$ \alpha  $};
            \node[wordBNode, anchor=west, minimum width=12mm]      (r_alpha2) at (r_alpha1.east) {$ \alpha  $};
            \node[wordBNode, anchor=west, minimum width=12mm]      (r_alpha3) at (r_alpha2.east) {$ \alpha  $};
            \node[wordBNode, anchor=west, minimum width=12mm]      (r_alpha4) at (r_alpha3.east) {$ \alpha  $};
            \node[wordBNode, anchor=west, minimum width=5mm]       (r_alpha5) at (r_alpha4.east) {$ p_\alpha  $};
            \node[wordBNode, anchor=west, minimum width=5mm]       (r_alpha6) at (r_alpha5.east) {$ W $};

            \node[] at ($(l_alpha0.north east)+(0, 0.3)$) {$s_j$};
        \end{tikzpicture}
        \end{center}
        But such an alignment contains a conflict between $W$ and the prefix $\edgeLabelingFn(s_j, s'_{j+1})$
        due to the premise of this claim.
        Therefore, we have established that Point~\ref{cond1} of \cref{lemma:conflictKaboom} holds.

        Observing that the above reasoning above $\edgeLabelingFn(s_j, s'_{j+1})$ and $\edgeLabelingFn(s_j, s_{j+1})$
        starting with a different letters allows us to write $S \concat \edgeLabelingFn(s_0, s_j)$ as
        \begin{equation}
            S \concat \edgeLabelingFn(s_0, s_j) = u_0 \alpha v'_0
        \end{equation}
        for $|v'_0| < |\alpha|$ such that $v'_0 \concat \edgeLabelingFn(s_{j}, s'_{j+1}) \not \in \prefixesOf{\alpha}$
        and $\alpha \not \in \prefixesOf{v'_0 \concat \edgeLabelingFn(s_{j}, s'_{j+1})}$, and, therefore, we can apply
        \cref{lemma:conflictKaboom} and conclude that $\chiSigma \models \varphi'$.
    \end{proof}

    \begin{claim} \label{claim:edgeLength}
        If $\rho$ contains a vertex $s_j$ such that $|\edgeLabelingFn(s_j,
        s'_{j+1})| < |W|$, then any $\chi
        = s_0, \dots, s_j, s'_{j+1}, \dots, s'_{m}$ of suitable
        length is a $\boundAut$-reaching path in $T_z$. Moreover,
        $\chiSigma \models \varphi'$.

    \end{claim}
    \begin{proof}
        We have that $s'_{j+1}$ has at least two adjacent vertices, $t_1, t_2 \in V$. As
        $|\edgeLabelingFn(s_j, s'_{j+1})| < |W|$, we have at least one of $\edgeLabelingFn(s_j, s'_{j+1}, t_1)$
        and $\edgeLabelingFn(s_j, s'_{j+1}, t_2)$ disagreeing with $W'$ on some position.
        The rest is the same as in the proof of \cref{claim:succCount}.
    \end{proof}
    
    Next, we apply these claims to the structure of $\aut_z$. The core of our argument
    is based on the fact that none of the $\boundAut$-reaching paths is a model.
    
    \begin{restatable}{claim}{claimTransitionStructure}\label{claim:transitionStructure}
        Any vertex $s_j$ has only two adjacent vertices $s_{j+1}$ and $s'_{j+1}$.
        Moreover:
        \begin{enumerate}
            \item $\edgeLabelingFn(s_j, s_{j+1})$ is a factor of some word $w_\alpha \in \alpha^*$, and
            \item $\edgeLabelingFn(s_j, s'_{j+1}) = W w$ for some word $w \in \alphabet^*$. Moreover,
                $s'_{j+1}$ was marked as a dead end.
        \end{enumerate}
    \end{restatable}
    \begin{proof}

    First, let us establish that every successor $s_j$ has two adjacent vertices.
    For a moment, let $s_j$ be a vertex with at least three successors $s_{j+1}, s'_{j+1}$
    and $s''_{j+1}$, where $s_{j+1}$ is taken by $\rho$.
    If $\edgeLabelingFn(s_{j}, s'_{j+1})$ differs at some position from $W$, then we can apply
    \cref{claim:succCount} and see that there is a $\boundAut$-reaching path $\chi$ with $\chiSigma \models \varphi'$, a contradiction.
    Alternatively, we have that there is no position at which $\edgeLabelingFn(s_{j}, s'_{j+1})$ and $W$ differ.
    But due to the construction of $T_z$, we have that $\edgeLabelingFn(s_{j}, s''_{j+1})$ differs from
    $\edgeLabelingFn(s_{j}, s'_{j+1})$ at the first letter, and, thus, it also differs from $W$ at the first position.
    The rest of the argument is as before.

    Dwelling on the fact that no $\boundAut$-reaching paths are a model, we
    have that for every vertex $s_j$ the edge $(s_j, s'_{j+1})$ not taken
    by $\pi$ cannot be labeled by a word shorter than $W'$, thanks to
    \cref{claim:edgeLength}. Otherwise, we would again reach a contradiction
    with there being no $\boundAut$-reaching paths that would allow us to produce an altered model.

    Combining \cref{claim:succCount} with \cref{claim:edgeLength}, we have that
    every edge $(s_j, s'_{j+1})$ not taken by $\rho$ must be labeled by a~word with the prefix $W$, i.e.,
    $\edgeLabelingFn(s_j, s'_{j+1}) = W \concat w_j$ for some word $w_{j} \in \Sigma^*$. 

    It remains to show that $s'_{j+1}$ is a dead end. We show that if $s'_{j+1}$ is not a dead end,
    then there is a $\boundAut$-reaching path $\chi$ such that $\chiSigma \models \varphi'$,
    reaching a contradiction.
     
    Assuming that $s'_{j+1}$ is not a dead end, we argue that any sequence of vertices $\chi = s_0, \dots, s'_{j+1}, \dots, s'_{m}$
    of a suitable length forms a $\boundAut$-reaching path and $\chiSigma \models \varphi'$.
    First, recall that we have $|\edgeLabelingFn(s_j, s'_{j+1})| > |W|$. Observe that if we consider the following alignment
    \begin{center}
    \tikzstyle{wordBNode} = [draw, scale=0.8, inner sep=1mm, text height=3mm, minimum height=7mm, baseline]
    \begin{tikzpicture}
        \node[wordBNode, anchor=west, minimum width=8mm, gray] (a0) at (0, 0) {$\dots$};
        \node[wordBNode, anchor=west, minimum width=30mm]      (a1) at (a0.east) {$\edgeLabelingFn(s_0 \ldots s_j)$};
        \node[wordBNode, anchor=west, minimum width=22mm]      (a2) at (a1.east) {$\edgeLabelingFn(s_j, s'_{j+1})$};
        \node[wordBNode, anchor=west, minimum width=14mm]      (a3) at (a2.east) {$\cdots$};

        \node[wordBNode, anchor=west, minimum width=8mm, gray] (r0) at (-1.0, -0.7)    {$ \dots $};
        \node[wordBNode, anchor=west, minimum width=12mm]      (r1) at (r0.east) {$ \alpha $};
        \node[wordBNode, anchor=west, minimum width=12mm]      (r2) at (r1.east) {$ \alpha $};
        \node[wordBNode, anchor=west, minimum width=12mm]      (r3) at (r2.east) {$ \alpha $};
        \node[wordBNode, anchor=west, minimum width=5mm]       (r4) at (r3.east) {$ p_\alpha $};
        \node[wordBNode, anchor=west, minimum width=5mm]       (r5) at (r4.east) {$ W $};
    \end{tikzpicture}
    \end{center}
    then there must be a conflict somewhere, otherwise $s'_{j+1}$ would be a dead end.

    The rest of the proof is almost identical to the proof of \cref{claim:edgeLength}, thus, we provide
    just a sketch. First,
    as we have shown above, there is a minimal suffix $S \in \suffixesOf{\sigma(\haystack_i)}$
    such that $\alpha \in \factorsOf{S}$. Therefore, if we consider any
    alignment such that $\sigma(\needle)$ is shifted more to the right,
    we can apply \cref{lemma:primitiveAlignment} and conclude that the $\alpha$ in $S$
    must be perfectly aligned with an $\alpha$ from $\sigma(\needle)$. But then we have
    an $\alpha$ overlapping with a prefix of $\edgeLabelingFn(s_{j}, s'_{j+1})$. Since
    $S \concat \edgeLabelingFn(s_0 \ldots s_{j+1}) \in \factorsOf{\alpha^*}$ but 
    $S \concat \edgeLabelingFn(s_0 \ldots s'_{j+1}) \not \in \factorsOf{\alpha^*}$,
    such an alignment must contain a conflict, and hence, Point~\ref{cond1}
    of \cref{lemma:conflictKaboom} is satisfied.

    Moreover, as we have 
    $S \concat \edgeLabelingFn(s_0 \ldots s_{j+1}) \in \factorsOf{\alpha^*}$ but 
    $S \concat \edgeLabelingFn(s_0 \ldots s'_{j+1}) \not \in \factorsOf{\alpha^*}$,
    we can write
    \begin{equation}
        S \concat \edgeLabelingFn(s_0 \ldots s'_{j+1}) = u_0 \circ \alpha \circ v'_0
    \end{equation}
    for $\alpha \not \in \prefixesOf{v'_0}$, $v'_0 \not \in \prefixesOf{\alpha}$, and some word $u_0 \in \alphabet^*$.
    Since $\edgeLabelingFn(s_0, s'_{j+1})$ is a prefix of $\edgeLabelingFn(s_0, s'_m)$
    we can extend the above to
    \begin{equation}
        S \concat \edgeLabelingFn(s_0 \ldots s'_m) = u_0 \circ \alpha \circ v_0
    \end{equation}
    s.t.\ $v'_0 \in \prefixesOf{v_0}$. Again, we have
    $\alpha \not \in \prefixesOf{v_0}$ and $v_0 \not \in \prefixesOf{\alpha}$, and,
    so, the word $S \concat \edgeLabelingFn(s_0 \ldots s'_m)$ satisfies Point~\ref{cond2}
    of \cref{lemma:conflictKaboom}.
    We invoke \cref{lemma:conflictKaboom} and obtain the claim.
    \end{proof}

    We are currently only a step away from obtaining a flat language $\langflatunderapp_z$
    such that if $\varphi$ is satisfiable with a model $\sigma$, then there is $\sigma' =
    \sigma \triangleq \{ z \mapsto w_z \}$ for some $w_z \in \langflatunderapp_z$ such that $\sigma' \models \varphi'$.
    We know that if our original model satisfies $|\sigma(z)| \ge \boundAut$,
    then either we can produce an altered model by $\Gamma$-expanding a~path from $\boundReachingPaths$,
    or there are no such alternative models and all $\boundReachingPaths$ share the common prefix $p$. 

    Assuming $|\sigma(z)| \ge \boundAut$, there are now two possibilities depending on whether $\sigma(z)$
    corresponds to a path $\pi$ in $T_z$ that eventually takes an edge $(s_i, s'_{i+1})$
    labeled with a word $\edgeLabelingFn(s_i, s'_{i+1})$ such that $W \in \prefixesOf{\edgeLabelingFn(s_i, s'_{i+1})}$.

    If $\pi$ takes no such edge, we have $\edgeLabelingFn(\pi) \in \factorsOf{\alpha^+}$,
    and, thus, $\sigma(z) \in \lAlphaFlat$ where
    \begin{equation}
        \lAlphaFlat \triangleq \factorsOf{\alpha^+} \cap \lang_z
    \end{equation}
    Note that $\factorsOf{\alpha^+}$ is a finite union of flat languages, i.e.,
    \begin{equation}
        \factorsOf{\alpha^+} =
            \bigcup_{\substack{p \in \prefixesOf{\alpha} \\ s \in \suffixesOf{\alpha}}} s \alpha^* p,
    \end{equation}
    and, thus,
    we observe that in the automaton corresponding to $\factorsOf{\alpha^+} \cap \lang_z$ we
    have any two words $v, u$ s.t.\ there is a~state $q$ with $q \move{v} q$
    and $q \move{u} q$ must be a power of the same word due to the intersection with $s \alpha^* p$.
    Applying \cref{lem:nonflatCharacterization}, we have that every language in the union is flat,
    and, hence $\lAlphaFlat$ is flat.
    
    Alternatively, $\pi$ eventually takes an edge $(s_i, s'_{i+1})$ with $W \in
    \prefixesOf{\edgeLabelingFn(s_i, s'_{i+1})}$, and, thus we have
    $\sigma(z) \in \factorsOf{\alpha^+} \concat W \concat V$ for some
    $V \in \Sigma^*$. However, there might be a lot of alternative vertices $s'_{i+1}$ that we can take,
    each labeled with $\edgeLabelingFn(s_{j}, s'_{i+1}) = W w_{i}$
    for some word $w_{i} \in \alphabet^*$. The problem is now that we do not know which $s'_{i+1}$
    is taken by $\sigma(z)$, if any. Luckily, there are only finitely many states that can label
    such vertices, and, hence, we can obtain a finite union of flat languages containing 
    a word that can be used to construct an~alternative model $\sigma' \models \varphi'$.

    Let $\aut_z = (Q, \alphabet, \Delta, I, F)$ be the automaton accepting the language
    associated with the variable $z$. Let $L_{\rightarrow q}$ to be the language of the automaton
    $\aut_{\rightarrow q} = (Q, \alphabet, \Delta, I, \{q\} )$,
    and, symmetrically, $L_{q \rightarrow}$ to be the language of
    the automaton $\aut_{q \rightarrow} = (Q, \alphabet, \Delta, \{q\}, F )$.

    Let us define a map $\pathExit\colon
    C(\aut_z) \times C(\aut_z) \rightarrow \alphabet^*$ allowing us to refer to
    the possible words labeling edges to alternative vertices $s'_{j+1}$ as
    $\pathExit(\stateLabelOf{s}, \stateLabelOf{s'}) \triangleq
    \edgeLabelingFn(\stateLabelOf{s}, \stateLabelOf{s'})$ if we have an edge
    $(s, s') \in E$ that is labeled by a word such that $W \in
    \prefixesOf{\edgeLabelingFn(\stateLabelOf{s}, \stateLabelOf{s'})}$,
    otherwise  
    $\pathExit(\stateLabelOf{s}, \stateLabelOf{s'}) \triangleq \emptyWord$.
    Since $\aut_z$ is a DFA, $\pathExit$ is well-defined,
    and we highlight that $\pathExit$ being does not depend on $\pi$.

    Hence if $\sigma(z) \not \in \factorsOf{\alpha^*}$, then 
    \begin{equation}
        \sigma(z) \in (\factorsOf{\alpha^+} \cap L_{\rightarrow q}) \concat \pathExit(q, r)
        \concat L_{r \rightarrow} \label{eq:zExitLangTmp}
    \end{equation}
    for some $q, r \in Q$ such that $\pathExit(q, r) \neq \epsilon$.
    We show that under these circumstances, we have all prerequisites of \cref{lemma:flatify}
    satisfied. Owing to \cref{eq:zExitLangTmp} we can write
    $\sigma(z) = s \alpha^k p W V$ for $s \in \suffixesOf{\alpha}$, $p \in
    \prefixesOf{\alpha}$, $k \ge 1$ and some word $V$.
    Moreover, we have $\aut_z$
    being composed of a single SCC thanks to the preprocessing we performed, and, $|\sigma(x)| > \boundFlat$
    for any $x \in \varsOf{\needle}$
    due to fixing the values of variables with words shorter than $\boundAut$ in the previous section.
    Hence, we have prerequisites of \cref{lemma:flatify} satisfied.
    Applying \cref{lemma:flatify}, we conclude that
    we can find a word $w_z \in \lGammaFlat$ such that
    $\sigma \variant \{z \mapsto w_z \} \models \varphi$
    with $\lGammaFlat$ being defined
    as
    \begin{equation}
        \lGammaFlat = \bigcup_{\substack{q, r \in Q \\
        \pathExit(q, r) \neq \emptyWord }}
        (\factorsOf{\alpha^*} \cap L_{\rightarrow q}) \concat \pathExit(q, r) \concat U(r) \concat \gamma^{\boundGamma}_z
    \end{equation}
    where $U(r)$ is the word $U$ from $\gammaPref{z}^{\boundGamma}$-expansion of any prefix ending in the state $r \in Q$,
    and $\gamma_z$ is the primitive word from the definition of $\gammaPref{z}^{\boundGamma}$-expansion.
    We observe that $\lGammaFlat$ is a flat language, since
    $\factorsOf{\alpha^+} \cap L_{\rightarrow q}$
    is flat, and the union is finite.

    Combining all these cases, we obtain the lemma, where the flat language $L^{\mathrm{Pref}}_z$
    stated by the lemma is
    \begin{equation} \label{eq:flatifiedLangDef}
        \langflatunderapp_z \triangleq \flatFinLang \cup \gammaPref{z}^{\boundGamma} (\edgeLabelingFn(\boundReachingPaths))
            \cup \lAlphaFlat \cup \lGammaFlat.
    \end{equation}
    where we lifted the definition of $\edgeLabelingFn$ and $\gammaPref{z}^{\boundGamma}$ to sets of paths
    and words.
    We note that both $\lAlphaFlat$ and $\lGammaFlat$ do not depend on the path $\pi$,
    and, consequently, both $\lAlphaFlat$ and $\lGammaFlat$ (soundly) overapproximate corresponding
    languages that would be otherwise obtained by using definitions that inspect~$\pi$.
\end{proof}

\fi

\end{document}